\documentclass[USenglish]{llncs}

\newcommand{\ExtendedVersion}[1]{#1} \newcommand{\PaperVersion}[1]{}

\usepackage{algorithm}
\usepackage{algorithmic} 
\usepackage{amsmath}
\usepackage{amssymb}
\usepackage[USenglish]{babel}
\usepackage{color}
\usepackage[colorlinks=true,linkcolor=black,citecolor=black]{hyperref}
\usepackage{graphicx}
\usepackage[latin1]{inputenc}
\usepackage{times}

\hypersetup{bookmarksopen=true}

\definecolor{darkgray}{rgb}{.3,.3,.3}

\newcommand{\FinalVersion}[1]{#1}\newcommand{\DraftVersion}[1]{}

\newcounter{todo}

\FinalVersion{\newcommand{\todo}[1]{}}%
\DraftVersion{\newcommand{\todo}[1]{\stepcounter{todo}\noindent{\color{red} \textbf{Todo~{\arabic{todo}}: #1}}}}

\FinalVersion{}
\DraftVersion{}

\FinalVersion{\newcommand{\removable}[1]{#1}}
\DraftVersion{\newcommand{\removable}[1]{{\color{darkgray}#1}}}

\newcommand{\hidden}[1]{}

\newcommand{\definedAs}    
	{=}

\renewcommand{\ldots}{\,...\,}

\newcommand{\fctsymDom}{\mathrm{dom}} 
\newcommand{\fctDom}[1]{\fctsymDom(#1)}

\newcommand{\true}{\mathrm{true}} 
\newcommand{\false}{\mathrm{false}} 

\newcommand{\tuple}[1]{\langle #1 \rangle} 
\newcommand{\bigtuple}[1]{\big\langle #1 \big\rangle}

\newcommand{\symAllURIs}{\mathcal{U}} 
\newcommand{\symAllLiterals}{\mathcal{L}} 
\newcommand{\symAllBNodes}{\mathcal{B}} 
\newcommand{\symAllTriples}{\mathcal{T}} 
\newcommand{\symAllVariables}{\mathcal{V}} 
\newcommand{\symURI}{u} 

\newcommand{\symRDFgraph}{G} 
\newcommand{\symRDFdataset}{\mathfrak{D}} 
\newcommand{\symRDFdatasetNGs}{\mathcal{N}} 

\newcommand{\fctsymURIs}{\mathrm{uris}} 
\newcommand{\fctURIs}[1]{\fctsymURIs(#1)}

\newcommand{\fctsymSBVars}{\mathrm{sbvars}} 
\newcommand{\fctSBVars}[1]{\fctsymSBVars(#1)}
\newcommand{\symAllPatterns}{\mathcal{P}} 
\newcommand{\symTP}{tp} 
\newcommand{\symBGP}{B} 
\newcommand{\symPattern}{P} 
\newcommand{\OpAND}{\text{ \normalfont\scriptsize\textsf{AND} }}
\newcommand{\OpUNION}{\text{ \normalfont\scriptsize\textsf{UNION} }}
\newcommand{\OpOPT}{\text{ \normalfont\scriptsize\textsf{OPT} }}
\newcommand{\OpFILTER}{\text{ \normalfont\scriptsize\textsf{FILTER} }}
\newcommand{\OpGRAPH}{\text{\normalfont\scriptsize\textsf{GRAPH} }} 

\newcommand{\muEmpty}{\mu_\emptyset}

\newcommand{\fctEvalPGD}[3]{[\![#1]\!]_{#2}^{#3}} 

\newcommand{\symAllDocs}{\mathcal{D}} 
\newcommand{\fctsymData}{\mathrm{data}}
\newcommand{\fctData}[1]{\fctsymData(#1)}

\newcommand{\symWoD}{W} 
\newcommand{\symDocs}{D} 
\newcommand{\fctsymADoc}{adoc} 
\newcommand{\fctADoc}[1]{\fctsymADoc(#1)}

\newcommand{\symWoDTuple}{\symWoD = \tuple{\symDocs,\fctsymADoc}}

\newcommand{\symWoDEx}{\symWoD_\mathsf{ex}}
\newcommand{\symDocsEx}{\symDocs_\mathsf{ex}}
\newcommand{\fctsymADocEx}{\fctsymADoc_\mathsf{ex}}
\newcommand{\fctADocEx}[1]{\fctsymADocEx(#1)}
\newcommand{\symWoDExTuple}{\symWoDEx = \tuple{\symDocsEx,\fctsymADocEx}}

\newcommand{\symReachCrit}{c} 
\newcommand{\cAll}{\symReachCrit_\mathsf{All}}
\newcommand{\cNone}{\symReachCrit_\mathsf{None}}
\newcommand{\cMatch}{\symReachCrit_\mathsf{Match}}

\newcommand{\symSeedURIs}{S} 

\newcommand{\fctEvalReachPcSW}[4]{[\![#1]\!]^{\mathsf{R}(#2,#3)}_{#4}}
\newcommand{\fctEvalSetQcW}[3]{[\![#1]\!]^{#2}_{#3}}

\newcommand{\symLP}{\mathit{lp}}
\newcommand{\symLPE}
	{\mathit{lpe}}
\newcommand{\peEmpty}{\varepsilon}
\newcommand{\peWildcard}
	{\,\underbar{\phantom{o}}\,}
\newcommand{\peContextURI}{+}
\newcommand{\peSubquery}[2]{\tuple{#1,#2}}
\newcommand{\peConcat}[2]{#1 / #2}
\newcommand{\peAlt}[2]{#1 | #2}
\newcommand{\peKleene}[1]{#1^*}
\newcommand{\peTest}[1]{[#1]}

\newcommand{\OpSEED}{\text{\normalfont\scriptsize\textsf{SEED} }}

\newcommand{\symCtxURI}{\symURI_\mathsf{ctx}}
\newcommand{\symCtxDoc}{d_\mathsf{ctx}}

\newcommand{\fctEvalPathPdW}[3]{[\![#1]\!]^{#2}_{#3}}

\newcommand{\symQuery}{q}

\newcommand{\fctEvalQW}[2]{[\![#1]\!]_{#2}}
\newcommand{\fctEvalQWS}[3]{[\![#1]\!]_{#2}^{#3}}

\newcommand{\enumA}{(i)}
\newcommand{\enumB}{(ii)}
\newcommand{\enumC}{(iii)}
\newcommand{\enumD}{(iv)}

\newcommand{\myquote}[1]
	{\textsl{``#1''}}

\newenvironment{mydef}[1]
	{\begin{definition}\normalfont}{\end{definition}}
\newcommand{\definedTerm}[1]
	{\textbf{#1}}

\newenvironment{myexample}{\begin{example}}{\end{example}}

\newtheorem{factTheorem}{Fact}
\newenvironment{fact}{\begin{factTheorem}\hspace{-1.5mm}\textbf{.}\normalfont}{\end{factTheorem}}

\newcommand{\wold}{Web of Linked Data}
\newcommand{\wolds}{Webs of Linked Data}
\newcommand{\doc}{document}
\newcommand{\docs}{{\doc}s}
\newcommand{\triple}{RDF triple}

\newcommand{\query}{LDQL query}
\newcommand{\queries}{LDQL queries}

\newcommand{\trans}{\operatorname{trans}}
\newcommand{\dataset}{\operatorname{dataset}}

\newcommand{\symNE}{\textit{ne}} 
\newcommand{\symPP}{\textit{pe}} 
\newcommand{\OpASK}{\text{\normalfont\scriptsize\textsf{ASK} }}
\newcommand{\OpBIND}{\text{\normalfont\scriptsize\textsf{BIND}}}
\newcommand{\OpAS}{\text{\normalfont\scriptsize\textsf{ AS} }}

\begin{document}

\title{LDQL: A Query Language for the Web of Linked Data}

\ExtendedVersion{%
\subtitle{\vspace{-1mm}(Extended Version)%
\vspace{-6mm}\footnote{This document is an extended version of a paper
	published in ISWC~2015~\cite{ProceedingsVersion}.%
}}
}

\author{Olaf Hartig\inst{1} \and Jorge P\'{e}rez\inst{2}}

\authorrunning{O. Hartig \and J. P\'{e}rez}   

\institute{
		\url{http://olafhartig.de/}
		\\[2mm] 
	\and
		Department of Computer Science, Universidad de Chile\\
		\email{jperez@dcc.uchile.cl}
}

\maketitle

\vspace{-7mm} 

\begin{abstract}

 The Web of Linked Data is composed of tons of RDF documents interlinked to each other forming a huge
 repository of distributed semantic data.
 Effectively querying this distributed data source is an important open problem in the Semantic Web area.
 In this paper, we propose LDQL, a declarative language to
 query Linked Data on the Web. One of the novelties of LDQL is that it
 expresses separately \enumA~patterns that describe the expected query result, and \enumB~%
 Web navigation paths that select the data sources to be used for computing the result.
 We present a formal syntax and semantics, prove equivalence rules, and study the
 expressiveness of the language.
 In particular, we show that LDQL is strictly more expressive than the query formalisms 
 that have been proposed previously for Linked Data on the Web.
 The high expressiveness allows LDQL to define queries 
 for which a complete execution is not computationally feasible over the Web.
 We formally study this issue and provide a syntactic sufficient condition to avoid this problem; 
 queries satisfying this condition are ensured to have a procedure
 to be effectively evaluated over the Web of Linked Data.


\end{abstract}

\vspace{-7mm} 

\vspace{-3mm}

\PaperVersion{\enlargethispage{\baselineskip}} 
\ExtendedVersion{\enlargethispage{2\baselineskip}} 

\section{Introduction}
\label{sec:Introduction}
\vspace*{-8pt}

In recent years an increasing amount of
	\removable{structured}
data has been published and interlinked on the World Wide Web~(WWW) in adherence to the Linked Data principles~\cite{BernersLee06:LinkedData}.
	These principles are based on standard Web technologies. In particular, \enumA~the Hypertext Transfer Protocol~(HTTP) is used to access data, \enumB~HTTP-based Uniform Resource Identifiers~(URIs) are used as identifiers for entities described in the data, and \enumC~the Resource Description Framework~(RDF) is used as data model.
Then, any HTTP URI in an RDF triple presents a \emph{data link} that enables software clients to retrieve more data by looking up the URI with an HTTP request.
The adoption of these principles has lead to the creation of a globally distributed dataspace%
	:
the \emph{Web of Linked~Data}.


The emergence of
	the Web of Linked Data
makes possible an \emph{online execution} of declarative queries over up-to-date data from a virtually unbounded set of data sources, each of which is readily accessible without any need for implementing source-spe\-cif\-ic APIs or wrappers.
This possibility has spawned research interest in approaches to query
	Linked Data on the WWW
as if it was a single~(distributed) database. For an overview on query execution techniques proposed in this context refer to~\cite{Hartig13:Survey}.

The main contribution of this paper is the proposal of LDQL, a novel query language for the Web of Linked Data.
The most important feature of LDQL is that it clearly separates query components for selecting que\-ry-rel\-e\-vant regions of the Web of Linked Data, 
from components for specifying the query result that has to be constructed from the data in the selected regions. 
The most basic construction in LDQL are tuples of the form~$\tuple{L,Q}$ where $L$ is an 
expression used to select a set of relevant \docs, and $Q$
is a query intended to be executed over the data in these \docs\ as if they were a single RDF repository.
In an abstract setting one can use several formalisms to express $L$ and $Q$.
In our proposal, for the former part we introduce the notion of \emph{link path expressions} 
that are a form of nested regular expressions~(with some other important features) 
used to navigate the link graph of the Web.
For the latter, we use standard SPARQL graph patterns.
To begin evaluating these queries one needs to specify a set of seed URIs.
The language also possesses features to dynamically~(at query time)
	identify
new
seed URIs to evaluate portions of a query.
Additionally, such queries can be combined by using conjunctions, disjunctions, and projection.
We present a formal syntax and semantics for LDQL, propose some
	rewrite
rules, and study
its expressive~power.


While there does not exist a standard language for expressing queries over Linked Data on the WWW,
a few options have been proposed. In particular, a first strand of research focuses on extending the scope 
of SPARQL such that an evaluation of SPARQL queries over Linked Data has a well-defined semantics~\cite{Harth12:CompletenessClassesForLDQueries,Hartig12:TheoryPaper,Hartig15:WebPPsPaper,Umbrich12:LinkedDataQueriesWithReasoning:Article}. 
A second strand of~research focuses on navigational languages~\cite{Fionda15:NautiLODArticle,Hartig15:WebPPsPaper}.
Although these languages have different motivations, 
a commonality of all~these proposals is that, in contrast to LDQL, 
the definition of que\-ry-rel\-e\-vant regions of the Web of Linked Data and the definition 
of que\-ry-rel\-e\-vant data
within the specified regions are mixed.

As our second main contribution we compare LDQL with three previously proposed formalisms
for querying the Web of Linked Data: \emph{SPARQL under reach\-abil\-i\-ty-based
query semantics}~\cite{Hartig12:TheoryPaper}, \emph{NautiLOD}~\cite{Fionda15:NautiLODArticle}, 
and \emph{SPARQL Property Path patterns under context-based semantics}~\cite{Hartig15:WebPPsPaper}. 
We formally prove that LDQL is strictly more expressive than every one of these.
We show that for every query $Q$ in the previous languages, one can effectively construct an LDQL query which
is equivalent to $Q$. Moreover, for every one of the previous languages, there exists an LDQL query
that cannot be expressed in that language. 
These results show that LDQL presents an interesting expressive power.

The downside of the expressiveness provided by LDQL
is the existence of queries for which a complete execution
	is not feasible in practice.  
To capture this issue formally, we define a notion of \emph{Web-safe\-ness} for LDQL queries. Then, the obvious question that arises is how to identify LDQL queries that are Web-safe.
Our last technical contribution is the identification of 
a sufficient syntactic condition for Web-safe\-ness.

The rest of the paper is structured as follows. Section~\ref{sec:DataModel} introduces a
data model that provides the basis for defining the semantics of LDQL. 
In Section~\ref{sec:Language} we formally define the syntax and semantics of
	LDQL and show some simple algebraic properties.
	In Section~\ref{sec:Comparisons} we compare LDQL with the three mentioned languages,
	 and in Section~\ref{sec:Safeness} we focus on Web-safe\-ness. 
	 Section~\ref{sec:Conclusion} concludes the paper and sketches
future work.
Proofs of the formal results in this paper can be found in
	\ExtendedVersion{the Appendix.}%
	\PaperVersion{an extended version of the~paper~\cite{ExtendedVersion}.}
	
A preliminary version of some of the results in this
	paper have been presented in a workshop~\cite{AMWVersion}.
This paper is a substantial extension of~\cite{AMWVersion} refining the
definition of LDQL and introducing
important changes to the syntax and the semantics of the language.
Moreover, the comparison with previous proposals was not discussed in~\cite{AMWVersion}.%
	~ 
\vspace*{-8pt}

\section{Data Model}
\label{sec:DataModel}
\vspace*{-8pt}

In this section we introduce a structural data model that captures the concept of a Web of Linked Data formally.
	As usual~\cite{Fionda15:NautiLODArticle,Harth12:CompletenessClassesForLDQueries,Hartig12:TheoryPaper,Hartig15:WebPPsPaper,Umbrich12:LinkedDataQueriesWithReasoning:Article}, for the definitions and analysis in this paper,
we assume
	that the Web is fixed during the execution of any single query.

	We
use the RDF data model~\cite{Cyganiak14:RDFConcepts} as a basis for our
	model of a Web of Linked Data. That is, we
assume three pairwise disjoint,
infinite sets $\symAllURIs$~(URIs), $\symAllBNodes$~(blank nodes), and $\symAllLiterals$~(literals).
	An
\emph{RDF triple} is a tuple $\tuple{s,p,o} \in \symAllTriples$ with $\symAllTriples \definedAs (\symAllURIs \cup \symAllBNodes) \times \symAllURIs \times (\symAllURIs \cup \symAllBNodes \cup \symAllLiterals)$.
For any RDF triple $t = \tuple{s,p,o}$ we write $\fctURIs{t}$ to denote the set of all URIs in
	$t$%
%
.

Additionally, we assume another infinite set $\symAllDocs$ that is disjoint from $\symAllURIs$, $\symAllBNodes$, and $\symAllLiterals$, respectively. We refer to elements in this set as
	\emph{\docs}
and use them to represent the concept of Web documents from which Linked Data can be extracted. Hence, we assume a function, say $\fctsymData$, that maps each \doc\ $d \in \symAllDocs$ to a finite set of RDF triples
	$\fctData{d} \subseteq \symAllTriples$%
%
	%
	~such
that the data of each \doc\ uses a unique set of blank~nodes%
.

Given these preliminaries, we are ready to define a \emph{\wold}. 

\begin{mydef}{\wold}%
	\label{def:WoLD}%
	A \definedTerm{\wold} is a tuple $\symWoDTuple$
		that consists of a set of \docs\ $\symDocs \subseteq \mathcal{D}$ and a partial function \,$\fctsymADoc \!: \symAllURIs \rightarrow \symDocs$\, that is surjective%
		.
%
\end{mydef}

	Function $\fctsymADoc$ of a \wold\ $\symWoDTuple$
captures the relationship between the URIs that can be looked up in this Web and the \docs\ that can be retrieved by such lookups. Since not every URI can be looked up, the function is partial. For any URI $\symURI \in \symAllURIs$ with $\symURI \in \fctDom{\fctsymADoc}$~(i.e., any URI that can be looked up in~$\symWoD$), \doc\ $d = \fctADoc{\symURI}$ can be considered the authoritative source of data for~$\symURI$ in~$\symWoD$~(hence, the name $\fctsymADoc$). To accommodate for \docs\ that are authoritative for multiple URIs, we do not require injectivity for function $\fctsymADoc$. However, we require surjectivity because we conceive \docs\ as irrelevant for a \wold\ if they cannot be retrieved by any URI lookup in this Web.

Let $\symWoDTuple$ be a \wold. 
$\symWoD$ is said to be finite~\cite{Hartig12:TheoryPaper} 
if its set $\symDocs$ of \docs\ is finite.
In this paper we assume that every \wold\ is finite.
%
Given \docs\ $d,d' \in \symDocs$ and a triple $t\in \fctData{d}$, 
we say that a URI $\symURI \in \fctURIs{t}$ establishes a \emph{data~link} 
from $d$ to $d'$\!, if $\fctADoc{\symURI} = d'$\!.
As a final concept, we formalize the notion of a \emph{link graph} associated to $W$\!.
This graph has \docs\ in $\symDocs$ as nodes, and directed edges representing data links between \docs. 
Each edge is associated with a label that identifies both
the particular \triple\ and the URI in this triple that establishes the corresponding data link.
These labels shall provide the basis for defining the navigational component of our 
query language. 

\begin{mydef}{Link Graph}%
	\label{def:LinkGraph}%
	The \definedTerm{link graph} of a \wold\ $\symWoDTuple$,
	is a directed, edge-la\-beled multigraph, $\mathcal{G}_W=\tuple{\symDocs,E_W}$,
	with set of edges $E_W \!\subseteq \symDocs \times (\symAllTriples\! \times \symAllURIs) \times \symDocs$ defined as 
$		E_W \definedAs
		\big\lbrace \tuple{d_\mathsf{src},(t,\symURI),d_\mathsf{tgt}} \mid
			t \in \fctData{d_\mathsf{src}}, \symURI \in \fctURIs{t} \text{ and } d_\mathsf{tgt}=\fctADoc{\symURI}
		\big\rbrace
		.
$
\end{mydef}

%
%
%

For a link graph edge $e=\tuple{d_\mathsf{src},(t,\symURI),d_\mathsf{tgt}}$, tuple $(t,u)$ is the label of $e$.
Moreover, we sometimes write $e\in \mathcal{G}_W$ to denote that $e$ is an edge in the link graph $\mathcal{G}_W$.

\begin{myexample}%
	\label{ex:WoLD}%
	As a running example for this paper assume a simple \wold\ $\symWoDExTuple$
		with
	three \docs, $d_\mathsf{A}$, $d_\mathsf{B}$, and $d_\mathsf{C}$~(i.e., $\symDocsEx = \lbrace d_\mathsf{A},d_\mathsf{B},d_\mathsf{C} \rbrace$). The
		data in these \docs\ are the following sets of {\triple}s:
	\begin{align*}
		\fctData{d_\mathsf{A}} = \lbrace &\tuple{\symURI_\mathsf{A},p_1,\symURI_\mathsf{B}},
		&
		\fctData{d_\mathsf{B}} = \lbrace &\tuple{\symURI_\mathsf{B},p_1,\symURI_\mathsf{C}} \rbrace;
		\\
		&\tuple{\symURI_\mathsf{B},p_2,\symURI_\mathsf{C}} \rbrace;
		&
		\fctData{d_\mathsf{C}} = \lbrace &\tuple{\symURI_\mathsf{A},p_2,\symURI_\mathsf{C}} \rbrace;
	\end{align*}
	and for function $\fctsymADocEx$ we have:
	$\fctADocEx{\symURI_\mathsf{A}} \!=\! d_\mathsf{A}$,
	$\fctADocEx{\symURI_\mathsf{B}} \!=\! d_\mathsf{B}$,
	$\fctADocEx{\symURI_\mathsf{C}} \!=\! d_\mathsf{C}$, and
	$\fctADocEx{p_1} = d_\mathsf{A}$~(i.e., $\fctDom{\fctsymADocEx} \!=\! \lbrace \symURI_\mathsf{A},\symURI_\mathsf{B},\symURI_\mathsf{C},p_1 \rbrace$).
		This Web contains 10 data links. For instance, URI $\symURI_\mathsf{A}$ in the \triple\ $\tuple{\symURI_\mathsf{A},p_2,\symURI_\mathsf{C}} \in \fctData{d_\mathsf{C}}$ establishes a data link to \doc\ $d_\mathsf{A}$. Hence, the corresponding edge in the link graph of $\symWoDEx$ is $\bigtuple{d_\mathsf{C},(\tuple{\symURI_\mathsf{A},p_2,\symURI_\mathsf{C}},\symURI_\mathsf{A}),d_\mathsf{A}}$.
	Fig\removable{ure}~\ref{fig:LinkGraph} illustrates
		the link graph $\mathcal{G}_{\symWoDEx}$ with all 10~edges.
\end{myexample}

\begin{figure}[tb]
	\centering
	\includegraphics[width=0.7\textwidth]{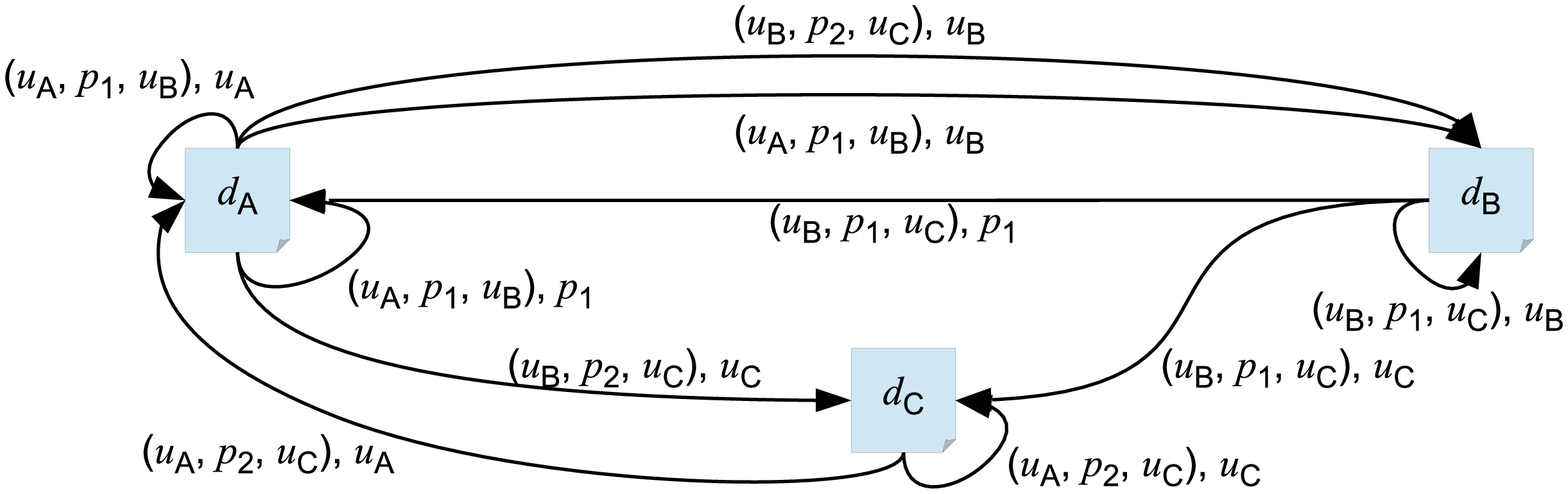}
	\vspace*{-5pt}
	\caption{The link graph $\mathcal{G}_{\symWoDEx}$ of our example \wold\ $\symWoDEx$.}
	\label{fig:LinkGraph}
	\vspace*{-10pt}
	
\end{figure}

%

\vspace*{-10pt}

\section{Definition of LDQL}
\label{sec:Language}
\vspace*{-5pt}

This section defines our Linked Data query language, LDQL. LDQL queries are meant to be evaluated over a \wold\ and each such query is built from two types of components: \emph{Link path expressions}~(\emph{LPEs}) for selecting \removable{que\-ry-rel\-e\-vant} \docs\ of the queried \wold; and SPARQL graph patterns for specifying the query 
result that has to be constructed from the data in the selected \docs.
For this paper, we assume \removable{that} the reader is familiar with the definition of SPARQL~\cite{Harris13:SPARQL1_1Language}, including the algebraic formalization introduced in~\cite{Perez09:SemanticsAndComplexityOfSPARQL,Arenas09:SemanticsAndComplexityOfSPARQLBookChapter}.
In particular, for SPARQL graph patterns we closely follow the formalization 
in~\cite{Arenas09:SemanticsAndComplexityOfSPARQLBookChapter} considering operators$\OpAND$\!,$\OpOPT$\!,$\OpUNION$\!,$\OpFILTER$\!, and $\OpGRAPH$\!, plus the operator $\OpBIND$ defined in~\cite{Harris13:SPARQL1_1Language}.
%
%
We begin this section by introducing the most basic concept of our language, 
the notion of link patterns. We use link patterns as the basis for navigating the link graph of a \wold.

\subsection{Link Patterns} \label{ssec:LinkPatterns}

%

A link pattern is a tuple in $\bigl( \symAllURIs \cup \lbrace \peWildcard , \peContextURI \rbrace \bigr) \times \bigl( \symAllURIs \cup \lbrace \peWildcard , \peContextURI \rbrace \bigr) \times \bigl( \symAllURIs \cup \symAllLiterals \cup \lbrace \peWildcard , \peContextURI \rbrace \bigr)$. 
Link patterns are used to match link graph edges in the context of a designated \emph{context} URI. 
The special symbol~$\peContextURI$ denotes a placeholder for the {context} URI.
The special symbol~$\peWildcard$ denotes a wildcard that will drive the direction of the navigation.
Before formalizing how link graph edges actually match link patterns, we show some intuition.
Consider the link graph of Web $\symWoDEx$ in Example~\ref{ex:WoLD}~(see Fig.~\ref{fig:LinkGraph}), 
and the link pattern $\tuple{\peContextURI,p_1,\peWildcard}$.
Intuitively, in the context of URI $\symURI_\mathsf{A}$, the edge
with label $(\tuple{\symURI_\mathsf{A},p_1,\symURI_\mathsf{B}},\symURI_\mathsf{B})$ from document $d_\mathsf{A}$
to document $d_\mathsf{B}$, matches the link pattern $\tuple{\peContextURI,p_1,\peWildcard}$. 
Notice that in the matching, the context URI $\symURI_\mathsf{A}$ 
takes the place of symbol $\peContextURI$, and $\symURI_\mathsf{B}$ takes the place of the wildcard symbol $\peWildcard$.
Notice that $\symURI_\mathsf{B}$ also denotes the direction of the edge that matches the link pattern.
On the other hand, the edge with label 
$(\tuple{\symURI_\mathsf{A},p_1,\symURI_\mathsf{B}},\symURI_\mathsf{A})$ from $d_\mathsf{A}$ to $d_\mathsf{A}$,
does not match $\tuple{\peContextURI,p_1,\peWildcard}$;
although $\symURI_\mathsf{B}$ can take the place of the wildcard symbol $\peWildcard$, the direction of the
edge is not to $\symURI_\mathsf{B}$.
That is, when matching an edge labeled by $(t,\symURI)$ we require URI $\symURI$ to be taking the place of a wildcard in the
link pattern.
When more than one wildcard symbol is used, the link pattern can be matched by 
edges pointing to the direction of any of the URIs taking the place of a wildcard.
For instance, in the context of $\symURI_\mathsf{A}$, 
the link pattern $\tuple{\peWildcard,p_2,\peWildcard}$ is matched by edges
	$\tuple{ d_\mathsf{A}, (\tuple{\symURI_\mathsf{B},p_2,\symURI_\mathsf{C}},\symURI_\mathsf{B}), d_\mathsf{B} }$ and $\tuple{ d_\mathsf{A}, (\tuple{\symURI_\mathsf{B},p_2,\symURI_\mathsf{C}},\symURI_\mathsf{C}), d_\mathsf{C} }$.
The next definition formalizes this
	notion of matching.

\begin{mydef}{Link Pattern}
A link graph edge with label $(\tuple{x_1,x_2,x_3},\symURI)$ \definedTerm{matches} a link pattern 
$\tuple{y_1,y_2,y_3}$ in the context of a URI $\symCtxURI$ 
if the following two properties hold:
\begin{enumerate}
	\item there exists $i\in \{1,2,3\}$ such that $y_i = \peWildcard$ and $x_i = \symURI$, and
	\item for every $i\in \{1,2,3\}$ either $y_i = \peContextURI$ and $x_i = \symCtxURI$, \ or $y_i = x_i$, \ or  $y_i = \peWildcard$.
\end{enumerate}
\end{mydef}


One of the rationales for \removable{adopting}
the notion of a context URI and the $\peContextURI$ symbol in our definition of link patterns, is to support cases in which link graph navigation has to be focused solely on data links that are \emph{authoritative}.
A data link represented by link graph edge $\tuple{d_\mathsf{src},(t,\symURI),d_\mathsf{tgt}} \in \mathcal{G}_\symWoD$ is authoritative
in a \wold\ $\symWoDTuple$
if $d_\mathsf{src}=\fctADoc{\symURI'}$ for some URI $\symURI'\in \fctURIs{t}$.
Thus, if we fix a context URI $\symCtxURI$, 
a link pattern that uses the $\peContextURI$ symbol allows us to follow only authoritative 
data links from \doc~$\symCtxDoc = \fctADoc{\symCtxURI}$.


\subsection{LDQL Queries}

The most basic construction in LDQL queries are tuples of the from $\tuple{L,\symPattern}$ where $L$ is an 
expression used to select a set of \docs\ from the Web of Linked Data, and $\symPattern$
is a SPARQL graph pattern to query these \docs\ as if they were a single RDF dataset. 
In an abstract setting, one can use any formalism to specify $L$ 
as long as $L$ defines sets of RDF documents.
In our proposal we use what we call \emph{link path expressions}~(LPEs) that are 
a form of nested regular expressions~\cite{PerezAG10} 
over the alphabet of link patterns.
Every link path expression begins its navigation in a context URI, traverses the Web, and returns a set of URIs; these URIs are used to construct an RDF dataset
with all the \docs\
	to be retrieved by looking up the URIs.
This dataset is passed to the SPARQL graph pattern
to obtain the final evaluation of the whole query.
Besides the basic constructions of the form $\tuple{L,P}$, in LDQL one can also use$\OpAND$\!,$\OpUNION$\! and projection, to combine
them. We also introduce an operator $\OpSEED$ that is used to dynamically change, at query time, 
the seed URI from which the navigation begins.
The next definition formalizes the syntax of LDQL queries and LPEs.

\begin{mydef}{LDQL Query}%
	\label{def:Syntax}
	The syntax of LDQL is given by the following production rules
	in which
	$\symLP$ is an arbitrary link pattern,
	$?v$ is a variable,
	$\symPattern$ is a SPARQL graph pattern~(as per~\cite{Arenas09:SemanticsAndComplexityOfSPARQLBookChapter}),
	$V$ is a finite set of variables,
	and $U$ is a finite set of URIs:
	\begin{align*}
		\symQuery \, :\definedAs \, \, &
			          \tuple{\symLPE,\symPattern}
			\,\mid\, (\OpSEED\ U\ \symQuery )
			\,\mid\, (\OpSEED\ ?v\ \symQuery )
			\,\mid\,  (\symQuery \OpAND \symQuery )
			\,\mid\, (\symQuery \OpUNION \symQuery )
			\,\mid\, \pi_V \symQuery
\\
		\symLPE \, :\definedAs \, \, &
			          \peEmpty
			\,\mid\,  \symLP
			\,\mid\,  \peConcat{\symLPE}{\symLPE}
			\,\mid\,  \peAlt{\symLPE}{\symLPE}
			\,\mid\,  \peKleene{\symLPE}
			\,\mid\,  \peTest{\symLPE}
			\,\mid\,  \peSubquery{?v}{\symQuery}
	\end{align*}
	Any expression that satisfies the production $\symQuery$ is an \definedTerm{\query},
	any expression that satisfies the production $\symLPE$ is a \definedTerm{link path expression} (\definedTerm{LPE}),
	and any \query\ of the form $\tuple{\symLPE,\symPattern}$ is a \definedTerm{basic \query}.
\end{mydef}


Before going into the formal semantics of LDQL and LPEs, we give some more intuition about how these
expressions are evaluated in a Web of Linked Data $\symWoD$\!.
As mentioned before, the most basic expression in LDQL is of the form $\tuple{\symLPE,\symPattern}$.
To evaluate this expression over $\symWoD$ we will need a set $\symSeedURIs$ of \emph{seed} URIs.
When evaluating $\tuple{\symLPE,\symPattern}$, 
every one of the seed URIs in $\symSeedURIs$ will trigger a navigation of
	link graph $\mathcal{G}_\symWoD$
via the link path expression 
$\symLPE$ starting on that seed. That is, the seed URIs are passed to $\symLPE$ as \emph{context} URIs
in which the LPE should be evaluated. These evaluations of $\symLPE$
will result in a set of URIs that are used to construct a dataset over which $\symPattern$
is finally evaluated. 

Regarding the navigation of
	link graph $\mathcal{G}_\symWoD$\!,
the most basic form of navigation is to follow a single link graph edge that matches
a link pattern $\symLP$.
When a navigation via a link pattern $\symLP$ is triggered from a context URI $\symURI$, we proceed as follows.
We first go to the authoritative \doc\ for $\symURI$, that is $\fctADoc{\symURI}$, and try to
	find outgoing link graph edges that match $\symLP$
in the context of $\symURI$~(as explained in
	Section~\ref{ssec:LinkPatterns}).
Every one of these matches defines a new context URI $\symURI'$ from which the navigation
	can~continue.
More complex forms of navigation are obtained by combining link patterns via classical regular expression
operators such as concatenation~$\peConcat{}{}$, disjunction~$\peAlt{}{}$, 
and recursive concatenation~$\peKleene{(\cdot)}$.
The nesting operator~$\peTest{\cdot}$ is used to test for existence of paths.
When a context URI $\symURI$ is passed to an expression $\peTest{\symLPE}$, it checks
	whether $\mathcal{G}_\symWoD$ contains a path from $\symCtxDoc = \fctADoc{\symURI}$ that matches $\symLPE$.
If such a path exists, the navigation can continue from 
the same context URI $\symURI$.
The most involved form of navigation is by using the expression $\tuple{?v,\symQuery}$ with $\symQuery$ an \query.
To evaluate this expression from context URI~$\symURI$ one first has to pass $\symURI$ as a seed URI for $\symQuery$ and recursively evaluate $\symQuery$ from that seed.
This evaluation generates a set of \removable{solution} mappings, and for every one of these mappings its value on variable
$?v$ is used as the new context URI from which the navigation continues.
Finally, note that 
our notion of LPEs does not provide an operator for navigating paths in their inverse direction. 
The reason for omitting such an operator is that traversing arbitrary data links backwards is impossible on the~WWW.

To formally define the semantics of LDQL we need to introduce some terminology. 
We first define a function $\dataset_\symWoD(\cdot)$ that from a set of URIs constructs an RDF dataset
with all the \docs\ pointed to by those URIs in $\symWoD$\!.
Formally, given a \wold\ $\symWoDTuple$ and a set $U$ of URIs,
$\dataset_\symWoD(U)$ is an RDF dataset~(as per \cite{Harris13:SPARQL1_1Language,Arenas09:SemanticsAndComplexityOfSPARQLBookChapter}) 
that has the set of triples
	$\{ t\in \fctData{\fctADoc{\symURI}} \mid \symURI\in U\cap\fctDom{\fctsymADoc}\}$
as default graph.
Moreover, for every URI $\symURI \in U\cap\fctDom{\fctsymADoc}$, $\dataset_\symWoD(U)$ contains 
the named graph $\tuple{\symURI,\fctData{\fctADoc{\symURI}}}$.

\begin{myexample} \label{ex:Dataset}
	Consider the Web $\symWoDEx$ in Example~\ref{ex:WoLD} and 
	the set of URIs $U= \lbrace \symURI_\mathsf{A}, \symURI_\mathsf{C} \rbrace$.
	Then $\dataset_{\symWoDEx}(U)$ has 
	$\lbrace \tuple{\symURI_\mathsf{A},p_1,\symURI_\mathsf{B}}, \tuple{\symURI_\mathsf{B},p_2,\symURI_\mathsf{C}}, \tuple{\symURI_\mathsf{A},p_2,\symURI_\mathsf{C}} \rbrace$
	as default graph, 
	and two named graphs, $\tuple{\symURI_\mathsf{A},\{\tuple{\symURI_\mathsf{A},p_1,\symURI_\mathsf{B}}, \tuple{\symURI_\mathsf{B},p_2,\symURI_\mathsf{C}}\}}$ and $\tuple{\symURI_\mathsf{C},\{\tuple{\symURI_\mathsf{A},p_2,\symURI_\mathsf{C}}\}}$.
\end{myexample}

In the formalization of the semantics of LDQL,
we use the standard join operator $\Join$ over sets of solution mappings~\cite{Harris13:SPARQL1_1Language,Perez09:SemanticsAndComplexityOfSPARQL}.
We also make use of the semantics of SPARQL graph patterns over datasets as defined
in~\cite{Arenas09:SemanticsAndComplexityOfSPARQLBookChapter}. 
In particular, given an RDF dataset $\symRDFdataset$, an RDF graph $\symRDFgraph$ in
$\symRDFdataset$, and a SPARQL graph pattern $\symPattern$, we denote by $\fctEvalPGD{\symPattern}{\symRDFgraph}{\symRDFdataset}$ the evaluation of $\symPattern$ over $\symRDFgraph$ in $\symRDFdataset$~\cite[Definition~13.3]{Arenas09:SemanticsAndComplexityOfSPARQLBookChapter}.

We are now ready to formally
	define
the semantics of LDQL and LPEs.
Given a~\wold\ $\symWoD$ and a set $\symSeedURIs$ of URIs, we formalize the evaluation of \queries\ over $\symWoD$
from the seed URIs $\symSeedURIs$, as a function 
$\fctEvalQWS{  \cdot  }{\symWoD}{\symSeedURIs}$ that given an \query, produces a set
of solution mappings.
Similarly, the evaluation of LPEs over $\symWoD$ from a context URI $\symURI$, is formalized as
a function $\fctEvalPathPdW{  \cdot  }{\symURI}{\symWoD}$ that given an LPE, produces a set of URIs.

\begin{mydef}{Evaluation Semantics of LDQL Queries and LPEs}%
	\label{def:LDQLSemantics}
	\label{def:SemanticsLPEs}
		Given a finite set $\symSeedURIs \!\subseteq\! \symAllURIs$, the \definedTerm{$\symSeedURIs$-based evaluation} of \queries~over a \wold\ $\symWoD \!=\! \tuple{\symDocs,\fctsymADoc}$,
	denoted by $\fctEvalQWS{\cdot}{\symWoD}{\symSeedURIs}$, is
	defined recursively as~follows:

\begin{align*}
	\fctEvalQWS{  \tuple{\symLPE,\symPattern}  }{\symWoD}{\symSeedURIs} &\definedAs
		\fctEvalPGD{\symPattern}{\symRDFgraph}{\symRDFdataset}
    \;\;\text{ where $\textstyle\symRDFdataset = \dataset_\symWoD\left(
    \bigcup_{u\in S}\fctEvalPathPdW{\symLPE}{u}{\symWoD}\right)$ with default graph $G$}
	,
	\\
	\fctEvalQWS{  (\OpSEED\ U\ \symQuery)  }{\symWoD}{\symSeedURIs} &\definedAs
		\fctEvalQWS{  \symQuery  }{\symWoD}{U}
	,
	\\[-0.5mm] 
	\fctEvalQWS{  (\OpSEED\ ?v\ \symQuery)  }{\symWoD}{\symSeedURIs} &\definedAs
		\textstyle \bigcup_{\symURI \in \symAllURIs} \bigl( \fctEvalQWS{\symQuery}{\symWoD}{\lbrace\symURI\rbrace} \Join \lbrace \mu_\symURI \rbrace \bigr) 
		\;\;\text{ where $\mu_\symURI = \lbrace ?v \mapsto \symURI \rbrace$ for all $\symURI \in \symAllURIs$}
	,
	\\
	\fctEvalQWS{  (\symQuery_1\! \OpUNION \symQuery_2)  }{\symWoD}{\symSeedURIs} &\definedAs
		\fctEvalQWS{\symQuery_1}{\symWoD}{\symSeedURIs} \cup \fctEvalQWS{\symQuery_2}{\symWoD}{\symSeedURIs},
		\hspace{11mm}
	\\
	\fctEvalQWS{  ( \symQuery_1\! \OpAND \symQuery_2)  }{\symWoD}{\symSeedURIs} &\definedAs
		\fctEvalQWS{\symQuery_1}{\symWoD}{\symSeedURIs} \!\Join \fctEvalQWS{\symQuery_2}{\symWoD}{\symSeedURIs},
		\hspace{4mm}
	\\
	\fctEvalQWS{  \,\pi_V \symQuery\,  }{\symWoD}{\symSeedURIs} &\definedAs
		\lbrace \mu \mid \text{there exists } \mu'\! \in \fctEvalQWS{\symQuery}{\symWoD}{\symSeedURIs} 
		\text{ such that $\mu$ and $\mu'$ are} 
		\\[-0.7mm]
		& \hspace{11mm} \text{ compatible and } \fctDom{\mu} = \fctDom{\mu'} \cap V
 		\rbrace
	.
\end{align*}
Now for the semantics of LPEs, given a context URI $\symCtxURI\in \fctDom{\fctsymADoc}$, the 
\definedTerm{$\symCtxURI$-based evaluation} of LPEs over~$\symWoD$\!, denoted by $\fctEvalPathPdW{\cdot}{\symCtxURI}{\symWoD}$\!, is defined recursively as follows:
\begin{align*}
	\fctEvalPathPdW{\,\peEmpty\,}{\symCtxURI}{\symWoD} &\definedAs
		\lbrace \symCtxURI \rbrace
		,
	\\
	\fctEvalPathPdW{\symLP}{\symCtxURI}{\symWoD} &\definedAs
		\lbrace \symURI \in \symAllURIs \mid 
		\text{there exist a link graph edge $\tuple{d_\mathsf{src},(t,\symURI),d_\mathsf{tgt}} \in \mathcal{G}_W$, with}
		\\[-0.7mm]
		& \hspace{18.5mm}\text{$d_\mathsf{src}=\fctADoc{\symCtxURI}$, that matches $\symLP$ in the context of $\symCtxURI\}$}
		,
		%
		%
	\\[-0.7mm] 
	\fctEvalPathPdW{\peConcat{\symLPE_1}{\symLPE_2}}{\symCtxURI}{\symWoD} &\definedAs
		\lbrace \symURI \in \fctEvalPathPdW{\symLPE_2}{\symURI'}{\symWoD} \mid \symURI' \in \fctEvalPathPdW{\symLPE_1}{\symCtxURI}{\symWoD} \rbrace
		,
	\\
	\fctEvalPathPdW{\peAlt{\symLPE_1}{\symLPE_2}}{\symCtxURI}{\symWoD} &\definedAs
		\fctEvalPathPdW{\symLPE_1}{\symCtxURI}{\symWoD} \cup \fctEvalPathPdW{\symLPE_2}{\symCtxURI}{\symWoD}
		,
	\\
	\fctEvalPathPdW{\peKleene{\symLPE}}{\symCtxURI}{\symWoD} &\definedAs
		\lbrace \symCtxURI \rbrace \cup \fctEvalPathPdW{\symLPE}{\symCtxURI}{\symWoD} \cup \fctEvalPathPdW{\peConcat{\symLPE}{\symLPE}}{\symCtxURI}{\symWoD} \cup \fctEvalPathPdW{\peConcat{\peConcat{\symLPE}{\symLPE}}{\symLPE}}{\symCtxURI}{\symWoD} \cup \ldots
		,
	\\
	\fctEvalPathPdW{\,\peTest{\symLPE}\,}{\symCtxURI}{\symWoD} &\definedAs
		\lbrace \symCtxURI \mid \fctEvalPathPdW{\symLPE}{\symCtxURI}{\symWoD} \neq \emptyset \rbrace
		,
	\\[-0.7mm] 
	\fctEvalPathPdW{\,\peSubquery{?v}{\symQuery}\,}{\symCtxURI}{\symWoD} &\definedAs
		\lbrace \symURI\in \symAllURIs \mid \text{there exists }\mu\in\fctEvalQWS{\symQuery}{\symWoD}{\lbrace \symCtxURI \rbrace}
		\text{ such that }\mu(?v)=\symURI  \rbrace 
	.
\end{align*}
Moreover, if $\symCtxURI\notin \fctDom{\fctsymADoc}$, then $\fctEvalPathPdW{\symLPE}{\symCtxURI}{\symWoD}=\emptyset$ for every LPE.
\end{mydef}

\begin{myexample} \label{ex:LPE}
	Let $\symLPE_\mathsf{ex}$ be the LPE \,$\peConcat{ \peKleene{\tuple{\peWildcard,p_1,\peWildcard}} }{ \peTest{\tuple{\peWildcard,p_2,\peWildcard}} }$. This LPE selects \docs\ that can be reached via arbitrarily long paths of data links with predicate $p_1$ and, additionally, have some outgoing data link with predicate $p_2$.
	For our example Web $\symWoDEx$ and context URI $\symURI_\mathsf{A}$,
		the LPE selects \docs\ $d_\mathsf{A} = \fctADocEx{\symURI_\mathsf{A}}$ and $d_\mathsf{C} = \fctADocEx{\symURI_\mathsf{C}}$. More precisely, we have~$\fctEvalPathPdW{\symLPE_\mathsf{ex}}{\symURI_\mathsf{A}}{\symWoDEx} = \lbrace \symURI_\mathsf{A}, \symURI_\mathsf{C} \rbrace$. Note that \doc\ $d_\mathsf{B}$ can also be reached via a $p_1$--path, but it does not pass the $p_2$--re\-lat\-ed test.
\end{myexample}

\begin{myexample} \label{ex:BasicQueryEvaluation}
	Consider a set of URIs $\symSeedURIs_\mathsf{ex} = \lbrace \symURI_\mathsf{A} \rbrace$ and a basic \query\ $
	\tuple{ \symLPE_\mathsf{ex}, \symBGP_\mathsf{ex} }$ whose LPE is
		$\symLPE_\mathsf{ex}$ as introduced in Example~\ref{ex:LPE}
	and whose SPARQL graph pattern is a basic graph pattern that contains two triple patterns, $\symBGP_\mathsf{ex} = \lbrace \tuple{?x,p_1,?y}, \tuple{?x,p_2,?z} \rbrace$.
	Given that we have $\fctEvalPathPdW{\symLPE_\mathsf{ex}}{\symURI_\mathsf{A}}{\symWoDEx} = \lbrace \symURI_\mathsf{A}, \symURI_\mathsf{C} \rbrace$~(cf.~Example~\ref{ex:LPE}), 
	$\dataset_{\symWoDEx}(\fctEvalPathPdW{\symLPE_\mathsf{ex}}{\symURI_\mathsf{A}}{\symWoDEx})$ has the default graph
	$\lbrace \tuple{\symURI_\mathsf{A},p_1,\symURI_\mathsf{B}}, \tuple{\symURI_\mathsf{B},p_2,\symURI_\mathsf{C}}, \tuple{\symURI_\mathsf{A},p_2,\symURI_\mathsf{C}} \rbrace$~(cf.~Example~\ref{ex:Dataset}).
%
%
	Then, according to the query semantics, 
	the result of query $\tuple{ \symLPE_\mathsf{ex}, \symBGP_\mathsf{ex} }$ over $\symWoDEx$ using seeds $\symSeedURIs_\mathsf{ex}$ consists of a single solution mapping,~%
		\removable{namely}
	$\mu = \lbrace ?x \mapsto \symURI_\mathsf{A}, ?y \mapsto \symURI_\mathsf{B}, ?z \mapsto \symURI_\mathsf{C} \rbrace$.
\end{myexample}

\begin{myexample} \label{ex:LDQLSemantics}
   Consider an \query\ $\symQuery_\mathsf{ex} = \left( \OpSEED\ ?x\ \bigtuple{\peEmpty, \tuple{?x,p_1,?w}} \right)$ whose subquery is a basic \query\ that has a single triple pattern as its SPARQL graph pattern. Additionally, let $\symQuery_\mathsf{ex}' = \bigtuple{\symLPE_\mathsf{ex},\lbrace \tuple{?x,p_1,?y}, \tuple{?x,p_2,?z} \rbrace}$ be the basic \query\ introduced in Example~\ref{ex:BasicQueryEvaluation}, and let $\symQuery_\mathsf{ex}''$ be the conjunction of these two queries; i.e., $\symQuery_\mathsf{ex}'' = ( \symQuery_\mathsf{ex} \OpAND \symQuery_\mathsf{ex}' )$.
	By Example~\ref{ex:BasicQueryEvaluation} we know that $\fctEvalQWS{\symQuery_\mathsf{ex}'}{\symWoDEx}{\symSeedURIs_\mathsf{ex}} = \lbrace \mu \rbrace$
	with $\mu=\lbrace ?x \mapsto \symURI_\mathsf{A},$ $?y \mapsto \symURI_\mathsf{B}, ?z \mapsto \symURI_\mathsf{C} \rbrace$.
%
	Furthermore, based on the data given in Example~\ref{ex:WoLD}, it is easy to see that $\fctEvalQWS{\symQuery_\mathsf{ex}}{\symWoDEx}{\symSeedURIs_\mathsf{ex}} = \lbrace \mu_1, \mu_2 \rbrace$ with $\mu_1 = \lbrace ?x \mapsto \symURI_\mathsf{A}, ?w \mapsto \symURI_\mathsf{B} \rbrace$ and $\mu_2 = \lbrace ?x \mapsto \symURI_\mathsf{B},$ $?w \mapsto \symURI_\mathsf{C} \rbrace$.
	For the
		$\symSeedURIs_\mathsf{ex}$-based
	evaluation of $\symQuery_\mathsf{ex}''$ over $\symWoDEx$, the result sets $\fctEvalQWS{\symQuery_\mathsf{ex}}{\symWoDEx}{\symSeedURIs_\mathsf{ex}}$ and $\fctEvalQWS{\symQuery_\mathsf{ex}'}{\symWoDEx}{\symSeedURIs_\mathsf{ex}}$ have to be joined. 
	Thus, we need to compute $\lbrace \mu_1, \mu_2 \rbrace \Join \lbrace \mu \rbrace$,
	which results in a single mapping $\mu'=\mu_1\cup\mu=\lbrace ?x \mapsto \symURI_\mathsf{A}, ?w \mapsto \symURI_\mathsf{C}, ?y \mapsto \symURI_\mathsf{B}, ?z \mapsto \symURI_\mathsf{C} \rbrace$.		
\end{myexample}


\subsection{Algebraic Properties of LDQL Queries} \label{sec:Equivalences}


As a basis for the discussion in the
	next
sections, we
	show  some simple algebraic~properties. 
	We say that LDQL queries $q$ and $q'$ are semantically equivalent, denoted by~$q \!\equiv\! q'$\!, if
		$\fctEvalQWS{q}{\symWoD}{\symSeedURIs}=\fctEvalQWS{q'}{\symWoD}{\symSeedURIs}$ holds for every \wold\ $\symWoD$ and every finite set $S \subseteq \symAllURIs$.

\begin{lemma}%
	\label{lem:Equivalences:AssocAndCommut}
	The operators$\OpAND$\! and$\OpUNION$\! are associative and commutative. 
\end{lemma}

\begin{lemma}%
	\label{lem:Equivalences:Distributiveness}
	Let $\symQuery_1$, $\symQuery_2$, $\symQuery_3$ be \queries, the following semantic equivalences hold:
	\begin{align}
		( \symQuery_1 \OpAND (\symQuery_2 \OpUNION \symQuery_3) ) &\equiv ( ( \symQuery_1 \OpAND \symQuery_2 ) \OpUNION ( \symQuery_1 \OpAND \symQuery_3 ) )
		\\
		\pi_V (\symQuery_1 \OpUNION \symQuery_2) &\equiv ( \pi_V \symQuery_1 \OpUNION \pi_V \symQuery_2 )
		\\
		( \OpSEED\ U\ (\symQuery_1 \OpUNION \symQuery_2) ) &\equiv ( ( \OpSEED\ U\ \symQuery_1 ) \OpUNION ( \OpSEED\ U\ \symQuery_2 ) )
		\\
		\left( \OpSEED\ ?v\ (\symQuery_1 \OpUNION \symQuery_2) \right) &\equiv ( ( \OpSEED\ ?v\ \symQuery_1 ) \OpUNION ( \OpSEED\ ?v\ \symQuery_2 ) )
	\end{align}
\end{lemma}

Lemma~\ref{lem:Equivalences:AssocAndCommut} allows us to write sequences of either\OpAND\! or\OpUNION\! without parentheses.
%
Our next result shows the power of the construction $\tuple{?v,q}$. In particular,
it shows the somehow surprising finding that link patterns $\symLP$,
concatenation $/$, disjunction $|$, and the
test $[\cdot]$, are just \emph{syntactic sugar} as they
can be simulated by using $\varepsilon$, $\peSubquery{?v}{\symQuery}$ and $(\cdot)^*$\!.

\begin{proposition}
\label{prop:syntacticSugar}
For every LDQL query $\symQuery$, there exists an LDQL query $\symQuery'$ s.t.\ $\symQuery \equiv \symQuery'$ and 
every LPE in $\symQuery'$
	consists only of
the symbol $\varepsilon$, the construction $\peSubquery{?v}{\symQuery}$, and operator~$(\cdot)^*$\!.
\end{proposition}
\begin{proof}[Sketch]
The proof is based on a recursive translation of link path expressions beginning with link patterns.
For instance, a link pattern of the form $\tuple{\peContextURI,p,\peWildcard}$ is encoded
by $\peSubquery{?v}{ \tuple{\varepsilon, (\OpGRAPH ?u\ (?u,p,?v))} }$, and we can similarly encode all types of link patterns.
To encode $/$ we make use of $\peSubquery{?v}{\symQuery}$ and the operator$\OpAND$\! inside $\symQuery$ as follows.
Consider an LPE $r=r_1/r_2$. It can be shown that $r$ is equivalent to $\peSubquery{?v}{\symQuery}$ where~$\symQuery$~is:
\vspace*{-5pt}

{\small
\[
	\bigl( \
		\tuple{ r_1, (\OpGRAPH ?x\ \{\ \}) }
		\, \OpAND \,
		\bigl( \OpSEED ?x\ \tuple{r_2, (\OpGRAPH ?v\ \{\ \})} \bigr)
	\ \bigr)
	.
\]}
Similarly, to encode $|$ we make use of$\OpUNION$\! and to encode $[\cdot]$ we use projection.
\end{proof}

Although not strictly necessary, we decided to keep link patterns and operators $/$, $|$, and $[\cdot]$ since they
represent a natural and intuitive way of
	expressing navigation paths.


\vspace*{-10pt}

\section{Comparison with Previous Linked Data Query Formalisms}
\label{sec:Comparisons}
\vspace*{-8pt}

In this section, we compare LDQL with alternative formalisms
to query Linked Data on the WWW.
There are
some general query languages for the WWW~(proposed before the advent of
Linked Data) 
that are related to our
proposal; in particular, WebSQL~\cite{WebSQL}, which is similar in spirit to
	LDQL
but different
in the features that the languages posses.
Two main novelties of
	LDQL
compared with WebSQL are the possibility to dynamically select seed URIs at query time, and the traversal of links according to properties of the queried documents that can be defined in the same \query. Neither of these are expressible in WebSQL. 
	While a complete formal comparison between LDQL and WebSQL is certainly very interesting, we leave it for future work and, instead, focus on three more recent proposals of query formalisms for the Web of Linked Data~\cite{Fionda15:NautiLODArticle,Hartig12:TheoryPaper,Hartig15:WebPPsPaper}.
We formally show that LDQL is strictly more expressive than every one of
	them. 


\vspace*{-15pt}

\subsection{Comparison with Property Paths under Context-Based Query Semantics}
\vspace*{-10pt}

Property paths (PPs for short) were introduced in SPARQL 1.1 as a way of adding 
navigational power to the language~\cite{Harris13:SPARQL1_1Language}. PPs are a form of regular expressions that are evaluated
over a single (local) RDF graph; a PP expression is used to retrieve pairs $\tuple{a,b}$ 
of nodes in the graph such that there is a path from $a$ to $b$ whose sequence of edge 
labels belongs~(as a string) to the regular language defined by the expression.
The syntax of PP expressions is given by the following 
grammar\footnote{In~\cite{Hartig15:WebPPsPaper} the reverse path construction $\text{\textasciicircum}\symPP$ is
also considered. We do not consider it here as the form of navigation of these reverse paths
does not represent a traversal of the link graph.},
where $p,u_1,u_2,\ldots,u_k$ are URIs.
\begin{align*}
	\symPP \, :\definedAs \, \, &
	p
  \,\mid\,\ 
  !(u_1|u_2|\cdots|u_k)
	\,\mid\,
		\peConcat{\symPP}{\symPP}
	\,\mid\,
		\peAlt{\symPP}{\symPP}
	\,\mid\,
		\peKleene{\symPP}
\end{align*}%
A PP-pat\-tern is defined as a tuple of the form $\tuple{\alpha,\symPP,\beta}$
where
	$\symPP$ is a PP expression, and $\alpha$ and $\beta$ are in $\symAllURIs \cup \symAllLiterals \cup \symAllVariables$.

In~\cite{Hartig15:WebPPsPaper} the authors adapted the semantics of PP-pat\-terns so that they can be used to query the Web of Linked Data. The proposed query semantics is called \emph{con\-text-based semantics}~\cite{Hartig15:WebPPsPaper}.
To define this semantics, the authors first introduce the notion of a \emph{context selector} for a \wold\ $\symWoD$%
	\!. This context selector is
a function $C^\symWoD\!(\cdot)$ that given a URI $\symURI \in \fctDom{\fctsymADoc}$ returns the RDF triples in 
$\fctData{\fctADoc{u}}$ that have $\symURI$ in the subject position.
Formally, for
	every URI
$\symURI\in \fctDom{\fctsymADoc}$ we have $C^\symWoD\!(\symURI)=\{\tuple{s,p,o}\in \fctData{\fctADoc{\symURI}}\mid s=\symURI\}$.
To simplify the exposition, the authors extended the definition of $C^\symWoD\!(\cdot)$ to also handle
	URIs not in $\fctDom{\fctsymADoc}$, and literals and blank nodes.
For any such RDF term $a$ they define $C^\symWoD\!(a)$ as the empty set.

The con\-text-based semantics for
	PPs
over the Web of Linked Data in~\cite{Hartig15:WebPPsPaper} is a bag semantics that follows closely the semantics for
	PPs
defined in the normative semantics of SPARQL~1.1~\cite{Harris13:SPARQL1_1Language}. Hence, both semantics use a procedure, the \emph{ArbitraryLengthPath} procedure~\cite{Harris13:SPARQL1_1Language}, to define the semantics of the $(\cdot)^*$ operator.
It was shown in~\cite{ArenasCP12} that for sets semantics, the 
normative semantics of
	PPs
can be defined by
	using
standard techniques for regular expressions.
To make the comparison with
	LDQL,
in this paper we adapt the
	con\-text-based
semantics for
	PPs
presented in~\cite{Hartig15:WebPPsPaper} by following the techniques in~\cite{ArenasCP12}, and consider only sets of mappings.
	To this end, we define a function $\fctEvalSetQcW{\cdot}{\mathsf{ctxt}}{\symWoD}$, that given a PP-pat\-tern, returns its evaluation under con\-text-based semantics
over the \wold\ $\symWoD$\!.
In the definition, for a solution mapping $\mu$ and an RDF term $\alpha$, we use
	\removable{the notation}
$\mu[\alpha]$ with the following
meaning: $\mu[\alpha]=\mu(\alpha)$ if $\alpha\in \fctDom{\mu}$, and $\mu[\alpha]=\alpha$ in the other case.
Similarly, $\mu[\tuple{s,p,o}]=\tuple{\mu[s],\mu[p],\mu[o]}$.

{\small
\begin{align*}
\fctEvalSetQcW{(\alpha,p,\beta)}{\mathsf{ctxt}}{\symWoD} & \definedAs 
	\{\mu\mid \fctDom{\mu}=\{\alpha,\beta\}\cap \mathcal{V}
	\text{ and }\mu[\tuple{\alpha,p,\beta}]\in C^{\symWoD}(\mu[\alpha])\} \\
\fctEvalSetQcW{(\alpha,!(u_1|\cdots|u_k),\beta)}{\mathsf{ctxt}}{\symWoD} & \definedAs 
	\{\mu\mid \fctDom{\mu}=\{\alpha,\beta\}\cap \mathcal{V}
	\text{ and exists }p\text{ s.t. }\\
	&\hspace*{30pt} \mu[\tuple{\alpha,p,\beta}]\in C^{\symWoD}(\mu[\alpha])
	\text{ and }p\notin\{u_1,\ldots,u_k\}\}	\\
\fctEvalSetQcW{(\alpha,\symPP_1/\symPP_2,\beta)}{\mathsf{ctxt}}{\symWoD} & \definedAs 
	\pi_{\{\alpha,\beta\}\cap\mathcal{V} }\big( \fctEvalSetQcW{(\alpha,\symPP_1,?v)}{\mathsf{ctxt}}{\symWoD} \Join 
	\fctEvalSetQcW{(?v,\symPP_2,\beta)}{\mathsf{ctxt}}{\symWoD}\big) \\
\fctEvalSetQcW{(\alpha,\symPP_1|\symPP_2,\beta)}{\mathsf{ctxt}}{\symWoD} & \definedAs 
	\fctEvalSetQcW{(\alpha,\symPP_1,\beta)}{\mathsf{ctxt}}{\symWoD} \cup
	\fctEvalSetQcW{(\alpha,\symPP_2,\beta)}{\mathsf{ctxt}}{\symWoD} \\
\fctEvalSetQcW{(\alpha,\symPP^*,\beta)}{\mathsf{ctxt}}{\symWoD} & \definedAs 
	\{\mu\mid \fctDom{\mu}=\{\alpha,\beta\}\cap \mathcal{V}, \mu[\alpha]=\mu[\beta]\text{ and }\mu[\alpha]\in \textit{terms}(\symWoD)\}\cup \\
	&\hspace*{15pt}\fctEvalSetQcW{(\alpha,\symPP,\beta)}{\mathsf{ctxt}}{\symWoD}\cup 
	\fctEvalSetQcW{(\alpha,\symPP/\symPP,\beta)}{\mathsf{ctxt}}{\symWoD}\cup
	\fctEvalSetQcW{(\alpha,\symPP/\symPP/\symPP,\beta)}{\mathsf{ctxt}}{\symWoD}\cup \cdots
\end{align*}}
\vspace*{-10pt}

A \emph{PP-based SPARQL query}~\cite{Hartig15:WebPPsPaper} is an expression formed by combining PP-pat\-terns 
using the standard SPARQL operators$\OpAND$\!,$\OpUNION$\!,$\OpOPT$\!,$\OpFILTER$\! and so on,
following the standard semantics for these operators~\cite{Arenas09:SemanticsAndComplexityOfSPARQLBookChapter}.
Our next results show that LDQL is strictly more expressive than PP-based SPARQL queries 
under con\-text-based semantics.

\begin{theorem}\label{theo:LDQLMoreThanPP}
There exists an LDQL query that cannot be expressed as a PP-based SPARQL query under con\-text-based semantics.
\end{theorem}

\begin{proof}[Sketch]
One can show that LDQL query $\symQuery \!=\! \big(\OpSEED\ U\ \big\langle \tuple{\peContextURI,p,\peWildcard}, (?x,?x,?x) \big \rangle\big)$ with $U = \lbrace \symURI \rbrace$ cannot be expressed by PPs under con\-text-based~seman\-tics because this semantics is ``blind'' to triples that are not authoritative.
For instance, in a Web $\symWoD = \tuple{ \lbrace d,d' \rbrace, \fctsymADoc }$ with
	$\fctData{d}=\{\tuple{\symURI,p,\symURI'}\}$, $\fctData{d'}=\{\tuple{\symURI'\!,p,\symURI},\tuple{\symURI,\symURI,\symURI}\}$, $\fctADoc{\symURI}=d$ and $\fctADoc{\symURI'}=d'$\!,
the evaluation of $\symQuery$ is the solution mapping
$\lbrace ?x \mapsto \symURI \rbrace$.
Notice that the only authoritative triple in $d'$ is $\tuple{\symURI'\!,p,\symURI}$ as
	$d'=\fctADoc{\symURI'}\neq \fctADoc{\symURI}$.
Hence, one can prove that PP-based SPARQL queries under con\-text-based semantics cannot access triple $\tuple{\symURI,\symURI,\symURI}$ in
$d'$\!, and thus, will never have $\lbrace ?x \mapsto \symURI \rbrace$ as~solution.
\end{proof}

\begin{theorem}\label{theo:LDQLCoversPP}
Let $\alpha,\beta \in \symAllURIs \cup \symAllLiterals \cup \symAllVariables$.
Then, for every PP-pat\-tern $\tuple{\alpha,\symPP,\beta}$, there exists an \query\ $\symQuery$ such that
$
\fctEvalSetQcW{ \tuple{\alpha,\symPP,\beta} }{\mathsf{ctxt}}{\symWoD} = \fctEvalQWS{\symQuery}{\symWoD}{\emptyset}
$
for every \wold\ $\symWoD$\!.
\end{theorem}

\begin{proof}[Sketch]
In the proof we provide a translation scheme from PPs to LDQL. 
One major complication is that PPs can retrieve literals and, in general, values that
are not in $\fctDom{\fctsymADoc}$, which are difficult to handle by LPEs.
For every PP-pat\-tern $\tuple{?x,\symPP,?y}$ we construct an \query\ $Q_{\symPP}(?x,?y)$.
For example, for $\tuple{?x,\symPP_1/\symPP_2,?y}$, our query is $\pi_{\lbrace ?x,?y \rbrace}\bigl( Q_{\symPP_1}(?x,?z)\OpAND Q_{\symPP_2}(?z,?y) \bigr)$,
and for $\tuple{?x,!(\symURI_1|\cdots|\symURI_k),?y}$ the translation is
$
\big( \OpSEED ?x\ \big\langle \peEmpty, 
\big( (?x,?p,?y) \OpFILTER (?p\neq \symURI_1 \land \cdots \land \,?p\neq \symURI_k)\big) \big\rangle\big).
$
To handle $\symPP^*$ we need to use the construction $\peSubquery{?v}{\symQuery}$ of LPEs, plus $\peKleene{(\cdot)}$\!.
\end{proof}

\vspace*{-10pt}

\subsection{Comparison with NautiLOD}
\vspace*{-3pt}

NautiLOD is a
	\removable{navigation}
language to traverse Linked Data on the
	WWW and to
perform actions~(such as sending emails) during the traversal~\cite{Fionda15:NautiLODArticle}.
	We
compare LDQL with NautiLOD without action rules.
The syntax of
	NautiLOD \removable{expressions~(without actions)}
is given by the following grammar~(where $p\in \symAllURIs$  and $\symPattern$ is a SPARQL graph~pattern).
\begin{align*}
	\symNE \, :\definedAs \, \, &
		p
  \,\mid\,
  p\text{\textasciicircum}
  \,\mid\,
  \tuple{\peWildcard}
	\,\mid\,
		\peConcat{\symNE}{\symNE}
	\,\mid\,
		\peAlt{\symNE}{\symNE}
	\,\mid\,
		\peKleene{\symNE}
	\,\mid\,
		\symNE\peTest{(\OpASK \symPattern)}
\end{align*}
%
%
In terms of our data model%
\footnote{%
	In~\cite{Fionda15:NautiLODArticle}, all URIs have an assigned set of RDF triples~(which may be
		empty%
	).
	In our data model one can have URIs not in $\fctDom{\fctsymADoc}$. Hence, to properly capture the semantics of NautiLOD in terms of our data model we have to introduce conditions of the form ``$\symURI' \in \fctDom{\fctsymADoc}$.''
}\!%
, the semantics of NautiLOD expressions over a
	\wold\
$\symWoD \!=\! \tuple{\symDocs,\fctsymADoc}$ from URI $\symURI \!\in\! \fctDom{\fctsymADoc}$ is
defined recursively as~follows.
\begin{align*}
	\fctEvalSetQcW{\, p \,}{\symURI}{\symWoD} & \definedAs
	\{\symURI '\mid \tuple{\symURI,p,\symURI'} \in \fctData{\fctADoc{\symURI}}\}
  \\
	\fctEvalSetQcW{\, p\text{\textasciicircum} \,}{\symURI}{\symWoD} & \definedAs
	\{\symURI' \mid \tuple{\symURI',p,\symURI} \in \fctData{\fctADoc{\symURI}}\}
	\\
	\fctEvalSetQcW{\, \tuple{\peWildcard} \,}{\symURI}{\symWoD} & \definedAs
	\{\symURI' \mid \tuple{\symURI,p,\symURI'} \in \fctData{\fctADoc{\symURI}}\text{ for some }p\in \symAllURIs\}
	\\[-1mm]
	\fctEvalSetQcW{\, \symNE_1/\symNE_2 \,}{\symURI}{\symWoD} & \definedAs
	\{ \symURI''\! \mid \symURI''\!\in \fctEvalSetQcW{\, \symNE_2 \,}{\symURI'}{\symWoD} \text{ for some } \symURI'\!\in \fctEvalSetQcW{\, \symNE_1 \,}{\symURI}{\symWoD} \text{ with } \symURI'\!\in\fctDom{\fctsymADoc} \}
	\\
	\fctEvalSetQcW{\, \symNE_1|\,\symNE_2 \,}{\symURI}{\symWoD} & \definedAs
	\fctEvalSetQcW{\, \symNE_1 \,}{\symURI}{\symWoD} \cup \fctEvalSetQcW{\, \symNE_2 \,}{\symURI}{\symWoD}
	\\
	\fctEvalSetQcW{\, \symNE^* \,}{\symURI}{\symWoD} & \definedAs
	\{u\} \cup \fctEvalSetQcW{\, \symNE \,}{\symURI}{\symWoD} \cup \fctEvalSetQcW{\, \symNE/\symNE \,}{\symURI}{\symWoD} \cup \fctEvalSetQcW{\, \symNE/\symNE/\symNE \,}{\symURI}{\symWoD} \cup \cdots
	\\
	\fctEvalSetQcW{\, \symNE\peTest{(\OpASK P)} \,}{\symURI}{\symWoD} & \definedAs
	\{ \symURI' \mid \symURI' \in \fctEvalSetQcW{\, \symNE \,}{\symURI}{\symWoD},\ \symURI' \in\fctDom{\fctsymADoc}\text{ and } \fctEvalPGD{P \,}{\fctData{\fctADoc{\symURI'}}}{} \neq \emptyset \}
\end{align*}

We next show that for every NautiLOD expression there exists an equivalent \query.
Notice that the evaluation of a 
NautiLOD expression is a set of URIs, whereas the evaluation of an \query\ is 
a set of mappings. Thus, to formally state our result we compare NautiLOD
with \queries\ that have a single \emph{free variable}.
Let $\symQuery(?x)$ be an \query\ with $?x$ as free variable%
	. We say that $\symQuery(?x)$ and a NautiLOD expression~$\symNE$
are equivalent if 
for every Web of Linked Data $\symWoDTuple$ and
	URIs $\symURI,\symURI'$ with
$\symURI \in \fctDom{\fctsymADoc}$ it holds that
$\symURI'\! \in \fctEvalSetQcW{\symNE}{\symURI}{\symWoD} \text{ if and only if } \lbrace ?x \mapsto \symURI' \rbrace \in \fctEvalQWS{\symQuery(?x)}{\symWoD}{\lbrace\symURI\rbrace}\!.$

\begin{theorem}\label{theo:LDQLCoversNautiLOD}
For every NautiLOD expression $\symNE$, there exists an \query\ $\symQuery(?x)$, with $?x$ a free variable, 
that is equivalent to $\symNE$.
\end{theorem}

\begin{proof}[Sketch]
The proof begins with a simple translation that replaces every $p\in\symAllURIs$ in a
NautiLOD expression by a link pattern $\tuple{\peContextURI,p,\peWildcard}$.
For instance, the expression~$p_1/p_2^*$ is translated into $\tuple{\peContextURI,p_1,\peWildcard}/\tuple{\peContextURI,p_2,\peWildcard}^*$\!.
	To translate $\tuple{\peWildcard}$ and $\peTest{(\OpASK P)}$ we use
$\peSubquery{?v}{\symQuery}$. 
The complete translation poses several other complications%
	~(as described in the \PaperVersion{extended version~\cite{ExtendedVersion}}\ExtendedVersion{appendix}). In particular, the last step of NautiLOD expressions
must be translated by using a SPARQL pattern and not an LPE.
For this we use the following property.
Given a regular expression $r$ that does not generate the empty word, one can always write $r$ as
$r_1/a_1|\cdots |r_k/a_k$ where the ${a_i}$\!'s are base symbols of the alphabet. Thus, we can translate $r$ by using LPEs to translate
	the ${r_i}$\!'s as outlined above; next, translate the ${a_i}$\!'s by using a method similar to the proof of Theorem~\ref{theo:LDQLCoversPP},
and finally use$\OpUNION$\! for $|$.
\end{proof}

Along the same lines of Theorem~\ref{theo:LDQLMoreThanPP} one can prove the following result.

\begin{theorem}\label{theo:LDQLMoreThanNautiLOD}
There exists an LDQL query $\symQuery(?x)$ that cannot be expressed in NautiLOD.
\end{theorem}
\vspace*{-17pt}

\subsection{Comparison with SPARQL under Reachability-Based Query Semantics}
\label{ssec:Comparison:ReachSem}

In~\cite{Hartig12:TheoryPaper} the author introduces a family of reach\-abil\-i\-ty-based query
	semantics based on which
SPARQL
	\removable{graph}
patterns can be used as a query language for Linked Data on the WWW.
Similar to how the scope of
the SPARQL part of a basic \query\ is restricted to
particular \docs, reach\-abil\-i\-ty-based semantics restrict the scope of SPARQL queries to \docs\ that can be reached by traversing a well-de\-fined set of data links. To specify what data links belong to such a set,
the notion of a \emph{reachability criterion} is used; 
that is, a function $\symReachCrit\!: \symAllTriples \times \symAllURIs \times \symAllPatterns \rightarrow \lbrace \true, \false \rbrace$ where $\symAllPatterns$ denotes the set of all SPARQL graph patterns.
	Then, given
such a reachability criterion $\symReachCrit$,
a~finite set $\symSeedURIs$ of URIs
and a SPARQL graph~pattern~$\symPattern$,
a~\doc\ $d \in \symAllDocs$ is \emph{($\symReachCrit,\symSeedURIs,\symPattern$)-reach\-able} in a \wold\ $\symWoDTuple$ if
	any
of the following two conditions~holds:

\begin{enumerate}
	\item
		There exists a URI $\symURI \in \symSeedURIs$ such that $\fctADoc{\symURI} = d$; or


	\item
		there exists a link graph edge $\tuple{d_\mathsf{src},(t,\symURI),d_\mathsf{tgt}} \in \mathcal{G}_\symWoD$ such that
		\enumA~$d_\mathsf{src}$ is $(\symReachCrit,\symSeedURIs,\symPattern)$-reach\-able in $\symWoD$\!,
		\enumB~$\symReachCrit(t,\symURI,\symPattern)=\true$,
		and \enumC~$d_\mathsf{tgt}=d$.
\end{enumerate}

Notice how the second condition restricts the notion of reachability by
	ignoring data links that do not
satisfy the given reachability criterion $\symReachCrit$.
Concrete examples of reachability criteria are $\cAll$, $\cNone$, and $\cMatch$~\cite{Hartig12:TheoryPaper}, where $\cAll$ selects all data links, and $\cNone$ ignores all data links; i.e., $\cAll(t,\symURI,\symPattern) \definedAs \true$ and $\cNone(t,\symURI,\symPattern) \definedAs \false$ for all \removable{tuples} $\tuple{t,\symURI,\symPattern} \in \symAllTriples \times \symAllURIs \times \symAllPatterns$.
In contrast to such an all-or-nothing strategy,
	\removable{criterion}
$\cMatch$ returns $\true$ for every data link whose triple matches a triple pattern of the given graph pattern; formally, $\cMatch(t,\symURI,\symPattern) \definedAs \true$ if and only if there exists some solution mapping $\mu$ such that $\mu[\symTP] = t$ for an arbitrary triple pattern $\symTP$ that is contained~in~$\symPattern$.


Given the notion of a reachability criterion, it is possible to define a family of~(reach\-abil\-i\-ty-based) query semantics for SPARQL. To this end, let $\symReachCrit$ be a reachability criterion, 
let $\symSeedURIs$ be a finite set of URIs,
and let $\symPattern$ be a SPARQL graph pattern. Then, for any \wold\ $\symWoDTuple$, the \emph{$\symSeedURIs$-based evaluation} of $\symPattern$ over $\symWoD$ \emph{under $\symReachCrit$-se\-man\-tics}, denoted by $\fctEvalReachPcSW{\symPattern}{\symReachCrit}{\symSeedURIs}{\symWoD}$\!, is
the set of solution mappings 
$\fctEvalPGD{\symPattern}{\symRDFgraph}{}$ where $\symRDFgraph$ is the RDF graph
	that consists of
all triples from all \docs\ that are ($\symReachCrit,\symSeedURIs,\symPattern$)-reach\-able in $\symWoD$\!.


While there exist an infinite number of possible reachability criteria, 
in this paper we focus on $\cAll$,
$\cNone$,
and $\cMatch$.
The following two results show that LDQL is strictly more expressive than SPARQL graph patterns under any of these three query semantics.

\begin{theorem}\label{thm:LDQLCoversReachSem}
	Let $\symReachCrit\in \{\cAll, \cNone, \cMatch\}$.
	For every SPARQL graph pattern $\symPattern$ there exists an \query\ $\symQuery$ such that $\fctEvalReachPcSW{\symPattern}{\symReachCrit}{\symSeedURIs}{\symWoD}\! = \fctEvalQWS{\symQuery}{\symWoD}{\symSeedURIs}$ for every
		Web $\symWoD$ and $\symSeedURIs \subseteq \symAllURIs$.
\end{theorem}

\begin{proof}[Sketch]
We only sketch the case of $\cAll$-se\-man\-tics.
	In this case, one can prove
that the LPE
$\symLPE^{\cAll}=\tuple{\peWildcard,\peWildcard,\peWildcard}^*$ simulates the reachability criterion $\cAll$, and, thus,
$\fctEvalReachPcSW{\symPattern}{\cAll}{\symSeedURIs}{\symWoD}\! = \fctEvalQWS{\tuple{\symLPE^{\cAll},P}}{\symWoD}{\symSeedURIs}$.
One can also find LPEs to simulate $\cNone$ and $\cMatch$.
\end{proof}

\begin{theorem}\label{thm:LDQLMoreThanReachSem}
	Let $\symReachCrit\in\{\cAll,\cNone,\cMatch\}$.
	There exists an \query\ $\symQuery$ for which~there does not exist a SPARQL pattern $\symPattern$ such that $\fctEvalReachPcSW{\symPattern}{\symReachCrit}{\symSeedURIs}{\symWoD}\! = \fctEvalQWS{\symQuery}{\symWoD}{\symSeedURIs}$ for every
		$\symWoD$ and $\symSeedURIs \subseteq \symAllURIs$.
\end{theorem}
%
%
%
%
\vspace*{-15pt}

\section{Web-Safeness of LDQL Queries} \label{sec:Safeness}
\vspace*{-8pt}

In this section we study the
	``Web-safe\-ness'' of \queries, where, informally,
we call a query \emph{Web-safe} if a complete execution of the query over the WWW is possible in practice~(which is not the case for all \queries\ as we shall see).
To provide a more formal definition of
	this notion of
Web-safe\-ness we make the following observations.
While the mathematical structures introduced by our data model capture the notion of Linked Data on the WWW formally%
	~(and, thus, allow us to provide a formal semantics for \queries)%
, in practice, these structures are not available completely
	for
the WWW. For instance,
given that an infinite number of strings can be used as HTTP URIs~\cite{Fielding99:HTTP}, we cannot assume complete information about which URIs are in the domain of the partial function $\fctsymADoc$~(i.e., can be looked up to retrieve some \doc) and which are not; in fact, disclosing this information would require a process that systematically tries to look up every possible HTTP URI and, thus, would never terminate. Therefore, it is also impossible to guarantee the discovery of every \doc\ in the set~$\symDocs$~(without looking up an infinite number of URIs). Consequently, any query whose execution requires a complete enumeration of this set is not feasible in practice. Based on these observations, we define \emph{Web-safe\-ness} of \queries\ as~follows.

\begin{mydef}{Web-Safeness}%
	\label{def:WebSafeness}%
	An \query\ $\symQuery$ is \definedTerm{Web-safe} if there exists an algorithm that, for any
		finite \wold\ $\symWoDTuple$
	and any finite set $\symSeedURIs$ of URIs,
	computes $\fctEvalQWS{\symQuery}{\symWoD}{\symSeedURIs}$ by looking up only a finite number of URIs without assuming
		an a priori availability of any information about
	the sets $\symDocs$ and $\fctDom{\fctsymADoc}$.
\end{mydef}

\begin{myexample} \label{ex:WebSafeness}
	Recall
		our example queries $\symQuery_\mathsf{ex}$, $\symQuery_\mathsf{ex}'$, and $\symQuery_\mathsf{ex}''$~(cf.~Example~\ref{ex:LDQLSemantics}).
%
	For query $\symQuery_\mathsf{ex} = \left( \OpSEED ?x\ \bigtuple{\peEmpty, \tuple{?x,p_1,?z}} \right)$, any URI $\symURI \in \symAllURIs$ may be used to obtain
		a nonempty subset of the query result
	as long as a lookup of $\symURI$ retrieves a \doc\ whose data includes {\triple}s that match $\tuple{\symURI,p_1,?z}$. Therefore, without access to $\symDocs$ or $\fctDom{\fctsymADoc}$ of the queried Web $\symWoDTuple$, the completeness of the computed query result can be guaranteed only by checking each of the infinitely many possible HTTP URIs. Hence, query $\symQuery_\mathsf{ex}$ is \emph{not} Web-safe.
	In contrast, although it contains
	$\symQuery_\mathsf{ex}$ as a subquery, query $\symQuery_\mathsf{ex}'' = ( \symQuery_\mathsf{ex} \OpAND \symQuery_\mathsf{ex}' )$ is Web-safe, and so is $\symQuery_\mathsf{ex}' = \tuple{\symLPE_\mathsf{ex},\symBGP_\mathsf{ex}}$. Given~$\symURI_\mathsf{A}$ as seed URI, a possible execution algorithm for $\symQuery_\mathsf{ex}'$ may first compute $\fctEvalPathPdW{\symLPE_\mathsf{ex}}{\symURI_\mathsf{A}}{\symWoD}$ by traversing the queried
		Web%
	~$\symWoD$ based on
	$\symLPE_\mathsf{ex}$. Thereafter, the algorithm retrieves \docs\ by looking up all URIs $\symURI \in \fctEvalPathPdW{\symLPE_\mathsf{ex}}{\symURI_\mathsf{A}}{\symWoD}$~(or simply keeps these \docs\ after the traversal); and, finally, the algorithm evaluates pattern $\symBGP_\mathsf{ex}$ over the union of the RDF data in the retrieved \docs. If
		$\symWoD$ is finite~(i.e., contains a finite number of \docs),
	the traversal process requires a finite number of URI lookups only, and so does the retrieval of \docs\ in the second step; the final step does not look up any URI.
	To see that $\symQuery_\mathsf{ex}''$ is also Web-safe
		we note that after executing subquery $\symQuery_\mathsf{ex}'$~(e.g., by using the algorithm as outlined before), the execution of the other~(non-Web-safe) subquery~$\symQuery_\mathsf{ex}$ can be reduced to
			a finite number of URI lookups,
		namely the URIs bound to variable $?x$ in solution mappings obtained for subquery $\symQuery_\mathsf{ex}'$. Although any other URI may also be used to obtain solution mappings for $\symQuery_\mathsf{ex}$, such solution mappings cannot be joined with any of the solution mappings for $\symQuery_\mathsf{ex}'$ and, thus, are irrelevant for the result of $\symQuery_\mathsf{ex}''$.
\end{myexample}

The example illustrates that there exists an \query\ that is not Web-safe. In fact, it is not difficult to see that the argument for the non-Web-safe\-ness of query $\symQuery_\mathsf{ex}$ as made in the example can be applied to any \query\ of the form $( \OpSEED ?x\ \symQuery )$ where \removable{subquery} $\symQuery$ is a~(satisfiable) basic \query; that is, none of these queries is Web-safe. However, the example also shows that more complex queries that contain such non-Web-safe subqueries may still be Web-safe.
	Therefore, we now show properties to identify \queries\ that are Web-safe even if some of their subqueries are not.
We begin with
	queries of the forms $\tuple{\symLPE,\symPattern}$,
$\pi_V \symQuery$, $( \OpSEED\ U\ \symQuery )$, and~$( \symQuery_1 \OpUNION \ldots \OpUNION \symQuery_n )$.

\begin{proposition} \label{prop:Safeness:FirstTrivialProperties}
	An \query\ $\symQuery$ is Web-safe if any of the following properties holds:
	\begin{enumerate}
		\item \label{prop:Safeness:FirstTrivialProperties:BasicQuery}
			Query $\symQuery$ is of the form $\tuple{\symLPE,\symPattern}$ and $\symLPE$ is Web-safe, where we call an LPE Web-safe if either
			\enumA~it is of the form $\peSubquery{?v}{\symQuery'}$ and \query\ $\symQuery'$ is Web-safe,
			or \enumB~it is of any form other than $\peSubquery{?v}{\symQuery'}$ and all its subexpressions~(if any) are Web-safe LPEs;

		\item \label{prop:Safeness:FirstTrivialProperties:ProjectionAndSEED}
			Query $\symQuery$ is of the form $\pi_V \symQuery'$ or $( \OpSEED\ U\ \symQuery' )$, and subquery $\symQuery'$ is Web-safe;~or

		\item \label{prop:Safeness:FirstTrivialProperties:UNION}
			Query $\symQuery$ is of the form $( \symQuery_1\OpUNION \ldots \OpUNION \symQuery_n )$ and each $\symQuery_i$~($1 \!\leq\! i \!\leq\! n$) is Web-safe.
	\end{enumerate}
\end{proposition}


	It remains to discuss
\queries\ of the form
	$( \symQuery_1\OpAND \ldots\OpAND \symQuery_m )$.
Our discussion of query~$\symQuery_\mathsf{ex}''
$ in Example~\ref{ex:WebSafeness} suggests that such queries can be shown to be Web-safe if all non-Web-safe subqueries are of the form $( \OpSEED ?v\ \symQuery )$ and it is possible to execute these subqueries by using variable bindings obtained from other subqueries. A~necessary condition for
	this
execution strategy is that the variable in question~(i.e., $?v$) is guaranteed to be bound in every possible solution mapping obtained from the other subqueries.

To
	allow for an automated verification of this condition
we adopt
	Buil-Aranda et al.'s notion of strongly bound variables%
~\cite{BuilAranda11:SemanticsAndOptimizationOfSPARQLFedExt}%
.
%
To this end, for any SPARQL graph pattern $\symPattern$, let $\fctSBVars{\symPattern}$ denote the set of strongly bound variables in $\symPattern$ as defined by Buil-Aranda et al.~\cite{BuilAranda11:SemanticsAndOptimizationOfSPARQLFedExt}.
	For the sake of space, we do not repeat the definition here. However, we
emphasize that $\fctSBVars{\symPattern}$ can be constructed recursively%
, and each variable
	in $\fctSBVars{\symPattern}$
is guaranteed to be bound in every possible solution for $\symPattern$~\cite[Proposition~1]{BuilAranda11:SemanticsAndOptimizationOfSPARQLFedExt}. To carry over these properties to \queries, we use the notion of strongly bound variables in SPARQL patterns to define the following notion of
	strongly bound variables in
\queries; thereafter, in Lemma~\ref{lem:StrongBoundedness}, we show the desired
	boundedness~%
guarantee.

\begin{mydef}{Strongly Bound Variables}%
	\label{def:StrongBoundedness}%
		The
	set of \definedTerm{strongly bound variables} in an \query\ $\symQuery$, denoted by $\fctSBVars{\symQuery}$, is defined recursively as follows:
	\begin{enumerate}
		\item If $\symQuery$ is of the form $\tuple{\symLPE,\symPattern}$, then $\fctSBVars{\symQuery} \definedAs \fctSBVars{\symPattern}$.
		\item If $\symQuery$ is of the form $( \symQuery_1 \OpAND \symQuery_2 )$, then $\fctSBVars{\symQuery} \definedAs \fctSBVars{\symQuery_1} \cup \fctSBVars{\symQuery_2}$.
		\item If $\symQuery$ is of the form $( \symQuery_1 \OpUNION \symQuery_2 )$, then $\fctSBVars{\symQuery} \definedAs \fctSBVars{\symQuery_1} \cap \fctSBVars{\symQuery_2}$.
		\item If $\symQuery$ is of the form $\pi_V \symQuery'$\!, then $\fctSBVars{\symQuery} \definedAs \fctSBVars{\symQuery'} \cap V$.
		\item If $\symQuery$ is of the form $(\OpSEED\ U\ \symQuery' )$, then $\fctSBVars{\symQuery} \definedAs \fctSBVars{\symQuery'}$.
		\item If $\symQuery$ is of the form $(\OpSEED ?v\ \symQuery' )$, then $\fctSBVars{\symQuery} \definedAs \fctSBVars{\symQuery'} \cup \lbrace ?v \rbrace$.
	\end{enumerate}
\end{mydef}

\begin{lemma} \label{lem:StrongBoundedness}
	Let $\symQuery$ be an \query. For every finite set $\symSeedURIs$ of URIs, every \wold\ $\symWoD$\!, and every solution mapping $\mu \in \fctEvalQWS{\symQuery}{\symWoD}{\symSeedURIs}$, it holds that $\fctSBVars{\symQuery} \subseteq \fctDom{\mu}$.
\end{lemma}


	We
are now ready to show the following result.

\begin{theorem} \label{thm:Safeness:AND}
	An \query\ of the form $( \symQuery_1\OpAND \symQuery_2\OpAND \ldots\OpAND \symQuery_m )$ is Web-safe if there exists a total order $\prec$ over the set of subqueries $\lbrace \symQuery_1, \symQuery_2, \ldots , \symQuery_m \rbrace$ such that for each subquery~$\symQuery_i$~($1 \leq i \leq m$), it holds that either \enumA~$\symQuery_i$ is Web-safe or \enumB~$\symQuery_i$ is of the form~$( \OpSEED ?v\ \symQuery )$ where $\symQuery$ is Web-safe and $?v \in \bigcup_{\symQuery_j \prec \symQuery_i} \fctSBVars{\symQuery_j}$.
\end{theorem}

\begin{proof}[Sketch]
	We prove Theorem~\ref{thm:Safeness:AND} based on an iterative algorithm that generalizes the execution
		of query $\symQuery_\mathsf{ex}''$ as outlined in Example~\ref{ex:WebSafeness}.
	That is, the algorithm executes the subqueries
		$\symQuery_1 \ldots \symQuery_m$ sequentially in the order $\prec$ such that each iteration executes one of the subqueries by using the solution mappings computed during the previous~iteration.
%
%
\end{proof}

With the results in this section we have all ingredients to devise a procedure to
	show Web-safe\-ness for a large number of queries~(including queries that are arbitrarily nested).
However, as a
	\removable{potential} limitation \removable{of such a procedure}
we note that Theorem~\ref{thm:Safeness:AND} can be applied only in cases in which all non-Web-safe subqueries are of the form~$( \OpSEED ?v\ \symQuery )$. For instance, the theorem cannot be applied to show that an \query\ of the form
$\bigl(
	\symQuery_1
	\OpAND
	( \symQuery_2 \OpUNION (\OpSEED ?x\ \symQuery_3) )
\bigr)$
is Web-safe if $?x \in \fctSBVars{\symQuery_1}$ and $\symQuery_1$, $\symQuery_2$ and $\symQuery_3$ are Web-safe.
On the other hand, for the semantically equivalent query
$\bigl(
	( \symQuery_1 \OpAND \symQuery_2 )
	\OpUNION
	( \symQuery_1 \OpAND (\OpSEED ?x\ \symQuery_3) )
\bigr)$ we can show Web-safe\-ness based on Theorem~\ref{thm:Safeness:AND}~(and Proposition~\ref{prop:Safeness:FirstTrivialProperties}). Fortunately, we may leverage the following fact to improve the effectiveness of \removable{applying Theorem~\ref{thm:Safeness:AND} in} the procedure that we aim to devise.

\begin{fact} \label{fact:WebSafenessAndEquivalence}
	If an \query\ $\symQuery$ is Web-safe, then so is any \query\ $\symQuery'$ with $\symQuery' \equiv \symQuery$.
\end{fact}

As a consequence of Fact~\ref{fact:WebSafenessAndEquivalence}, we may use the equivalences in Lemma~\ref{lem:Equivalences:Distributiveness} to rewrite a given
	query
into an equivalent query that is more suitable for testing Web-safe\-ness based on our results. To this end, we introduce specific normal forms for~\queries:

\begin{mydef}{Normal Forms}%
	\label{def:NormalForms}%
	An \query\ is in \definedTerm{\textsc{\textsf{union}}-free normal form} if it is of the form
		$( \symQuery_1 \OpAND \ldots \OpAND \symQuery_m )$
	with $m \geq 1$ and each
	$\symQuery_i$~($1 \leq i \leq m$) is either \enumA~a basic \query\ or \enumB~of the form $\pi_V \symQuery$, $( \OpSEED\ U\ \symQuery )$ or $( \OpSEED\ ?v\ \symQuery )$ such that subquery~$\symQuery$ is
		in \OpUNION\!-free
	normal form.
%
		An
	\query\ is in \definedTerm{\textsc{\textsf{union}} normal form} if it is of the form
		$( \symQuery_1 \OpUNION \ldots \OpUNION \symQuery_n )$
	with $n \!\geq\! 1$ and each
	$\symQuery_i$~($1 \!\leq\! i \!\leq\! n$) is in\OpUNION\!-free normal~form.
\end{mydef}

The following result is an immediate consequence of Lemma~\ref{lem:Equivalences:Distributiveness}.

\begin{corollary} \label{cor:NormalForm}
	Every \query\ is equivalent to an \query\ in\OpUNION\! normal form.
\end{corollary}

In conjunction with Fact~\ref{fact:WebSafenessAndEquivalence}, Corollary~\ref{cor:NormalForm} allows us to focus on \queries\ in\OpUNION\! normal form without losing generality.
We are now ready to specify our procedure that applies the results in this paper to test a given \query~$\symQuery$ for Web-safe\-ness: First, by using the equivalences in Lemma~\ref{lem:Equivalences:Distributiveness}, the query has to be rewritten into a semantically equivalent \query\
	$\symQuery_\mathsf{nf} \!=\! ( \symQuery_1 \OpUNION \ldots \OpUNION \symQuery_n )$
that is in\OpUNION\! normal form.
	Next,
the following
test has to be repeated for every subquery $\symQuery_i$~($1 \leq i \leq n$); recall that each of these subqueries is in\OpUNION\!-free normal form; i.e.,
	$\symQuery_i = ( \symQuery^i_1 \OpAND \ldots \OpAND \symQuery_{m_i}^i )$.
The test is to find an order for their subqueries
	$\symQuery_1^i, \ldots , \symQuery_{m_i}^i$
that satisfies the conditions in Theorem~\ref{thm:Safeness:AND}. Every top-lev\-el subquery $\symQuery_i$~($1 \leq i \leq n$) for which such an order exists, is Web-safe~(cf.~Theorem~\ref{thm:Safeness:AND}). 
	If all top-lev\-el subqueries are identified to be Web-safe by this test, then $\symQuery_\mathsf{nf}$ is Web-safe~(cf.~Proposition~\ref{prop:Safeness:FirstTrivialProperties}), and so is $\symQuery$~(cf.~Fact~\ref{fact:WebSafenessAndEquivalence}).

	
The given conditions are sufficient to show Web-safeness of LDQL.
It remains open whether there
	exists a~(decidable) 	sufficient \emph{and}~necessary condition for Web-safeness.

\vspace*{-10pt}

\section{Concluding Remarks and Future Work}
\label{sec:Conclusion}
\vspace*{-5pt}

LDQL, the query language that we introduce in this paper, allows users to express
queries over Linked Data on the WWW. We defined LDQL such that navigational 
features for selecting the que\-ry-rel\-e\-vant \docs\ on the Web are separate from patterns that are meant to be evaluated over the data in the selected \docs. This separation distinguishes LDQL from other approaches to express queries over Linked Data.

We focused on expressiveness, by comparing LDQL with previous formalisms, and on the notion
of Web-safeness. Several topics remain open for future work. One of them is the complexity of query evaluation.
A classical complexity analysis is easy to perform if we assume that all the data and documents are available
as if they were in a centralized repository, and that they can be processed via a RAM machine model.
We conjecture that under this model, the data complexity of evaluating LDQL will be polynomial.
Nevertheless, a more interesting complexity analysis should consider a model that captures the inherent 
way of accessing the \wold\ via HTTP requests, the overhead of data communication and transfer, 
the distribution of data and documents, etc.
A more practical direction for future research on LDQL is the development of approaches to 
actually implement \queries\ efficiently.\medskip

\noindent
{\bf Acknowledgements}\; P\'erez is supported by the Millennium Nucleus Center for Semantic Web Research, Grant NC120004, and Fondecyt grant 1140790.

\bibliographystyle{splncs03}

\vspace*{-10pt}

\begin{thebibliography}{10}
\providecommand{\url}[1]{\texttt{#1}}
\providecommand{\urlprefix}{URL }

\vspace*{-5pt}

\bibitem{ArenasCP12}
Arenas, M., Conca, S., P{\'{e}}rez, J.: Counting beyond a yottabyte, or how
  {SPARQL} 1.1 property paths will prevent adoption of the standard. In: WWW
  2012. pp. 629--638 (2012)

\bibitem{Arenas09:SemanticsAndComplexityOfSPARQLBookChapter}
Arenas, M., Gutierrez, C., P{\'{e}}rez, J.: {On the Semantics of {SPARQL}}. In:
  Semantic Web Information Management - {A} Model-Based Perspective, chap.~13,
  pp. 281--307. Springer (2009)

\bibitem{BernersLee06:LinkedData}
Berners-Lee, T.: {L}inked {D}ata. At
  http://www.w3.org/DesignIssues/LinkedData.html (2006)

%

\bibitem{BuilAranda11:SemanticsAndOptimizationOfSPARQLFedExt}
{Buil-Aranda}, C., Arenas, M., Corcho, O.: {Semantics and Optimization of the
  SPARQL 1.1 Federation Extension}. In: Proc. 8th Extended Semantic Web Conf.
  (2011)

\bibitem{Cyganiak14:RDFConcepts}
Cyganiak, R., Wood, D., Lanthaler, M.: {RDF 1.1 Concepts and Abstract Syntax}.
  W3C Recommendation (Feb 2014)

\bibitem{Fielding99:HTTP}
Fielding, R., Gettys, J., Mogul, J.C., Frystyk, H., Masinter, L., Leach, P.J.,
  Berners-Lee, T.: {Hypertext Transfer Protocol -- HTTP/1.1} (Jun 1999)

\bibitem{Fionda15:NautiLODArticle}
Fionda, V., Pirr\`{o}, G., Gutierrez, C.: {NautiLOD: A Formal Language for the
  Web of Data Graph}. ACM Transactions on the Web  9(1),  5:1--5:43 (2015)

\bibitem{Harris13:SPARQL1_1Language}
Harris, S., Seaborne, A., Prud'hommeaux, E.: {{SPARQL} 1.1 Query Language}. W3C
  Recommendation (Mar 2013)

\bibitem{Harth12:CompletenessClassesForLDQueries}
Harth, A., Speiser, S.: {On Completeness Classes for Query Evaluation on Linked
  Data}. In: Proc. 26th AAAI Conf. (2012)

\bibitem{AMWVersion}
Hartig, O.: {LDQL: A Language for Linked Data Queries}. In AMW 2015

\bibitem{Hartig12:TheoryPaper}
Hartig, O.: {{SPARQL} for a {W}eb of {L}inked {D}ata: Semantics and
  Computability}. In: Proc. 9th Extended Semantic Web Conf. (2012)

\bibitem{Hartig13:Survey}
Hartig, O.: {An Overview on Execution Strategies for {L}inked {D}ata Queries}.
  Datenbank-Spektrum  13(2) (2013)

\PaperVersion{%
%
\bibitem{ExtendedVersion}
Hartig, O., P{\'{e}}rez, J.: {LDQL: A Query Language for the Web of Linked Data ({Extended Version})}. CoRR  abs/1507.04614 (2015),
  http://arxiv.org/abs/1507.04614
%
}
\ExtendedVersion{%
%
\bibitem{ProceedingsVersion}
Hartig, O., P{\'{e}}rez, J.: {LDQL: A Query Language for the Web of Linked Data}. In: Proc. 14th Int. Semantic Web Conf. (2015)
%
}

\bibitem{Hartig15:WebPPsPaper}
Hartig, O., Pirr\`{o}, G.: {A Context-Based Semantics for SPARQL Property Paths
  over the Web}. In: Proc. 12th Extended Semantic Web Conf. (2015)

\bibitem{WebSQL}  
Mendelzon, A.~O., Mihaila, G.~A., Milo T.: {Querying the World Wide Web}. In: PDIS (1996)

\bibitem{Perez09:SemanticsAndComplexityOfSPARQL}
P\'{e}rez, J., Arenas, M., Gutierrez, C.: {Semantics and Complexity of
  {SPARQL}}. ACM Transactions on Database Systems  34 (2009)

\bibitem{PerezAG10}
P{\'{e}}rez, J., Arenas, M., Gutierrez, C.: {{nSPARQL}: {A} Navigational Language for {RDF}}. J. Web Sem.  8(4),  255--270 (2010)

\bibitem{Umbrich12:LinkedDataQueriesWithReasoning:Article}
Umbrich, J., Hogan, A., Polleres, A., Decker, S.: {Link Traversal Querying for
  a Diverse Web of Data}. Semantic Web Journal  (2014)

\end{thebibliography}

\ExtendedVersion{%
\newpage
\appendix
\section{Proofs} \label{app:Proofs}
\subsection{Proof of Lemma~\ref{lem:Equivalences:AssocAndCommut}} \label{proof:lem:Equivalences:AssocAndCommut}

We formalize the claims in Lemma~\ref{lem:Equivalences:AssocAndCommut} as follows: Let $\symQuery_1$, $\symQuery_2$, and $\symQuery_3$ be \queries, the following semantic equivalences hold:
\begin{align}
	( \symQuery_1 \OpAND \symQuery_2 ) &\equiv ( \symQuery_2 \OpAND \symQuery_1 )
	\\
	( \symQuery_1 \OpUNION \symQuery_2 ) &\equiv ( \symQuery_2 \OpUNION \symQuery_1 )
	\\
	( \symQuery_1 \OpAND (\symQuery_2 \OpAND \symQuery_3) ) &\equiv ( ( \symQuery_1 \OpAND \symQuery_2 ) \OpAND \symQuery_3 )
	\\
	( \symQuery_1 \OpUNION (\symQuery_2 \OpUNION \symQuery_3) ) &\equiv ( ( \symQuery_1 \OpUNION \symQuery_2 ) \OpUNION \symQuery_3 )
\end{align}

Since the definition of LDQL operators\OpAND\! and\OpUNION\! is equivalent to their SPARQL counterparts, these semantic equivalences follow from corresponding equivalences for SPARQL graph patterns as shown by P\'{e}rez et al.~\cite[Lemma~2.5]{Perez09:SemanticsAndComplexityOfSPARQL}.

\subsection{Proof of Lemma~\ref{lem:Equivalences:Distributiveness}} \label{proof:lem:Equivalences:Distributiveness}

The equivalences follow directly from the definition of every operator.

\subsection{Proof of Proposition~\ref{prop:syntacticSugar}}

The proof is based on a recursive translation of link path expressions beginning with link patterns.
Let $\tuple{y_1,y_2,y_3}$ be a link pattern. We construct an LPE $\trans_L(\tuple{y_1,y_2,y_3})$ as follows.
Assume that $y_1=\peWildcard$, then we construct the LDQL query
\[
q_1=  \big\langle \varepsilon, (\OpGRAPH ?u\ (?\textit{out},Y_2,Y_3))\big\rangle
\]
where 
(i) if $y_2=\peContextURI$ then $Y_2=?u$, (ii) if $y_2\in \symAllURIs$ then $Y_2=y_2$ and (iii) if $y_2=\peWildcard$ then $Y_2=?y_2$.
And similarly, if $y_3=\peContextURI$ then $Y_3=?u$, (ii) if $y_3\in \symAllURIs$ then $Y_3=y_3$ and (iii) if $y_3=\peWildcard$ then $Y_3=?y_3$.
If $?y_2=\peWildcard$ then que can construct query $q_2=  \tuple{ \varepsilon, (\OpGRAPH ?u\ (Y_1,?\textit{out},Y_3))}$,
and if $?y_3=\peWildcard$ query $q_3=  \tuple{\varepsilon, (\OpGRAPH ?u\ (Y_1,Y_2,?\textit{out}))}$, following a similar
process as for $q_1$.
Now consider the query $q$ which is the \!$\OpUNION$\! of the above queries for every $y_i=\peWildcard$.
Then LPE $\trans_L(\tuple{y_1,y_2,y_3})$ is constructed as
\[
\trans_L(\tuple{y_1,y_2,y_3})=\peSubquery{?\textit{out}}{ q}.
\]
It is not difficult to prove that 
$\fctEvalSetQcW{\trans_L(\tuple{y_1,y_2,y_3})}{u}{\symWoD} = \fctEvalQW{{\tuple{y_1,y_2,y_3}}}{\symWoD}^{u}$.

We now define the translation in general:
\begin{itemize}
\item For the case of LPE $r=r_1/r_2$, we have that $\trans_L(r)=\peSubquery{?v}{\symQuery}$ where~$\symQuery$~is:
\[
	\bigl( \
		\tuple{ \trans_L(r_1), (\OpGRAPH ?x\ \{\ \}) }
		\, \OpAND \,
		\bigl( \OpSEED ?x\ \tuple{\trans_L(r_2), (\OpGRAPH ?v\ \{\ \})} \bigr)
	\ \bigr)
	.
\]
\item For the case of LPE $r=r_1|r_2$, we have that
$\trans_L(r)=\peSubquery{?v}{\symQuery}$ where~$\symQuery$~is:
\[
	\bigl( \
		\tuple{ \trans_L(r_1), (\OpGRAPH ?v\ \{\ \}) }
		\, \OpUNION \,
		\tuple{ \trans_L(r_2), (\OpGRAPH ?v\ \{\ \}) }
	\ \bigr)
	.
\]
\item For the case of LPE $r=[r_1]$, we have that
$\trans_L(r)=\peSubquery{?v}{\symQuery}$ where~$\symQuery$~is:
\[
	\bigl( \
		\tuple{ \varepsilon, (\OpGRAPH ?v\ \{\ \}) }
		\, \OpAND \,
		\pi_{\{?v\}}\bigl( \OpSEED ?v\ \tuple{\trans_L(r_1), (\OpGRAPH ?x\ \{\ \})} \bigr)
	\ \bigr)
	.
\]
\end{itemize}
The general proof proceed by induction.
We next prove that $\fctEvalSetQcW{\trans_L(r_1|r_2)}{u}{\symWoD} = \fctEvalQW{r_1|r_2}{\symWoD}^{u}$.
The proof for the other cases are similar.
Thus assume that $u'\in \fctEvalQW{r_1|r_2}{\symWoD}^{u}$, then we know that $u'\in \fctEvalQW{r_1}{\symWoD}^{u}\cup
\fctEvalQW{r_2}{\symWoD}^{u}$.
If $u'\in \fctEvalQW{r_1}{\symWoD}^{u}$ then by induction hypothesis we know that 
$u'\in\fctEvalQW{\trans_L(r_1)}{\symWoD}^{u}$.
Now notice that 
\[
\fctEvalQW{\tuple{ \trans_L(r_1), (\OpGRAPH ?v\ \{\ \}) }}{\symWoD}^{\{u\}} =
\fctEvalQW{(\OpGRAPH ?v\ \{\ \})}{}^{\mathcal D}
\]
Where $\mathcal{D}=\dataset_\symWoD(\fctEvalQW{\trans_L(r_1)}{\symWoD}^{u})$.
Thus given that $u'\in\fctEvalQW{\trans_L(r_1)}{\symWoD}^{u}$ we know that $\mathcal{D}$
has a dataset $\tuple{u',\fctData{\fctADoc{u'}}}$, which implies that
$\{?v\to u'\}$ is a solution for $\fctEvalQW{(\OpGRAPH ?v\ \{\ \})}{}^{\mathcal D}$,
and thus $\{?v\to u'\}\in \fctEvalQW{\tuple{ \trans_L(r_1), (\OpGRAPH ?v\ \{\ \}) }}{\symWoD}^{\{u\}}$.
From this it is straightforward to conclude that 
$u'\in \fctEvalSetQcW{\trans_L(r_1|r_2)}{u}{\symWoD}$.
The other direction is similar.

\subsection{Proof of Theorem~\ref{theo:LDQLMoreThanPP}}

Consider the LDQL $Q$ query given by
\[
\big(\OpSEED\ u\ \big\langle \tuple{\peContextURI,p,\peWildcard}, (?x,?x,?x) \big \rangle\big)
\]
with $u,p\in \symAllURIs$. 
Now assume that there exists a property path pattern $P$ and a set of URIs $S$ such that 
\[
\fctEvalSetQcW{P}{\mathsf{ctxt}}{\symWoD} = \fctEvalQW{Q}{\symWoD}^{S}
\]
for every Web of Linked Data $\symWoD$.
Let $u'\in \symAllURIs$.
Consider now $\symWoD_1$ having only two documents $d_1=\{(u,p,u')\}$ and $d_2=\{(a,a,a)\}$ and such that
$\fctADoc{u}=d_1$ and $\fctADoc{u'}=d_2$.
Moreover, consider $\symWoD_2$ having also two documents $d_1=\{(u,p,u')\}$ and $d_3=\{(b,b,b)\}$
such that $\fctADoc{u}=d_1$ and $\fctADoc{u'}=d_3$.
First notice that for every $S$ we have that
\[
\fctEvalQW{Q}{\symWoD_1}^{S}=\{\{?x\to a\}\}\;\; \neq \;\;
\fctEvalQW{Q}{\symWoD_2}^{S}=\{\{?x\to b\}\}
\]
Notice that $C^{\symWoD_1}(u)=C^{\symWoD_2}(u)=\{(u,p,u')\}$ and $C^{\symWoD_1}(u')=C^{\symWoD_2}(u')=\emptyset$.
In general, we have that for every term $v\neq u$ it holds that $C^{\symWoD_1}(v)=C^{\symWoD_2}(v)=\emptyset$.
This essentially shows that the context selectors $C^{\symWoD_1}$ and $C^{\symWoD_2}$ are equivalent.
Given that the semantics of property paths is based on context selectors it is easy to prove 
that for every PP-based SPARQL query $R$ we have that 
$\fctEvalSetQcW{R}{\mathsf{ctxt}}{\symWoD_1}=\fctEvalSetQcW{R}{\mathsf{ctxt}}{\symWoD_2}$.
This can be done by induction in the construction of PP-based SPARQL queries.
For example, the evaluation of a base PP-pattern of the form 
$(v,p,\beta)$, with $v\in \symAllURIs$ and $\beta\in \symAllURIs\cup \mathcal{V}$ over $\symWoD_1$ is given by
\[
\fctEvalSetQcW{(v,p,\beta)}{\mathsf{ctxt}}{\symWoD_1}=\{\mu\mid \fctDom{\mu}=\{\beta\}\cap \mathcal{V}
\text{ and }\mu[\tuple{v,p,\beta}]\in C^{\symWoD_1}(v)\}
\]
which is equal to $\fctEvalSetQcW{(v,p,\beta)}{\mathsf{ctxt}}{\symWoD_2}$ since $C^{\symWoD_1}(v)=C^{\symWoD_2}(v)$.
All the other cases for the construction of property paths are equivalent.
Moreover, since for the case of property path patterns the evaluation is the same over $\symWoD_1$ and over $\symWoD_2$,
we have that for a general PP-based query using operator $\OpAND$, $\OpUNION$, $\OpOPT$ and so on, the evaluation is
also the same. Thus we have that 
\[
\fctEvalSetQcW{P}{\mathsf{ctxt}}{\symWoD_1} = \fctEvalSetQcW{P}{\mathsf{ctxt}}{\symWoD_2}
\]
but also that
\[
\fctEvalQW{Q}{\symWoD_1}^{S} \neq 
\fctEvalQW{Q}{\symWoD_2}^{S}
\]
which contradicts the fact that $\fctEvalSetQcW{P}{\mathsf{ctxt}}{\symWoD}=\fctEvalQW{Q}{\symWoD}^{S}$ for every 
Web of Linked Data $W$.

\subsection{Proof of Theorem~\ref{theo:LDQLCoversPP}}



We associate to every property-path expression $r$, 
an LDQL query $Q_r(?x,?y)$ with
$?x$ and $?y$ as free variables. 
The definition of $Q_r(?x,?y)$ is by induction in the construction of property-path expressions.
In the construction, all the variables mentioned, besides $?x$ and $?y$, are considered
as fresh variables.
\begin{itemize}
\item If $r\in \symAllURIs$ then $Q_r(?x,?y)=(\OpSEED ?x\ \langle \varepsilon, (?x,r,?y)\rangle)$.
\item If $r=\ !(u_1\mid \cdots \mid u_k)$ with $u_i\in \symAllURIs$ then $Q_r(?x,?y)$ is defined as
\[
\bigg( \OpSEED ?x\ \big\langle \varepsilon, 
\big( (?x,?p,?y) \OpFILTER (?p\neq u_1 \wedge \cdots \wedge ?p\neq u_k)\big) \big\rangle\bigg).
\]
\item If $r=r_1/r_2$ then $Q_r(?x,?y)$ is defined as
\[
\pi_{\{?x,?y\}}\big( Q_{r_1}(?x,?z) \OpAND Q_{r_2}(?z,?y)\big).
\]
\item If $r=r_1|r_2$ then $Q_r(?x,?y)$ is defined as
\[
\big( Q_{r_1}(?x,?y) \OpUNION Q_{r_2}(?x,?y) \big).
\]
\item If $r=r_1^*$ then $Q_r(?x,?y)$ is defined as follows. 
%
%
%
First consider the LDQL query
\[
Q_\varepsilon(?x,?y) = \pi_{\{?x,?y\}} ( \OpSEED ?f\ \tuple{\varepsilon, P} )
\]
where $P$ is the following pattern
\begin{align*}
P =  & \big( (?x,?p,?o) \OpAND (?y,?p,?o) \OpFILTER (?x = ?y) \big)   \OpUNION \\
   & \big( (?s,?x,?o) \OpAND (?s,?y,?o) \OpFILTER (?x = ?y) \big)   \OpUNION \\
   & \big( (?s,?p,?x) \OpAND (?s,?p,?y) \OpFILTER (?x = ?y) \big)
\end{align*}

Now consider the LDQL query $Q_s(?v)$
defined as
\[
Q_s(?v) = \big( \langle \varepsilon, (\OpGRAPH ?u\ \{\ \}) \rangle \OpAND Q_{r_1}(?u,?v) \big).
\]
Then, query $Q_r(?x,?y)$ is defined by
\[
Q_\varepsilon(?x,?y) \OpUNION \\
\big( ( \OpSEED ?x\ \langle (?v, Q_s(?v))^*, (\OpGRAPH ?z\ \{\ \}) \rangle ) \OpAND Q_{r_1}(?z,?y) \big)
\]
\end{itemize}
We prove now that for every property path pattern $(?x,r,?y)$ we have that
\[
\fctEvalSetQcW{(?x,r,?y)}{\mathsf{ctxt}}{\symWoD} = \fctEvalQW{Q_r(?x,?y)}{\symWoD}^{\emptyset}.
\]
The proof is by induction in the construction of $Q_r(?x,?y)$. We proceed by cases.

\begin{itemize}
\item Assume that $r\in \symAllURIs$. Then $\mu\in \fctEvalQW{Q_r(?x,?y)}{\symWoD}^{\emptyset}$
if and only if $$\mu\in \fctEvalQW{(\OpSEED ?x\ \langle \varepsilon, (?x,r,?y)\rangle)}{\symWoD}^{\emptyset}.$$
Notice that this occurs if and only if 
there exists a mapping $\mu'$ and a URI $u$ such that $\mu' \in \fctEvalQW{\langle \varepsilon, (?x,r,?y)\rangle}{\symWoD}^{\{u\}}$, $\mu'$ is compatible with the mapping $\{?x\to u\}$, and $\mu=\mu'\cup\{?x\to u\}$.
Now, given that $\fctEvalPathPdW{\varepsilon}{\{u\}}{W}=\{u\}$, we have that 
$\mu' \in \fctEvalQW{\langle \varepsilon, (?x,r,?y)\rangle}{\symWoD}^{\{u\}}$ if and only if 
$\mu'\in \fctEvalPGD{(?x,r,?y)}{}{\mathcal{D}}$ with $\mathcal{D}$ the data set
with $\fctData{\fctADoc{u}}$ as default graph.
With all this we have that $\mu\in \fctEvalQW{Q_r(?x,?y)}{\symWoD}^{\emptyset}$
if and only if $\fctDom{\mu}=\{?x,?y\}$, $\mu(?x)\in \fctDom{\fctsymADoc}$, 
and $(\mu(?x),r,\mu(?y))\in \fctData{\fctADoc{\mu(?x)}}$,
which is exactly the property
\[
\mu((?x,r,?y))\in C^W(\mu(?x)).
\]
This last property holds if and only if $\mu\in \fctEvalSetQcW{(?x,r,?y)}{\mathsf{ctxt}}{\symWoD}$.

\item For the case in which 
$r=\ !(u_1\mid \cdots \mid u_k)$ with $u_i\in \symAllURIs$, the proof is similar.
We have that 
\[
\mu\in \fctEvalQW{\bigg( \OpSEED ?x\ \big\langle \varepsilon, 
\big( (?x,?p,?y) \OpFILTER (?p\neq u_1 \wedge \cdots \wedge ?p\neq u_k)\big) \big\rangle\bigg)}{\symWoD}^{\emptyset}
\]
if and only if $\mu$ is in the evaluation of
\[
\big( (?x,?p,?y) \OpFILTER (?p\neq u_1 \wedge \cdots \wedge ?p\neq u_k)\big) 
\]
over the graph $\fctData{\mu(?x)}$.
This happens if and only if
\[
 \mu((?x,p,?y))\in C^W(\mu(?x)) \text{for }p\notin \{u_1,\ldots,u_k\},
\]
which is exactly the property
\[
\mu\in \fctEvalSetQcW{(?x,!(u_1\mid \cdots \mid u_k),?y)}{\mathsf{ctxt}}{\symWoD}.
\]
\item For the cases $r=r_1/r_2$, $r=r_1| r_2$, the semantics of
the corresponding LDQL query exactly matches the semantics of the property path expression. 
Just notice that the semantics of $\OpAND$ is that of the join, and the semantics of
$\OpUNION$ is that of the set union.

\item For the case of $r=r_1^*$ 
we have that $\mu\in \fctEvalQW{(?x,r_1^*,?y)}{W}^{\mathsf{ctxt}}$
if and only if $\fctDom{\mu}=\{?x, ?y\}$ and
(i) $\mu(?x)=\mu(?y)$ and $\mu(?x),\mu(?y)\in \textit{terms}(W)$, or
(ii) $\mu\in \fctEvalQW{(?x,r_1^k,?y)}{W}^{\mathsf{ctxt}}$ for some $k>0$.
For the case $(i)$ it is easy to see that $\mu\in \fctEvalQW{Q_\varepsilon(?x,?y)}{W}^{\emptyset}$.
Just notice that if $\mu(?x)$ is in $\textit{terms}(W)$ then there exists a URI $u\in \fctDom{\fctsymADoc}$
and a triple $t$ in $\fctData{\fctADoc{u}}$ such that $\mu(?x)$ appears in $t$.
If $\mu(?x)$ appears in the subject position, then we know that $\mu$ is compatible
with a mapping in $\fctEvalQW{\big( (?x,?p,?o) \OpAND (?y,?p,?o) \OpFILTER (?x = ?y) \big)}{}^{\mathcal{D}}$
where $\mathcal D$ is a dataset with a default graph in which $t$ appears.
Finally, given that $Q_\varepsilon(?x,?y)=\pi_{\{?x,?y\}}(\OpSEED ?f\ \tuple{\varepsilon, P})$ 
and we know that $u$ is a possible value for variable $?f$, we obtain that 
$\mu\in \fctEvalQW{Q_\varepsilon(?x,?y)}{W}^{\emptyset}$.
If $\mu(?x)$ appears in the predicate or object position, the proof is similar.
For the case (ii) we will show that 

{\small
\begin{multline}\label{eq:goal}
\mu\in \fctEvalQW{Q_\varepsilon(?x,?y)}{W}^{\emptyset} \cup \\ 
\bigcup_{i=0}^{k-1}
\fctEvalQW{\pi_{\{?x,?y\}}\big( ( \OpSEED ?x\ \langle (?v, Q_s(?v))^{i}, (\OpGRAPH ?z\ \{\ \}) \rangle ) \OpAND Q_{r_1}(?z,?y) \big)}{W}^{\emptyset}
\end{multline}}

We will use an inductive argument. Assume $k=1$, then 
$\mu\in \fctEvalQW{(?x,r_1,?y)}{W}^{\mathsf{ctxt}}$.
By the induction hypothesis on the construction of property paths, we have that
$\mu\in \fctEvalQW{Q_{r_1}(?x,?y)}{W}^{\emptyset}$.
Now, if $\mu(?x)\notin \fctDom{\fctsymADoc}$ then we have that $\mu(?x)=\mu(?y)$
and thus $\mu\in \fctEvalQW{Q_\varepsilon(?x,?y)}{W}^{\emptyset}$.
If $\mu(?x)\in \fctDom{\fctsymADoc}$,
then it is easy to see that
\[
\mu\in \fctEvalQW{\pi_{\{?x,?y\}}\big( ( \OpSEED ?x\ \langle \varepsilon, (\OpGRAPH ?z\ \{\ \}) \rangle ) \OpAND Q_{r_1}(?z,?y) \big)}{W}^{\emptyset}
\]
This is because $\{?x\to \mu(?x),?z\to \mu(?x)\}\in \fctEvalQW{( \OpSEED ?x\ \langle \varepsilon, (\OpGRAPH ?z\ \{\ \}) \rangle )}{W}^{\emptyset}$
which is compatible with $\mu$.
Now assume that $\mu\in \fctEvalQW{(?x,r_1^{k+1},?y)}{W}^{\mathsf{ctxt}}$,
thus
\[
\mu\in \pi_{\{?x,?y\}}(\fctEvalQW{(?x,r_1^{k},?z)}{W}^{\mathsf{ctxt}}\Join \fctEvalQW{(?z,r_1,?y)}{W}^{\mathsf{ctxt}}).
\]
By induction hypothesis we have that

{\small
\begin{multline*}
\mu\in \pi_{\{?x,?y\}}\bigg( \bigg[ \fctEvalQW{Q_\varepsilon(?x,?z)}{W}^{\emptyset} \cup \\ 
\bigcup_{i=0}^{k-1}
\fctEvalQW{\pi_{\{?x,?z\}}\big( ( \OpSEED ?x\ \langle (?v, Q_s(?v))^{i}, (\OpGRAPH ?u\ \{\ \}) \rangle ) \OpAND Q_{r_1}(?u,?z) \big)}{W}^{\emptyset}\bigg]\\  \Join \fctEvalQW{Q_{r_1}(?z,?y)}{W}^{\emptyset}\bigg).
\end{multline*}}

Then we know that there exists $j$ such that $0\leq j\leq k-1$ such that

{\small
\begin{multline*}
\mu\in \pi_{\{?x,?y\}}\bigg( \bigg[ \fctEvalQW{Q_\varepsilon(?x,?z)}{W}^{\emptyset} \cup \\ 
\fctEvalQW{\pi_{\{?x,?z\}}\big( ( \OpSEED ?x\ \langle (?v, Q_s(?v))^{j}, (\OpGRAPH ?u\ \{\ \}) \rangle ) \OpAND Q_{r_1}(?u,?z) \big)}{W}^{\emptyset}\bigg]\\  \Join \fctEvalQW{Q_{r_1}(?z,?y)}{W}^{\emptyset}\bigg).
\end{multline*}}

If $\mu\in \pi_{\{?x,?y\}}\big( \fctEvalQW{Q_\varepsilon(?x,?z)}{W}^{\emptyset} \Join \fctEvalQW{Q_{r_1}(?z,?y)}{W}^{\emptyset}\big)= \fctEvalQW{Q_{r_1}(?x,?y)}{W}^{\emptyset}$ we can apply the same argument
as in the base case. Now assume that 

{\small
\begin{multline*}
\mu\in \pi_{\{?x,?y\}}\bigg(  \\ 
\fctEvalQW{\pi_{\{?x,?z\}}\big( ( \OpSEED ?x\ \langle (?v, Q_s(?v))^{j}, (\OpGRAPH ?u\ \{\ \}) \rangle ) \OpAND Q_{r_1}(?u,?z) \big)}{W}^{\emptyset}\\  \Join \fctEvalQW{Q_{r_1}(?z,?y)}{W}^{\emptyset}\bigg).
\end{multline*}}

Then we know that there exists a mapping $\mu'$ and $\mu''$ such that 
\[
\mu'\in \fctEvalQW{\pi_{\{?x,?z\}}\big( ( \OpSEED ?x\ \langle (?v, Q_s(?v))^{j}, (\OpGRAPH ?u\ \{\ \}) \rangle ) \OpAND Q_{r_1}(?u,?z)\big)}{W}^{\emptyset}
\]
and 
\[
\mu''\in \fctEvalQW{Q_{r_1}(?z,?y)}{W}^{\emptyset}
\]
$\mu$ equals $\mu'\cup\mu''$ restricted to variables $?x,?y$.
Notice that $\mu''$ is compatible with $\mu'$
thus, we have that $\mu'(?z)=\mu''(?z)$.
Now, if $\mu''(?z)\notin \fctDom{\fctsymADoc}$,
since $\mu''\in \fctEvalQW{Q_{r_1}(?z,?y)}{W}^{\emptyset}$,
then necessarily $\mu''(?z)=\mu''(?y)$,
and given that $\mu'$ is compatible with $\mu''$ we obtain that
$\mu'(?z)=\mu''(?y)$. All this implies that 

{\small
\begin{multline*}
\mu\in \pi_{\{?x,?y\}}\bigg(  \\ 
\fctEvalQW{\pi_{\{?x,?z\}}\big( ( \OpSEED ?x\ \langle (?v, Q_s(?v))^{j}, (\OpGRAPH ?u\ \{\ \}) \rangle ) \OpAND Q_{r_1}(?u,?y) \big)}{W}^{\emptyset}\bigg).
\end{multline*}}

and thus \eqref{eq:goal} holds.
Assume now that $\mu''(?z)\in \fctDom{\fctsymADoc}$.
We will prove that
\[
\mu'\in \fctEvalQW{ ( \OpSEED ?x\ \langle (?v, Q_s(?v))^{j+1}, (\OpGRAPH ?z\ \{\ \}) \rangle )}{W}^{\emptyset}.
\]
We know that 
\[
\mu'\in \fctEvalQW{\pi_{\{?x,?z\}}\big( ( \OpSEED ?x\ \langle (?v, Q_s(?v))^{j}, (\OpGRAPH ?u\ \{\ \}) \rangle ) \OpAND Q_{r_1}(?u,?z)\big)}{W}^{\emptyset}
\]
Thus $\mu'$ equals $\mu_1\cup\mu_2$ (restricted to variables $?x,?z$) where
\[
\mu_1\in \fctEvalQW{ ( \OpSEED ?x\ \langle (?v, Q_s(?v))^{j}, (\OpGRAPH ?u\ \{\ \}) \rangle ) }{W}^{\emptyset}
\]
and
\[
\mu_2\in \fctEvalQW{ Q_{r_1}(?u,?z) }{W}^{\emptyset}
\]
Thus, regarding $\mu_1$ we know that there exists a sequence of URIs, $u_1,u_2,\ldots u_j$ such
that $\mu_1(?x)=u_1$, $\mu_1(?u)=u_j$ and
$u_{i+1}\in \fctEvalQW{ (?v, Q_s(?v)) }{W}^{u_i}$.
Now, recall that the definition of $Q_s(?v)$ is 
\[
Q_s(?v) = \big( \langle \varepsilon, (\OpGRAPH ?f\ \{\ \}) \rangle \OpAND Q_{r_1}(?f,?v) \big).
\]
Then essentially what we have is that 
\[
\{?f\to u_i,?v\to u_{i+1}\}\in \fctEvalQW{ Q_{r_1}(?f,?v) }{W}^{\emptyset}.
\]
Moreover, since $\mu_1$ and $\mu_2$ are compatible, we know that $\mu_1(?u)=\mu_2(?u)=u_j$
and since $\mu_2\in \fctEvalQW{ Q_{r_1}(?u,?z) }{W}^{\emptyset}$
we know that 
\[
\{?f\to u_j, ?v\to \mu_2(?z)\}\in \fctEvalQW{ Q_{r_1}(?f,?v) }{W}^{\emptyset}.
\]
Finally, given that we are assuming that $\mu''(?z)=\mu_2(?z)$ is in $\fctDom{\fctsymADoc}$
we have that
\[
\{?x\to \mu_1(?x),?z\to \mu_2(?z)\}\in  \fctEvalQW{ ( \OpSEED ?x\ \langle (?v, Q_s(?v))^{j+1}, (\OpGRAPH ?z\ \{\ \}) \rangle )}{W}^{\emptyset}
\]
which is what we wanted to prove.
Thus we have that 
\[
\mu'\in \fctEvalQW{ ( \OpSEED ?x\ \langle (?v, Q_s(?v))^{j+1}, (\OpGRAPH ?z\ \{\ \}) \rangle )}{W}^{\emptyset}.
\]
and also that 
\[
\mu''\in \fctEvalQW{Q_{r_1}(?z,?y)}{W}^{\emptyset}
\]
and given that $\mu$ equals $\mu'\cup\mu''$ restricted to variables $?x,?y$, we have that

{\small
\begin{multline*}
\mu\in 
\fctEvalQW{\pi_{\{?x,?y\}}\big( ( \OpSEED ?x\ \langle (?v, Q_s(?v))^{j+1}, (\OpGRAPH ?z\ \{\ \}) \rangle ) \OpAND Q_{r_1}(?z,?y) \big)}{W}^{\emptyset}
\end{multline*}}

and since $j+1\leq k$ we obtain

{\small
\begin{multline*}
\mu\in \fctEvalQW{Q_\varepsilon(?x,?y)}{W}^{\emptyset} \cup \\ 
\bigcup_{i=0}^{k}
\fctEvalQW{\pi_{\{?x,?y\}}\big( ( \OpSEED ?x\ \langle (?v, Q_s(?v))^{i}, (\OpGRAPH ?z\ \{\ \}) \rangle ) \OpAND Q_{r_1}(?z,?y) \big)}{W}^{\emptyset}
\end{multline*}}


If one assumes that $\mu\in \fctEvalQW{Q_{r}(?x,?y)}{W}^{\mathsf{ctxt}}$ then by an argument on exactly
the same lines of the argument above, one can show that $\mu\in\fctEvalQW{(?x,r_1^*,?y)}{W}^{\mathsf{ctxt}}$.
\end{itemize}

We have shown how to construct an equivalent LDQL query for every property path pattern with two variables.
If the triple does not have two variables, we need a slightly different construction, in particular
for the case in which $(\cdot)^*$ is used. We now show the details of the construction but leave the
complete proof as an exercise (it can be completed using the arguments of the previous part of this proof).

Consider a propery path pattern $(\alpha,r,\beta)$ where $\alpha$ is a URI or variable,
and $\beta$ is a URI, variable or literal.
Then for the cases $r=p\in \symAllURIs$, $r=!(u_1|\cdots|u_k)$, $r=r_1/r_2$, $r=r_1|r_2$, we construct
a query as $Q_r(\alpha,\beta)$ where $Q_r(\alpha,\beta)$ is query $Q_r(?x,?y)$ where all occurrences of $?x$ has 
been replaced by $\alpha$ and all occurrences of $?y$ has been replaced by $\beta$.
For the case of $r=r_1^*$ we need to do a slightly different construction.
For a pattern $(u,r,?y)$ we construct a query $P_r(?y)$ as 
\[
\tuple{\varepsilon, \OpBIND(u \OpAS ?y)} \OpUNION
\big( ( \OpSEED \{u\}\ \langle (?v, Q_s(?v))^*, (\OpGRAPH ?z\ \{\ \}) \rangle ) \OpAND Q_{r_1}(?z,?y) \big)
\]
For a pattern $(?x,r,v)$ we construct a query $S_r^v(?x)$ as 
\[
\tuple{\varepsilon, \OpBIND(u \OpAS ?x)} \OpUNION
\big( ( \OpSEED ?x\ \langle (?v, Q_s(?v))^*, (\OpGRAPH ?z\ \{\ \}) \rangle ) \OpAND T(?z) \big)
\]
where $T(?z)$ is either $Q_{r_1}(?z,v)$ or $S_{r_1}^v(?z)$ depending on the form of $r_1$.
Finally, for a pattern $(u,r,v)$ we construct a query $U_r$ as 
\begin{multline*}
\tuple{\varepsilon, (\OpBIND(u \OpAS ?x) \OpAND \OpBIND(v \OpAS ?y)) \OpFILTER (?x = ?y)} \OpUNION \\
\big( ( \OpSEED \{u\}\ \langle (?v, Q_s(?v))^*, (\OpGRAPH ?z\ \{\ \}) \rangle ) \OpAND T(?z) \big)
\end{multline*}
where $T(?z)$ is either $Q_{r_1}(?z,v)$ or $S_{r_1}^v(?z)$ depending on the form of $r_1$.

Finally consider a property path pattern $(\ell,r,\beta)$, where $\ell$ is a literal.
Then for the cases $r=p\in \symAllURIs$, $r=!(u_1|\cdots|u_k)$ we should translate it
into an unsatisfiable query. One way of obtaining that query is, for example, with
an expression 
\[
\tuple{\varepsilon, (\OpBIND(\ell \OpAS ?x) \OpAND \OpBIND(\ell \OpAS ?y)) \OpFILTER (?x \neq ?y)}
\]
For the cases $r=r_1/r_2$ and $r=r_1|r_2$ we follow the same construction as if $\ell$ were a URI
but with the last base case.
For the case of $r=r_1^*$, if $\beta$ is a variable $y$ we consider the following query 
\[
\tuple{\varepsilon, \OpBIND(\ell \OpAS ?y)}.
\]
and if $\beta$ is a URI or literal the query
\[
\tuple{\varepsilon, (\OpBIND(\ell \OpAS ?x) \OpAND \OpBIND(\beta \OpAS ?y)) \OpFILTER (?x = ?y)}.
\]
The correctness of this translation can be proved along the same lines as for the case 
of property path pattern $(?x,r,?y)$.

%
%
%
%
%

\subsection{Proof of Theorem~\ref{theo:LDQLCoversNautiLOD}}

We proceed by induction showing how to translate every posible NautiLOD query.
The translation works in two parts.
We first define the following function $\trans_N(\cdot)$ that given a NautiLOD query, produces an LPE.
\begin{align*}
\trans_N(p) & \definedAs \tuple{\peContextURI,p,\peWildcard} \\
\trans_N(p\text{\textasciicircum}) & \definedAs \tuple{\peWildcard,p,\peContextURI} \\
\trans_N(\tuple{\peWildcard}) & \definedAs \big(?x, \big\langle \varepsilon, (\OpGRAPH ?u\ (?u,?p,?x)) \big\rangle\big) \\
\trans_N(n_1/n_2) & \definedAs \trans_N(n_1)/\trans_N(n_2) \\
\trans_N(n_1|n_2) & \definedAs \trans_N(n_1)|\trans_N(n_2) \\
\trans_N(n^*) & \definedAs \trans_N(n)^* \\
\trans_N(n[(\OpASK P)]) & \definedAs \trans_N(n)/[(?x, \langle \varepsilon, (\OpGRAPH ?x\ P)\rangle] \\
\end{align*}
Before presenting the complete translations, we prove the following result.
Let $n$ be a NautiLOD expression, then for every Web of Linked Data and URIs $u,v\in \fctDom{\fctsymADoc}$
we have that 
\[
v\in \fctEvalQW{n}{\symWoD}^u\text{\;\; if and only if \;\;} v\in \fctEvalQW{\trans_N(n)}{\symWoD}^{\{u\}}.
\]
The proof is by induction in the construction of the NautiLOD expression.
\begin{itemize}
\item for the case of $p\in \symAllURIs$ we have that
\[
\fctEvalQW{p}{\symWoD}^u=\{u'\mid (u,p,u')\in \fctData{\fctADoc{u}}\}
\]
notice that $v\in \fctDom{\fctsymADoc}$ and $v\in \fctEvalQW{p}{\symWoD}^u$, if and only if
there is a link from document $\fctADoc{u}$ to document $\fctADoc{v}$ that matches $\tuple{\peContextURI,p,\peWildcard}$.
This happens, if and only if $v\in \fctEvalQW{\tuple{\peContextURI,p,\peWildcard}}{\symWoD}^{\{u\}}$, which is what
we wanted to prove.
\item the case for $p\text{\textasciicircum}$ is similar but using $\tuple{\peContextURI,p,\peWildcard}$.
\item the case for $\tuple{\peWildcard}$. Just notice that 
$v\in \fctEvalQW{\tuple{\peWildcard}}{\symWoD}^u$ if and only if there exists a $p\in \symAllURIs$ such that
$(u,p,v)\in \fctData{\fctADoc{u}}$.
On the other hand we have that 
$v\in \fctEvalQW{\trans_N(\tuple{\peWildcard})}{\symWoD}^{\{u\}} = \fctEvalQW{\big(?x, \big\langle \varepsilon, (\OpGRAPH ?u\ (?u,?p,?x)) \big\rangle\big)}{\symWoD}^{\{u\}}$
if and only if 
$v\in \fctEvalQW{\pi_{?x}(\OpGRAPH ?u\ (?u,?p,?x))}{}^{\cal D}$
where $\mathcal{D}=\{\fctData{\fctADoc{u}}, \tuple{u,\fctData{\fctADoc{u}}}\}$.
Thus $v\in \fctEvalQW{\trans_N(\tuple{\peWildcard})}{\symWoD}^{\{u\}}$ if and only if there exists $p$ such that
$(u,p,v)\in \fctData{\fctADoc{u}}$. This proves the desired property.
\item for the case of an expression $n_1/n_2$, we have that 
$v$ in $\fctDom{\fctsymADoc}$ is in $\fctEvalQW{n_1/n_2}{\symWoD}^u$ if and only if, there exists 
$v'\in \fctDom{\fctsymADoc}$ such that
$v'\in \fctEvalQW{n_1}{\symWoD}^u$ and $v\in \fctEvalQW{n_2}{\symWoD}^{v'}$.
The we can apply or induction hypothesis and we have that
$v\in \fctEvalQW{n_1/n_2}{\symWoD}^u$ if and only if 
$v'\in \fctEvalQW{\trans_N(n_1)}{\symWoD}^{\{u\}}$ and $v\in \fctEvalQW{\trans_N(n_2)}{\symWoD}^{\{v'\}}$,
and thus $v\in \fctEvalQW{\trans_N(n_1/n_2)}{\symWoD}^{\{u\}}$.
\item cases $n_1|n_2$ and $n^*$ are direct from the definition of NautiLOD and LDQL.
\item for the case of expression $n[(\OpASK P)]$ we have that $v\in \fctEvalQW{n[(\OpASK P)]}{\symWoD}^u$
if and only if $v\in \fctEvalQW{n}{\symWoD}^u$, $v \in\fctDom{\fctsymADoc}$ and
$\fctEvalQW{P}{\fctData{\fctADoc{v}}}\neq\emptyset$.
On the other hand, we have that $v\in \fctEvalQW{\trans_N(n[(\OpASK P)])}{\symWoD}^{\{u\}}$ if and only if
\[
v\in \fctEvalQW{\trans_N(n)/[(?x, \langle \varepsilon, (\OpGRAPH ?x\ P)\rangle]}{\symWoD}^{\{u\}}.
\]
This happens if and only if there exists a $v'$ such that 
$v'\in \fctEvalQW{\trans_N(n)}{\symWoD}^{\{u\}}$ and 
$v\in \fctEvalQW{[(?x, \langle \varepsilon, (\OpGRAPH ?x\ P)\rangle]}{\symWoD}^{\{v'\}}$.
From the last property and the semantics of $[\cdot ]$ in LDQL, 
we have that $v=v'$ and that $\fctEvalQW{(?x, \langle \varepsilon, (\OpGRAPH ?x\ P)\rangle}{\symWoD}^{\{v\}}\neq \emptyset$.
The last holds if and only if $\fctEvalQW{\pi_{?x}(\OpGRAPH ?x\ P)}{}^{\mathcal{D}}\neq \emptyset$,
with $\mathcal{D}$ the RDF dataset $\{\fctData{\fctADoc{v}}, \tuple{v,\fctData{\fctADoc{v}}}\}$.
Thus we have that $v\in \fctEvalQW{\trans_N(n[(\OpASK P)])}{\symWoD}^{\{u\}}$ if and only if
$v\in \fctEvalQW{\trans_N(n)}{\symWoD}^{\{u\}}$ and $\fctEvalQW{P}{\fctData{\fctADoc{u}}}\neq\emptyset$.
Applying our induction hypothesis we have 
$v\in \fctEvalQW{n}{\symWoD}^u$ and $\fctEvalQW{P}{\fctData{\fctADoc{u}}}\neq\emptyset$, which is
exactly what we needed to prove.
\end{itemize}
Notice that the hypothesis that $v\in \fctDom{\fctsymADoc}$ was fundamental to prove the previous result.
Nevertheless, the output of a NautiLOD query can be a URI not in $\fctDom{\fctsymADoc}$ or even a literal,
so we need to do a different translation in general. Thus, we use now $\trans_N(\cdot)$ 
to translate a general NautiLOD query. 
Given a NautiLOD expression $n$ we have two cases.
Assume first that $n$, as a regular expression, does not produce the empty string $\varepsilon$.
Then, by using reglar language results, we know that we can write an equivalent expression $n'$
of the form
\[
n_1/e_1 \mid \cdots \mid n_k/e_k \mid m_1[(\OpASK P_1)] \mid \cdots \mid m_\ell[(\OpASK P_\ell)]
\]
where every $n_i$ and $m_j$ is a NautiLOD query, and every $e_i$ is either of the form $p$, or $p\text{\textasciicircum}$,
or $\tuple{\peWildcard}$.
We are ready now to produce an LDQL query $Q_n(?x)$ which is equivalent to $n$. The query is constructed as follows.
\begin{multline*}
Q_n(?x) = \pi_{\{?x\}}\bigg( 
\langle \trans_N(n_1), Q_1\rangle \OpUNION \cdots \OpUNION \langle \trans_N(n_k), Q_k\rangle \OpUNION \\
\langle \trans_N(m_1), (\OpGRAPH ?x\ P_1)\rangle \OpUNION \cdots \OpUNION 
\langle \trans_N(m_\ell), (\OpGRAPH ?x \ P_\ell)\rangle \bigg)
\end{multline*}
where query $Q_i$ depends on the form of $e_i$:
\begin{itemize}
\item if $e_i=p$ then $Q_i=(\OpGRAPH ?u\ (?u,p,?x))$
\item if $e_i=p\text{\textasciicircum}$ then $Q_i=(\OpGRAPH ?u\ (?x,p,?u))$
\item if $e_i=\tuple{\peWildcard}$ then $Q_i=(\OpGRAPH ?u\ (?u,?p,?x))$
\end{itemize}
Now to prove the correctness of our construction, assume that $v\in \fctEvalQW{n}{\symWoD}^u$. 
Then we know that $v\in \fctEvalQW{n_i/e_i}{\symWoD}^u$ or $v\in \fctEvalQW{m_i[(\OpASK P_i)]}{\symWoD}^u$ 
for some $i$.
If $v\in \fctEvalQW{n_i/e_i}{\symWoD}^u$ we know that there exists a $v'$ such that
$v' \in \fctEvalQW{n_i}{\symWoD}^u$ and $v\in \fctEvalQW{e_i}{\symWoD}^{v'}$.
Notice that, since $v\in \fctEvalQW{e_i}{\symWoD}^{v'}$, and $e_i$ is either $p$, or $p\text{\textasciicircum}$,
or $\tuple{\peWildcard}$ then we know that $v'$ is in $\fctDom{\fctsymADoc}$.
Thus we can apply our previous result to conclude from $v' \in \fctEvalQW{n_i}{\symWoD}^u$
that $v'\in \fctEvalQW{\trans_N(n_i)}{\symWoD}^{\{u\}}$.
Now if $e_i=p$ then from $v\in \fctEvalQW{e_i}{\symWoD}^{v'}$ we conclude that $(v',p,v)\in \fctData{\fctADoc{v'}}$
and thus $\fctEvalQW{(?u,p,?x)}{\fctData{\fctADoc{v'}}}$ contains the mapping $\mu=\{?u\to v',?x\to v\}$,
then $\fctEvalQW{(\OpGRAPH ?u\ (?u,p,?x))}{}^{\mathcal D}$ has $\mu$ as solution, with
$\mathcal{D}=\{\fctData{\fctADoc{v'}},\tuple{v',\fctData{\fctADoc{v'}}}\}$.
Given that $v'\in \fctEvalQW{\trans_N(n_i)}{\symWoD}^{\{u\}}$, we have that 
\[
\mu=\{?u\to v',?x\to v\}\in \fctEvalQW{\langle\trans_N(n_i), Q_i \rangle}{\symWoD}^{\{u\}}.
\]
Finally, given that $Q_n(?x)$ only keep the $?x$ variable, we have that $\{?x\to v\}$ is in
$\fctEvalQW{Q_n(?x)}{\symWoD}^{\{u\}}$, which is what we wanted to show.
If $e_i=p\text{\textasciicircum}$ or $e_i=\tuple{\peWildcard}$ the proof is the essentially the same.

Now assume that $v\in \fctEvalQW{m_i[(\OpASK P_i)]}{\symWoD}^u$. 
This implies that $v$ is in $\fctEvalQW{m_i}{\symWoD}^u$ and that $\fctEvalQW{P_i}{\fctData{\fctADoc{v}}}\neq\emptyset$.
By the semantics of NautiLOD, we have that $v$ is in $\fctDom{\fctsymADoc}$ (otherwise we could
not have been able to evaluate $P$), and thus we can apply our result above to obtain that
$v\in \fctEvalQW{\trans_N(m_i)}{\symWoD}^{\{u\}}$.
Now, given that $\fctEvalQW{P_i}{\fctData{\fctADoc{v}}}\neq\emptyset$ we have that
$\fctEvalQW{(\OpGRAPH ?x \ P_i)}{}^{\mathcal{D}}\neq\emptyset$ where $\mathcal{D}=\{\fctData{\fctADoc{v}}, \tuple{v, \fctData{\fctADoc{v}}}\}$.
Moreover, we have that every mapping $\mu$ in $\fctEvalQW{(\OpGRAPH ?x \ P_i)}{}^{\mathcal{D}}$
is such that $\mu(?x)=v$.
All these facts implies that mapping $\mu'=\{?x\to v\}$ is in
$\fctEvalQW{\langle \trans_N(m_\ell), (\OpGRAPH ?x \ P_\ell)\rangle}{\symWoD}^{\{u\}}$,
and thus $\mu'$ is in $\fctEvalQW{Q_n(?x)}{\symWoD}^{\{u\}}$
which is exactly what we wanted to prove.

If we start by assuming that $\mu=\{?x\to v\}$ is in $\fctEvalQW{Q_n(?x)}{\symWoD}^{\{u\}}$,
then following a similar reasoning as above one concludes that $v\in \fctEvalQW{n}{\symWoD}^{u}$.

To complete the proof we have to cover the case in which $n$, as a regular expression, can produce the
empty string.
Then, by applying some classical regular languages properties, one can rewrite $n$ as $\varepsilon|n'$ with 
$n'$ an expression that does not produce the empty string $\varepsilon$. Thus we can translate $n$ into
the LDQL query
\[
\tuple{\varepsilon, (\OpGRAPH\ ?x\ \{\ \})} \OpUNION Q_{n'}(?x)
\]
Notice that for every $u\in \fctData{\fctADoc{v}}$ we have that $\fctEvalQW{\tuple{\varepsilon, (\OpGRAPH\ ?x\ \{\ \})}}{\symWoD}^{\{u\}}$ results in a single mapping $\mu=\{?x\to u\}$.

\subsection{Proof of Theorem~\ref{theo:LDQLMoreThanNautiLOD}}

Recall that NautiLOD can only express paths and no combination of those paths via SPARQL operators is allowed.
Thus, it is easy to prove that NautiLOD cannot express operators such as 
$\OpSEED$, $\OpAND$, $\OpUNION$ that are natively allowed
in LDQL. Thus to make a stronger claim, we will prove that there exists simple LDQL query not using the mentioned
operators, that cannot be expressed using NautiLOD. 
The proof is similar to the proof of Theorem~\ref{theo:LDQLMoreThanPP}.

Thus, consider the LDQL $Q(?x)$ query given by
\[
\big\langle \tuple{\peContextURI,p,\peWildcard}, (?x,?x,?x) \big \rangle
\]
with $p\in \symAllURIs$. 
Now assume that there exists a NautiLOD expression $n$ such that 
\[
\fctEvalSetQcW{n}{v}{\symWoD} = \fctEvalQW{Q(?x)}{\symWoD}^{\{v\}}
\]
for every Web of Linked Data $\symWoD$ and $v\in \fctDom{\fctsymADoc}$.
Let $u,u',a,b$ be different elements in $\symAllURIs$ that are not mentioned in $n$.
Consider now $\symWoD_1$ having only two documents $d_1=\{(u,p,u')\}$ and $d_2=\{(a,a,a)\}$ and such that
$\fctADoc{u}=d_1$ and $\fctADoc{u'}=d_2$.
Moreover, consider $\symWoD_2$ having also two documents $d_1=\{(u,p,u')\}$ and $d_3=\{(b,b,b)\}$
such that $\fctADoc{u}=d_1$ and $\fctADoc{u'}=d_3$.
First notice that 
\[
\fctEvalQW{Q(?x)}{\symWoD_1}^{\{u\}}=\{\{?x\to a\}\}\;\; \neq \;\;
\fctEvalQW{Q(?x)}{\symWoD_2}^{\{u\}}=\{\{?x\to b\}\}
\]
We now prove that $\fctEvalSetQcW{n}{u}{\symWoD_1}=\fctEvalSetQcW{n}{u}{\symWoD_2}$ which is a contradiction.
To prove this, we show that for every subexpression $e$ of $n$, and for every possible URI $v$, it holds that
$\fctEvalSetQcW{e}{v}{\symWoD_1}=\fctEvalSetQcW{e}{v}{\symWoD_2}$. 
First notice that $\symWoD_1$ and $\symWoD_2$ has only two URIs in $\fctDom{\fctsymADoc}$, namely, $u$ and $u'$,
thus, we only have to reason for the cases in which $v=u$ or $v=u'$.
We proceed by induction.
\begin{itemize}
\item Assume that $e=r\in \symAllURIs$. Given that in $\symWoD_1$ and $\symWoD_2$ the URI $u$ is associated with the same
document (document $d_1$), then $\fctEvalSetQcW{r}{u}{\symWoD_1}=\fctEvalSetQcW{r}{u}{\symWoD_2}$.
Moreover, given that $r\neq a$ and $r\neq b$ (recall that $n$ does not mention $a$ or $b$), we have that
$\fctEvalSetQcW{r}{u'}{\symWoD_1}=\fctEvalSetQcW{r}{u'}{\symWoD_2}=\emptyset$.
\item Assume that $e=r\text{\textasciicircum}$ with $r\in \symAllURIs$. 
Exactly the same argument as the above case applies.
\item Assume that $e=\tuple{\peWildcard}$. For the same reason as in the above two cases we have that 
$\fctEvalSetQcW{r}{u}{\symWoD_1}=\fctEvalSetQcW{r}{u}{\symWoD_2}$.
Now consider $\fctEvalSetQcW{\tuple{\peWildcard}}{u'}{\symWoD_1}$. Then we have that URI $v$ is in
$\fctEvalSetQcW{\tuple{\peWildcard}}{u'}{\symWoD_1}$ if and only if, there exists some $p$ such that
$(u',p,v)\in \fctData{\fctADoc{u'}}$, but the only triple in $\fctData{\fctADoc{u'}}$ is $(a,a,a)$
and since $a\neq u'$ we have that $\fctEvalSetQcW{\tuple{\peWildcard}}{u'}{\symWoD_1}=\emptyset$.
For a similar reason we obtain that $\fctEvalSetQcW{\tuple{\peWildcard}}{u'}{\symWoD_2}=\emptyset$,
completing this part of the proof.
\item The cases $e=r_1/r_2$, $e=r_1|r_2$ and $e=r^*$ follows from the base cases proved above.
\item Assume $e=r[(\OpASK P)]$.
By definition we have that 
\[
	\fctEvalQW{r\peTest{(\OpASK P)}}{\symWoD}^v =
	\{v'\mid v' \in \fctEvalQW{r}{\symWoD}^v,\ v'\in\fctDom{\fctsymADoc}\text{ and } \fctEvalQW{P}{\fctData{\fctADoc{v'}}}\neq \emptyset \}
\]
By induction hypothesis we have that $\fctEvalQW{r}{\symWoD_1}^v=\fctEvalQW{r}{\symWoD_2}^v$ for $v=u,u'$.
Thus we only need to prove that the evaluation of $P$ is always the same.
given that $\fctData{\fctADoc{u}}$ is the same document in $\symWoD_1$ and $\symWoD_2$, we have that for
$u$ the property holds.
Now consider $\fctEvalQW{P}{d_2}$ and $\fctEvalQW{P}{d_3}$ with $d_1=\{(a,a,a)\}$ and $d_2=\{(b,b,b)\}$.
Recall that $P$ does not mention $a$ or $b$, thus we have that if $\mu\in \fctEvalQW{P}{d_2}$ then the mapping $\mu'$
obtained from $\mu$ by replacing every occurrence of $a$ by $b$, is in $\fctEvalQW{P}{d_3}$, and vice versa.
Thus we have that $\fctEvalQW{P}{d_2}=\emptyset$ if and only if $\fctEvalQW{P}{d_3}=\emptyset$.
This proves that 
$\fctEvalQW{r[(\OpASK P)]}{\symWoD_1}^v=\fctEvalQW{r[(\OpASK P)]}{\symWoD_2}^v$ for $v=u,u'$.
\end{itemize}
We have finished the proof that $\fctEvalSetQcW{n}{u}{\symWoD_1}=\fctEvalSetQcW{n}{u}{\symWoD_2}$
thus contradicting the fact that $n$ is equivalent to $Q(?x)$.

\subsection{Proof of Theorem~\ref{thm:LDQLCoversReachSem}} \label{proof:thm:LDQLCoversReachSem}

Let $\symPattern$ be an arbitrary SPARQL graph pattern,
let $\symWoDTuple$ be an arbitrary \wold,
and let $\symSeedURIs$ be some finite set of URIs.
To prove the theorem we use the~(basic) \queries\
$\tuple{\symLPE^{\cAll},\symPattern}$,
$\tuple{\symLPE^{\cNone},\symPattern}$, and
$\tuple{\symLPE^{\cMatch},\symPattern}$,
with the following LPEs:

\medskip
\noindent
\begin{tabular}{rp{105mm}}
	$\symLPE^{\cAll}$ \ &
	is $\peKleene{ \tuple{\peWildcard,\peWildcard,\peWildcard} }$,
\\[1mm]
	$\symLPE^{\cNone}$ \ &
	is $\peEmpty$, and
\\[1mm]
	$\symLPE^{\cMatch}$ \ &
	is $\peKleene{ \bigl(
		\peSubquery{?s}{\symQuery_1} \,\vert\,
		\peSubquery{?p}{\symQuery_1} \,\vert\,
		\peSubquery{?o}{\symQuery_1} \,\vert\,
		\ldots \,\vert\,
		\peSubquery{?s}{\symQuery_m} \,\vert\,
		\peSubquery{?p}{\symQuery_m} \,\vert\,
		\peSubquery{?o}{\symQuery_m}
	\bigr) }$
	where
	$?s$, $?p$ and $?o$ are fresh variables~(not used in $\symPattern$),
	$m$ is the number of triple patterns in $\symPattern$,
	and for each such triple pattern $\symTP_k$~($1 \leq k \leq m$) there exists a subquery $\symQuery_k$ of the form $\tuple{\peEmpty,\symPattern_k}$ with a SPARQL pattern $\symPattern_k$ that is constructed as follows:
	$\symPattern_k$ contains the triple pattern $\tuple{?s,?p,?o}$ and---depending on the form of \removable{the corresponding triple pattern} $\symTP_k = \tuple{s_k,p_k,o_k}$---may contain additional$\OpFILTER$\! operators;
	in particular,
	if $s_k \notin \symAllVariables$, then $\symPattern_k$ contains$\OpFILTER ?s=s_k$; 
	if $p_k \notin \symAllVariables$, then $\symPattern_k$ contains$\OpFILTER ?p=p_k$; and
	if $o_k \notin \symAllVariables$, then $\symPattern_k$ contains$\OpFILTER ?o=o_k$.
\end{tabular}
\medskip

Then, for each reachability criterion $\symReachCrit \in \lbrace \cAll,\cNone,\cMatch \rbrace$ with its corresponding LPE $\symLPE^{\symReachCrit}$ as specified above, we have to show the following equivalence:

\begin{equation} \label{eq:proof:LDQLCoversReachSem}
	\fctEvalReachPcSW{\symPattern}{\symReachCrit}{\symSeedURIs}{\symWoD}\! = \fctEvalQWS{ \tuple{\symLPE^{\symReachCrit}\!,\symPattern} }{\symWoD}{\symSeedURIs}
	.
\end{equation}

By the definition of the reach\-abil\-i\-ty-based query semantics~(cf.~Section~\ref{ssec:Comparison:ReachSem}) and the definition of LDQL query semantics~(cf.~Definition~\ref{def:LDQLSemantics}), it is sufficient to prove the following lemma to show that (\ref{eq:proof:LDQLCoversReachSem}) holds for each $\symReachCrit \in \lbrace \cAll,\cNone,\cMatch \rbrace$.

\begin{lemma}\label{lem:LDQLCoversReachSem}
	For each $\symReachCrit \in \lbrace \cAll,\cNone,\cMatch \rbrace$, the set of all \docs\ that are ($\symReachCrit,\symSeedURIs,\symPattern$)-reach\-able in $\symWoD$ is equivalent to the following set of \docs:
	\begin{equation*}
		\symDocs_\mathsf{LPE}^{\symReachCrit} = \lbrace \fctADoc{\symURI} \mid \symURI \in \fctEvalPathPdW{\symLPE^{\symReachCrit}}{\symCtxURI}{\symWoD} \text{ for some } \symCtxURI \in \symSeedURIs \rbrace
		.
	\end{equation*}
\end{lemma}

Notice that for each $\symReachCrit \in \lbrace \cAll,\cNone,\cMatch \rbrace$, the set $\symDocs_\mathsf{LPE}^{\symReachCrit}$ is the set of \docs\ selected by evaluating $\symLPE^{\symReachCrit}$ over~$\symWoD$ using every URI in $\symSeedURIs$ as context URI.
In the following, we prove Lemma~\ref{lem:LDQLCoversReachSem} for each of the three reachability criteria, $\cAll$, $\cNone$, and~$\cMatch$.

\subsubsection{$\cAll$-semantics:}

To prove Lemma~\ref{lem:LDQLCoversReachSem} for $\cAll$ we show that the set $\symDocs_\mathsf{LPE}^{\cAll}$ is both a subset and a superset of the set of all ($\cAll,\symSeedURIs,\symPattern$)-reach\-able \docs\ in $\symWoD$\!.

We begin with the former.
Hence, for an arbitrary \doc\ in $\symDocs_\mathsf{LPE}^{\cAll}$ we have to show that this \doc\ is ($\cAll,\symSeedURIs,\symPattern$)-reach\-able in $\symWoD$\!.
Let $d_\mathsf{LPE} \in \symDocs_\mathsf{LPE}^{\cAll}$ be such a \doc.
Since $d_\mathsf{LPE} \in \symDocs_\mathsf{LPE}^{\cAll}$, we know that there exist two URIs, $\symCtxURI$ and $\symURI$, such~that

\begin{itemize}
	\item $\symCtxURI \in \symSeedURIs$,
	\item $\symURI \in \fctEvalPathPdW{\symLPE^{\cAll}}{\symCtxURI}{\symWoD}$, and
	\item $d_\mathsf{LPE} = \fctADoc{\symURI}$.
\end{itemize}

\noindent
Then, either we have $\symCtxURI = \symURI$ or $\symCtxURI \neq \symURI$.
	\removable{In the following, we discuss these two cases.}

If $\symCtxURI = \symURI$, then 
$d_\mathsf{LPE} = \fctADoc{\symCtxURI}$ and, thus, \doc\ $d_\mathsf{LPE}$ is ($\cAll,\symSeedURIs,\symPattern$)-reach\-able in~$\symWoD$ because it satisfies the first of the two alternative conditions for reachability as given in Section~\ref{ssec:Comparison:ReachSem}.

If $\symCtxURI \neq \symURI$, then,
given that $\symURI \in \fctEvalPathPdW{\symLPE^{\cAll}}{\symCtxURI}{\symWoD}$,
there exists a nonempty sequence of link graph~edges
\begin{align*}
	\tuple{d_1,(t_1,\symURI_1),d_1'} \in \mathcal{G}_\symWoD,&&
	\tuple{d_2,(t_2,\symURI_2),d_2'} \in \mathcal{G}_\symWoD,&&
	\ldots ,&&
	\tuple{d_n,(t_n,\symURI_n),d_n'} \in \mathcal{G}_\symWoD
\end{align*}
such that
\begin{itemize}
	\item $d_1 = \fctADoc{\symCtxURI}$,
	\item $d_i' = d_{i+1}$ for all $i \in \lbrace 1, \ldots , n-1 \rbrace$, and
	\item $d_n' = d_\mathsf{LPE}$~(and $\symURI_n = \symURI$).
\end{itemize}

\noindent
Then, since $d_1 = \fctADoc{\symCtxURI}$ and $\symCtxURI \in \symSeedURIs$, we have that \doc\ $d_1$ is ($\cAll,\symSeedURIs,\symPattern$)-reach\-able in~$\symWoD$~(the \doc\ satisfies the first of the two conditions for reachability as given in Section~\ref{ssec:Comparison:ReachSem}).
As a consequence, we can
	\removable{use the fact that $d_i' = d_{i+1}$ for all $i \in \lbrace 1, \ldots , n-1 \rbrace$ to}
show that all other \docs\ connected by the sequence of link graph edges are also ($\cAll,\symSeedURIs,\symPattern$)-reach\-able in~$\symWoD$~(they satisfy the second condition).
Therefore, due to $d_n' = d_\mathsf{LPE}$, \doc\ $d_\mathsf{LPE}$ is ($\cAll,\symSeedURIs,\symPattern$)-reach\-able in~$\symWoD$\!.

After showing that in both cases, $\symCtxURI = \symURI$ and $\symCtxURI \neq \symURI$, \doc\ $d_\mathsf{LPE} \in \symDocs_\mathsf{LPE}^{\cAll}$ is ($\cAll,\symSeedURIs,\symPattern$)-reach\-able in~$\symWoD$\!, we conclude that the set $\symDocs_\mathsf{LPE}^{\cAll}$ is a subset of the set of all ($\cAll,\symSeedURIs,\symPattern$)-reach\-able \docs\ in $\symWoD$\!. It remains to show that $\symDocs_\mathsf{LPE}^{\cAll}$ is also a superset.

To this end, let $d_\mathsf{R}$ be a \doc\ that is ($\cAll,\symSeedURIs,\symPattern$)-reach\-able in $\symWoD$\!. We have to show that $d_\mathsf{R}$ is in $\symDocs_\mathsf{LPE}^{\cAll}$.
We note that document $d_\mathsf{R}$ may be ($\cAll,\symSeedURIs,\symPattern$)-reach\-able in $\symWoD$ because it satisfies either the first or the second of the two alternative conditions for reachability as given in Section~\ref{ssec:Comparison:ReachSem}. In the following, we discuss both cases.

If $d_\mathsf{R}$ satisfies the first condition, there exists a URI $\symURI_\mathsf{R} \!\in\! \symSeedURIs$ such that $\fctADoc{\symURI_\mathsf{R}} = d_\mathsf{R}$.
Since $\symLPE^{\cAll}$ is $\peKleene{ \tuple{\peWildcard,\peWildcard,\peWildcard} }$, we also have $\symURI_\mathsf{R} \in \fctEvalPathPdW{\symLPE^{\cAll}}{\symURI_\mathsf{R}}{\symWoD}$.
Therefore, we can use URI~$\symURI_\mathsf{R}$ as both $\symCtxURI$ and $\symURI$ in the definition of $\symDocs_\mathsf{LPE}^{\cAll}$, which shows that $d_\mathsf{R} \in \symDocs_\mathsf{LPE}^{\cAll}$.

If $d_\mathsf{R}$ satisfies the second condition, then there exist both a seed URI $\symURI_0 \in \symSeedURIs$ and a nonempty sequence of link graph~edges
\begin{align*}
	\tuple{d_1,(t_1,\symURI_1),d_1'} \in \mathcal{G}_\symWoD,&&
	\tuple{d_2,(t_2,\symURI_2),d_2'} \in \mathcal{G}_\symWoD,&&
	\ldots ,&&
	\tuple{d_n,(t_n,\symURI_n),d_n'} \in \mathcal{G}_\symWoD
\end{align*}
such that
\begin{itemize}
	\item $d_1 = \fctADoc{\symURI_0}$,
	\item $d_i' = d_{i+1}$ for all $i \in \lbrace 1, \ldots , n-1 \rbrace$, and
	\item $d_n' = d_\mathsf{R}$ and, thus, $d_\mathsf{R} = \fctADoc{\symURI_n}$.
\end{itemize}

\noindent
Moreover, every such link graph edge $\tuple{d_j,(t_j,\symURI_j),d_j'}$ matches link pattern $\tuple{\peWildcard,\peWildcard,\peWildcard}$ in the context of URI $\symURI_{j-1}$~($1 \leq j \leq n$).
Therefore, since $\symLPE^{\cAll}$ is $\peKleene{ \tuple{\peWildcard,\peWildcard,\peWildcard} }$, we have $\symURI_n \in \fctEvalPathPdW{\symLPE^{\cAll}}{\symURI_0}{\symWoD}$.
Then,
	\removable{with $d_\mathsf{R} = \fctADoc{\symURI_n}$ and $\symURI_0 \in \symSeedURIs$,}
we can use $\symURI_n$ as $\symURI$ and $\symURI_0$ as $\symCtxURI$ in the definition of $\symDocs_\mathsf{LPE}^{\cAll}$, which shows that $d_\mathsf{R} \in \symDocs_\mathsf{LPE}^{\cAll}$.

In conclusion, independent of whether $d_\mathsf{R}$ satisfies the first or the second condition for being ($\cAll,\symSeedURIs,\symPattern$)-reach\-able in $\symWoD$\!, we find that $d_\mathsf{R} \in \symDocs_\mathsf{LPE}^{\cAll}$. Hence, $\symDocs_\mathsf{LPE}^{\cAll}$ is not only a subset of all ($\cAll,\symSeedURIs,\symPattern$)-reach\-able \docs\ in $\symWoD$\!, but also a superset thereof, which shows that both sets are equivalent~(as claimed in Lemma~\ref{lem:LDQLCoversReachSem}).

\subsubsection{$\cNone$-semantics:}

To prove Lemma~\ref{lem:LDQLCoversReachSem} for $\cNone$ we show that the set $\symDocs_\mathsf{LPE}^{\cNone}$ is both a subset and a superset of the set of all ($\cNone,\symSeedURIs,\symPattern$)-reach\-able \docs\ in $\symWoD$\!.

To begin with the former, assume an arbitrary \doc\ in $d_\mathsf{LPE} \in \symDocs_\mathsf{LPE}^{\cNone}$. We have to show that this \doc\ is ($\cNone,\symSeedURIs,\symPattern$)-reach\-able in $\symWoD$\!.
Since $d_\mathsf{LPE} \in \symDocs_\mathsf{LPE}^{\cAll}$, we know that there exist two URIs, $\symCtxURI$ and $\symURI$, such~that

\begin{itemize}
	\item $\symCtxURI \in \symSeedURIs$,
	\item $\symURI \in \fctEvalPathPdW{\symLPE^{\cNone}}{\symCtxURI}{\symWoD}$, and
	\item $d_\mathsf{LPE} = \fctADoc{\symURI}$.
\end{itemize}

\noindent
Given that $\symLPE^{\cNone}$ is $\peEmpty$, by $\symURI \in \fctEvalPathPdW{\symLPE^{\cNone}}{\symCtxURI}{\symWoD}$ and Definition~\ref{def:SemanticsLPEs}, we obtain that $\symURI = \symCtxURI$ and, thus, $d_\mathsf{LPE} = \fctADoc{\symCtxURI}$.
Therefore, \doc\ $d_\mathsf{LPE}$ is ($\cNone,\symSeedURIs,\symPattern$)-reach\-able in~$\symWoD$ because it satisfies the first of the two alternative conditions for reachability as given in Section~\ref{ssec:Comparison:ReachSem}.
As a consequence, we can conclude that the set $\symDocs_\mathsf{LPE}^{\cNone}$ is a subset of the set of all ($\cNone,\symSeedURIs,\symPattern$)-reach\-able \docs\ in $\symWoD$\!.

To show that $\symDocs_\mathsf{LPE}^{\cNone}$ is also a superset, let $d_\mathsf{R}$ be an arbitrary \doc\ that is ($\cNone,\symSeedURIs,\symPattern$)-reach\-able in $\symWoD$\!. We have to show that $d_\mathsf{R}$ is in $\symDocs_\mathsf{LPE}^{\cNone}$.
We note that $d_\mathsf{R}$ can be ($\cNone,\symSeedURIs,\symPattern$)-reach\-able in $\symWoD$ only if it satisfies the first of the two alternative conditions for reachability as given in Section~\ref{ssec:Comparison:ReachSem}~(for $\cNone$, the second condition cannot be satisfied by any \doc\ because 
$\cNone(t,\symURI,\symPattern) \definedAs \false$ for all $\tuple{t,\symURI,\symPattern} \in \symAllTriples \times \symAllURIs \times \symAllPatterns$).
Therefore, given that $d_\mathsf{R}$ satisfies the first condition, there exists a URI $\symURI_\mathsf{R} \in \symSeedURIs$ such that $\fctADoc{\symURI_\mathsf{R}} = d_\mathsf{R}$.
Since $\symLPE^{\cNone}$ is $\peEmpty$, we also have $\symURI_\mathsf{R} \in \fctEvalPathPdW{\symLPE^{\cNone}}{\symURI_\mathsf{R}}{\symWoD}$.
Therefore, we can use URI~$\symURI_\mathsf{R}$ as both $\symCtxURI$ and $\symURI$ in the definition of $\symDocs_\mathsf{LPE}^{\cNone}$ and, thus, obtain that $d_\mathsf{R} \in \symDocs_\mathsf{LPE}^{\cNone}$, which shows that the set $\symDocs_\mathsf{LPE}^{\cNone}$ is a superset of the set of all ($\cNone,\symSeedURIs,\symPattern$)-reach\-able \docs\ in $\symWoD$\!.
Since we have shown before that $\symDocs_\mathsf{LPE}^{\cNone}$ is also a subset
of the set of all \docs\ that are ($\cNone,\symSeedURIs,\symPattern$)-reach\-able in $\symWoD$\!, we conclude that both sets are equivalent. Hence, Lemma~\ref{lem:LDQLCoversReachSem} holds for reachability criterion $\cNone$.

\subsubsection{$\cMatch$-semantics:}

It remains to prove Lemma~\ref{lem:LDQLCoversReachSem} for $\cMatch$. To this end, we show that the set $\symDocs_\mathsf{LPE}^{\cMatch}$ is both a subset and a superset of the set of all \docs\ that are ($\cMatch,\symSeedURIs,\symPattern$)-reach\-able in $\symWoD$\!. As before, we begin with the former.

Let $d_\mathsf{LPE}$ be an arbitrary \doc\ in $\symDocs_\mathsf{LPE}^{\cAll}$. We have to show that this \doc\ is ($\cMatch,\symSeedURIs,\symPattern$)-reach\-able in $\symWoD$\!.
Since $d_\mathsf{LPE} \in \symDocs_\mathsf{LPE}^{\cMatch}$, we know by the definition of $\symDocs_\mathsf{LPE}^{\cMatch}$~(as given in Lemma~\ref{lem:LDQLCoversReachSem}) that there exist two URIs, $\symCtxURI$ and $\symURI$, such~that

\begin{itemize}
	\item $\symCtxURI \in \symSeedURIs$,
	\item $\symURI \in \fctEvalPathPdW{\symLPE^{\cMatch}}{\symCtxURI}{\symWoD}$, and
	\item $d_\mathsf{LPE} = \fctADoc{\symURI}$.
\end{itemize}

\noindent
Given that $\symURI \in \fctEvalPathPdW{\symLPE^{\cMatch}}{\symCtxURI}{\symWoD}$,
there exists a nonempty sequence of URIs
$\symURI_0, \symURI_1, \ldots, \symURI_n$
and a corresponding sequence of \docs\ $d_0, d_1, \ldots, d_n$
such~that
\begin{itemize}
	\item $d_i = \fctADoc{\symURI_i}$ for each $i \in \lbrace 0, \ldots , n \rbrace$,
	\item $\symURI_0 = \symCtxURI$,
	\item $\symURI_n = \symURI$~(and, thus, $d_n = d_\mathsf{LPE}$), and

	\item for each $i \in \lbrace 1, \ldots , n \rbrace$, there exists a triple pattern $\symTP_k$ in $\symPattern$~($1 \leq k \leq m$) such that $\symURI_i \in \fctEvalPathPdW{\peSubquery{?v}{\symQuery_k}}{\symURI_{i-1}}{\symWoD}$ where $?v \in \lbrace ?s,?p,?o \rbrace$ and $\symQuery_k$ is the \query\ that corresponds to $\symTP_k$ as specified in the definition of $\symLPE^{\cMatch}$ above.
\end{itemize}

We show by induction over $n$ that all $n+1$ \docs, $d_0, d_1, \ldots, d_n$, are ($\cMatch,\symSeedURIs,\symPattern$)-reach\-able in~$\symWoD$\!, and, thus, so is $d_\mathsf{LPE} = d_n$.

\medskip
\textit{Base case (n $=$ 0)}:
$d_0$ is ($\cMatch,\symSeedURIs,\symPattern$)-reach\-able in $\symWoD$ because $d_0 = \fctADoc{\symURI_0}$, $\symURI_0 = \symCtxURI$, and $\symCtxURI \in \symSeedURIs$; i.e., $d_0$ satisfies the first condition as specified in Section~\ref{ssec:Comparison:ReachSem}.

\medskip
\textit{Induction step (n $>$ 0)}:
By induction, we assume that \doc\ $d_{n\text{-}1}$ is ($\cMatch,\symSeedURIs,\symPattern$)-reach\-able in $\symWoD$\!.
To show that $d_n$ is also ($\cMatch,\symSeedURIs,\symPattern$)-reach\-able in $\symWoD$ we aim to show that $d_n$ satisfies the second
	condition
for reachability as given in Section~\ref{ssec:Comparison:ReachSem}. That is, we aim to show that there exists a link graph edge $\tuple{d_\mathsf{src},(t,\symURI),d_\mathsf{tgt}} \in \mathcal{G}_\symWoD$ such that
\enumA~$d_\mathsf{src}$ is $(\cMatch,\symSeedURIs,\symPattern)$-reach\-able in $\symWoD$\!,
\enumB~$\cMatch(t,\symURI,\symPattern)=\true$,
\enumC~$\symURI = \symURI_n$,
and \enumD~$d_\mathsf{tgt}=d_n$.
Let $d_\mathsf{src}$ be $d_{n\text{-}1}$, which is ($\cMatch,\symSeedURIs,\symPattern$)-reach\-able in $\symWoD$ by our inductive hypothesis. Hence, it remains to show the existence of a link graph edge $\tuple{d_{n\text{-}1},(t,\symURI_n),d_n} \in \mathcal{G}_\symWoD$ for which $\cMatch(t,\symURI_n,\symPattern)=\true$.
To this end, we use the fourth of the four aforementioned properties of the sequence of URIs $\symURI_0, \symURI_1, \ldots, \symURI_n$.

Let $\symTP_k = \tuple{s_k,p_k,o_k}$ be a triple pattern in $\symPattern$ such that $\symURI_n \in \fctEvalPathPdW{\peSubquery{?v}{\symQuery_k}}{\symURI_{n-1}}{\symWoD}$ where $?v \in \lbrace ?s,?p,?o \rbrace$  and $\symQuery_k$ is the \query\ that corresponds to $\symTP_k$ as specified in the definition of $\symLPE^{\cMatch}$ above; i.e., $\symQuery_k$ is a basic \query\ of the form $\tuple{\peEmpty,\symPattern_k}$ where SPARQL pattern $\symPattern_k$ contains the triple pattern $\tuple{?s,?p,?o}$ and
\enumA~if $s_k \notin \symAllVariables$, then $\symPattern_k$ contains$\OpFILTER ?s=s_k$,
\enumB~if $p_k \notin \symAllVariables$, then $\symPattern_k$ contains$\OpFILTER ?p=p_k$, and
\enumC~if $o_k \notin \symAllVariables$, then $\symPattern_k$ contains$\OpFILTER ?o=o_k$.

Since $\symURI_n \in \fctEvalPathPdW{\peSubquery{?v}{\symQuery_k}}{\symURI_{n-1}}{\symWoD}$\!, by Definition~\ref{def:SemanticsLPEs}, there exists a solution mapping~$\mu$ such that
$\mu \in \fctEvalQWS{\symQuery_k}{\symWoD}{\lbrace\symURI_{n-1}\rbrace}$
and $\mu(?v) = \symURI_n$.
Moreover, since $\symQuery_k$ is of the form $\tuple{\peEmpty,\symPattern_k}$, we have 
$\mu \in \fctEvalPGD{\symPattern_k}{\symRDFgraph}{\symRDFdataset}$ where $\textstyle\symRDFdataset = \dataset_\symWoD\left( \lbrace \symURI_{n-1} \rbrace \right)$ with default graph $G = \fctData{d_{n-1}}$.
Then, due to the construction of $\symPattern_k$, it is easily verified that there exists an RDF triple $t \in \fctData{d_{n-1}}$ such that $\mu[\symTP_k] = t$ and $\symURI_n \in \fctURIs{t}$. As a consequence,
\enumA~$\cMatch(t,\symURI_n,\symTP_k) = \true$ and
\enumB~by Definition~\ref{def:LinkGraph}, there exists a link graph edge $\tuple{d_{n-1},(t,\symURI_n),d_n} \in \mathcal{G}_\symWoD$.
Finally, since $\symTP_k$ is
	a triple pattern
in $\symPattern$, we also have $\cMatch(t,\symURI_n,\symPattern) = \true$.

\medskip
\noindent
While this concludes showing that the set $\symDocs_\mathsf{LPE}^{\cMatch}$ is a subset of the set of all
	\docs\ that are ($\cMatch,\symSeedURIs,\symPattern$)-reach\-able
in $\symWoD$\!, we now show that it is also a superset thereof.

Let $d_\mathsf{R}$ be a \doc\ that is ($\cMatch,\symSeedURIs,\symPattern$)-reach\-able in $\symWoD$\!. We have to show that $d_\mathsf{R}$ is in $\symDocs_\mathsf{LPE}^{\cMatch}$.
Since $d_\mathsf{R}$ is ($\cMatch,\symSeedURIs,\symPattern$)-reach\-able in $\symWoD$\!, there exist a nonempty sequence of URIs $\symURI_0, \symURI_1, \ldots, \symURI_n$,
a corresponding sequence of \docs
\begin{align*}
	d_0 = \fctADoc{\symURI_0},&&
	d_1 = \fctADoc{\symURI_1},&&
	d_2 = \fctADoc{\symURI_2},&&
	\ldots ,&&
	d_n = \fctADoc{\symURI_n},&&
\end{align*}
and a corresponding sequence of link graph~edges
\begin{align*}
	\tuple{d_1',(t_1,\symURI_1),d_1} \in \mathcal{G}_\symWoD,&&
	\tuple{d_2',(t_2,\symURI_2),d_2} \in \mathcal{G}_\symWoD,&&
	\ldots ,&&
	\tuple{d_n',(t_n,\symURI_n),d_n} \in \mathcal{G}_\symWoD
\end{align*}
such that
\begin{itemize}
	\item $\symURI_0 \in \symSeedURIs$,
	\item $\cMatch(t_i,\symURI_i,\symPattern) = \true$ for all $i \in \lbrace 1, \ldots , n \rbrace$,
	\item $d_i' = d_{i-1}$ for all $i \in \lbrace 1, \ldots , n \rbrace$, and
	\item $d_n = d_\mathsf{R}$ and, thus, $d_\mathsf{R} = \fctADoc{\symURI_n}$.
\end{itemize}

We aim to show that each of the $n+1$ \docs, $d_0, d_1, d_2, \ldots, d_n$, is in $\symDocs_\mathsf{LPE}^{\cMatch}$, and, thus, so is $d_\mathsf{R} = d_n$. To this end, it is sufficient to show that each of the $n+1$ URIs, $\symURI_0, \symURI_1, \ldots, \symURI_n$, is in $\fctEvalPathPdW{\symLPE^{\cMatch}}{\symURI_0}{\symWoD}$. 
Then,
	\removable{with $\symURI_0 \in \symSeedURIs$,}
for each $i \in \lbrace 0, \ldots , n \rbrace$ we can use URI $\symURI_i$ as $\symURI$ and $\symURI_0$ as $\symCtxURI$ in the definition of $\symDocs_\mathsf{LPE}^{\cMatch}$, which shows that \doc\ $d_i = \fctADoc{\symURI_i}$ is in $\symDocs_\mathsf{LPE}^{\cMatch}$.
We use proof by induction.

\medskip
\textit{Base case (n $=$ 0)}:
Since $\symLPE^{\cMatch}$ is of the form $\peKleene{(\cdot)}$\!, we have $\symURI_0 \in \fctEvalPathPdW{\symLPE^{\cMatch}}{\symURI_0}{\symWoD}$.

\medskip
\textit{Induction step (n $>$ 0)}:
By induction, we assume that $\symURI_{n-1} \in \fctEvalPathPdW{\symLPE^{\cMatch}}{\symURI_0}{\symWoD}$. Then, to show that $\symURI_n$ is also in $\fctEvalPathPdW{\symLPE^{\cMatch}}{\symURI_0}{\symWoD}$, it is sufficient to show that $\symURI_n$ is in $\fctEvalPathPdW{ \symLPE^{\cMatch}_\mathsf{step} }{\symURI_{n-1}}{\symWoD}$ where $\symLPE^{\cMatch}_\mathsf{step}$ is
$
	\peSubquery{?s}{\symQuery_1} \,\vert\,
	\peSubquery{?p}{\symQuery_1} \,\vert\,
	\peSubquery{?o}{\symQuery_1} \,\vert\,
	\ldots \,\vert\,
	\peSubquery{?s}{\symQuery_m} \,\vert\,
	\peSubquery{?p}{\symQuery_m} \,\vert\,
	\peSubquery{?o}{\symQuery_m}
$
such that $\symLPE^{\cMatch} = \peKleene{\left( \symLPE^{\cMatch}_\mathsf{step} \right)}$\!.
Due to the existence of
link graph edge $\tuple{d_n',(t_n,\symURI_n),d_n} \in \mathcal{G}_\symWoD$ with $d_n' = d_{n-1}$, we know by Definition~\ref{def:LinkGraph} that there exists
	a
triple $t_n \in \fctData{d_{n-1}}$ with $\symURI_n \in \fctURIs{t_n}$. Moreover, since $\cMatch(t_n,\symURI_n,\symPattern) = \true$, there exist both a triple pattern $\symTP_k$ in $\symPattern$ and a solution mapping $\mu$ such that $\mu[\symTP_k] = t_n$. Then, given the \query\ $\symQuery_k = \tuple{\peEmpty,\symPattern_k}$ that is constructed for $\symTP_k$ as specified in the definition of $\symLPE^{\cMatch}$, it easy to verify that there exists a solution mapping $\mu'$ such that $\mu' \in \fctEvalPGD{\symPattern_k}{\symRDFgraph}{\symRDFdataset}$ and $\mu'(?v) = \symURI_n$, where $?v \in \lbrace ?s,?p,?o \rbrace$ and $\textstyle\symRDFdataset = \dataset_\symWoD\left( \lbrace \symURI_{n-1} \rbrace \right)$ with default graph $G = \fctData{d_{n-1}}$.
Then, by Definition~\ref{def:LDQLSemantics}, we also have $\mu' \in \fctEvalQWS{\symQuery_k}{\symWoD}{\lbrace \symURI_{n-1} \rbrace}$\!~and, thus, $\symURI_n \in \fctEvalPathPdW{ \peSubquery{?v}{\symQuery_k} }{\symURI_{n-1}}{\symWoD}$\!.
Since $\peSubquery{?v}{\symQuery_k}$ is a disjunct in $\symLPE^{\cMatch}_\mathsf{step}$, we also obtain $\symURI_n \in \fctEvalPathPdW{ \symLPE^{\cMatch}_\mathsf{step} }{\symURI_{n-1}}{\symWoD}$ and, thus,~$\symURI_n \in \fctEvalPathPdW{\symLPE^{\cMatch}}{\symURI_0}{\symWoD}$.

\medskip
As argued before, as a consequence of $\symURI_n \in \fctEvalPathPdW{\symLPE^{\cMatch}}{\symURI_0}{\symWoD}$~(and $\symURI_0 \in \symSeedURIs$), we can show that \doc\ $d_\mathsf{R} = \fctADoc{\symURI_n}$ is in $\symDocs_\mathsf{LPE}^{\cMatch}$ by using $\symURI_n$ as $\symURI$ and $\symURI_0$ as $\symCtxURI$ in the definition of $\symDocs_\mathsf{LPE}^{\cMatch}$~(cf.~Lemma~\ref{lem:LDQLCoversReachSem}). Therefore, the set $\symDocs_\mathsf{LPE}^{\cMatch}$ is not only a subset of the set of all
	\docs\ that are ($\cMatch,\symSeedURIs,\symPattern$)-reach\-able
in $\symWoD$~(as shown before), but also a superset. Hence, both sets are equivalent and, thus, Lemma~\ref{lem:LDQLCoversReachSem} holds for $\cMatch$.

\subsection{Proof of Theorem~\ref{thm:LDQLMoreThanReachSem}} \label{proof:thm:LDQLMoreThanReachSem}

In the proof we use the following simple LDQL query $Q(?x)$ given by
\[
\big\langle \tuple{\peContextURI,p,\peWildcard},(?x, ?x, ?x)\big\rangle.
\]

We prove first that the reachability criterion $\cNone$ cannot express $Q(?x)$.
On the contrary, assume that there exists a SPARQL pattern $P$ such that
\[
\fctEvalReachPcSW{P}{\cNone}{S}{\symWoD}\ = \fctEvalQWS{ Q(?x) }{\symWoD}{S}
\]
for every $S$ and $\symWoD$.
Let $u,u',a,b$ be different elements in $\symAllURIs$ that are not mentioned in $P$.
Consider now $\symWoD_1$ having only two documents $d_1=\{(u,p,u')\}$ and $d_2=\{(a,a,a)\}$ and such that
$\fctADoc{u}=d_1$ and $\fctADoc{u'}=d_2$.
Moreover, consider $\symWoD_2$ having also two documents $d_1=\{(u,p,u')\}$ and $d_3=\{(b,b,b)\}$
such that $\fctADoc{u}=d_1$ and $\fctADoc{u'}=d_3$.
First notice that 
\[
\fctEvalQW{Q(?x)}{\symWoD_1}^{\{u\}}=\{\{?x\to a\}\}\;\; \neq \;\;
\fctEvalQW{Q(?x)}{\symWoD_2}^{\{u\}}=\{\{?x\to b\}\}
\]
It is easy to see that $\fctEvalReachPcSW{P}{\cNone}{\{u\}}{\symWoD_1}=\fctEvalReachPcSW{P}{\cNone}{\{u\}}{\symWoD_2}$. Just notice that from $\{u\}$,
the set of reachable documents following the $\cNone$ criterium
is the same set $\{d_1\}$ in both $\symWoD_1$ and $\symWoD_2$. 
Thus we have that $\fctEvalReachPcSW{P}{\cNone}{\{u\}}{\symWoD_1}=\fctEvalReachPcSW{P}{\cNone}{\{u\}}{\symWoD_2}$ but $\fctEvalQW{Q(?x)}{\symWoD_1}^{\{u\}} \neq \fctEvalQW{Q(?x)}{\symWoD_2}^{\{u\}}$ which is a contradiction.

To continue with the proof, we now show that the reachability criterion $\cAll$ cannot express $Q(?x)$. To obtain a contradiction, 
assume that there exists a pattern $P$ such that 
\[
\fctEvalReachPcSW{P}{\cAll}{S}{\symWoD}\ = \fctEvalQWS{ Q(?x) }{\symWoD}{S}
\]
for every $S$ and $\symWoD$.
Let $u,u',a,b$ be different elements in $\symAllURIs$ that are not mentioned in $P$.
Consider now $\symWoD_1=(\{d_1,d_2,d_3\},\fctsymADoc_1)$ having three documents $d_1=\{(u,p,u')\}$, $d_2=\{(a,a,a)\}$ and $d_3=\{(b,b,b)\}$ and such that $\fctsymADoc_1(u)=d_1$, $\fctsymADoc_1(u')=d_2$ and 
$\fctsymADoc_1(a)=d_3$.
Moreover, consider $\symWoD_2=(\{d_1,d_2,d_3\},\fctsymADoc_2)$ having exactly the same documents as $\symWoD_1$, and such that $\fctsymADoc_2(u)=d_1$, $\fctsymADoc_2(u')=d_3$ and $\fctsymADoc_2(b)=d_2$.
First notice that 
\[
\fctEvalQW{Q(?x)}{\symWoD_1}^{\{u\}}=\{\{?x\to a\}\}\;\; \neq \;\;
\fctEvalQW{Q(?x)}{\symWoD_2}^{\{u\}}=\{\{?x\to b\}\}.
\]
Now notice that from $\{u\}$, the set of reachable documents in $\symWoD_1$ 
following the $\cAll$ criterium is the set $\{d_1,d_2,d_3\}$; $d_1$ is the document associated to $u$, $d_2$ is reachable from $d_1$ via the URI $u'$, and $d_3$ is reachable from $d_2$
via the URI $a$.
Moreover, the set reachable documents from $\{u\}$ in $\symWoD_2$ is also $\{d_1,d_2,d_3\}$;
$d_1$ is the document associated to $u$, $d_3$ is reachable from $d_1$ via the URI $u'$, and $d_2$ is reachable from $d_3$ via URI $b$.
Given that the set of reachable documents is the same in both $\symWoD_1$ and $\symWoD_2$ we have $\fctEvalReachPcSW{P}{\cAll}{\{u\}}{\symWoD_1}=\fctEvalReachPcSW{P}{\cAll}{\{u\}}{\symWoD_2}$.
Given that $\fctEvalQW{Q(?x)}{\symWoD_1}^{\{u\}} \neq \fctEvalQW{Q(?x)}{\symWoD_2}^{\{u\}}$ we obtain our desired contradiction.

We consider now the case of $\cMatch$, and prove that it cannot express $Q(?x)$.
To obtain a contradiction, 
assume that there exists a pattern $P$ such that 
\[
\fctEvalReachPcSW{P}{\cMatch}{S}{\symWoD}\ = \fctEvalQWS{ Q(?x) }{\symWoD}{S}
\]
for every $S$ and $\symWoD$.
Let $u,u',u'',a$ be different elements in $\symAllURIs$ that are not mentioned in $P$.
Consider now $\symWoD_1$ having two documents 
$d_1=\{(u,p,u')\}$ and $d_2=\{(a,a,a)\}$ and such that
$\fctADoc{u}=d_1$ and $\fctADoc{u'}=d_2$.
Moreover, consider $\symWoD_2$ having also two documents $d_1'=\{(u'',p,u')\}$ and $d_2'=\{(a,a,a)\}$
such that $\fctADoc{u}=d_1'$ and $\fctADoc{u'}=d_2'$.
First notice that 
\[
\fctEvalQW{Q(?x)}{\symWoD_1}^{\{u\}}=\{\{?x\to a\}\}\;\; \neq \;\;
\fctEvalQW{Q(?x)}{\symWoD_2}^{\{u\}}=\emptyset.
\]
We prove now that $\fctEvalReachPcSW{P}{\cMatch}{\{u\}}{\symWoD_1}=\fctEvalReachPcSW{P}{\cMatch}{\{u\}}{\symWoD_2}$.
Now given that $d_1$ is the document associated to $u$ in $\symWoD_1$, we have that
$d_1$ is ($\cMatch,\{u\},P$)-reach\-able in $\symWoD_1$.
Similarly, we know that $d_1'$ is ($\cMatch,\{u\},P$)-reach\-able in $\symWoD_2$.
Moreover, given that $P$ does not mention $u$, $u'$ and $u''$ we have that
$(u,p,u')$ matches a triple pattern in $P$ if and only if $(u'',p,u')$ matches
a triple pattern in $P$. Thus we have that
$d_2$ is ($\cMatch,\{u\},P$)-reach\-able in $\symWoD_1$ if and only if
$d_2'$ is ($\cMatch,\{u\},P$)-reach\-able in $\symWoD_2$.
Thus we have only two cases, either
\begin{itemize}
\item $\{d_1\}$ is the set of ($\cMatch,\{u\},P$)-reach\-able documents in $\symWoD_1$, and $\{d_1'\}$ is the set of ($\cMatch,\{u\},P$)-reach\-able documents in $\symWoD_2$, or
\item $\{d_1,d_2\}$ is the set of ($\cMatch,\{u\},P$)-reach\-able documents in $\symWoD_1$, and $\{d_1',d_2'\}$ is the set of ($\cMatch,\{u\},P$)-reach\-able documents in $\symWoD_2$.
\end{itemize}
In the first case we have that $\fctEvalReachPcSW{P}{\cMatch}{\{u\}}{\symWoD_1}$ is obtained by evaluating $P$ over graph $G_1=\{(u,p,u')\}$, and that  $\fctEvalReachPcSW{P}{\cMatch}{\{u\}}{\symWoD_2}$ is obtained by evaluating $P$ over graph $G_2=\{(u'',p,u')\}$.
Given that $P$ does not mention $u$, $u'$ and $u''$ we obtain that the evaluation
of $P$ over $G_1$ is the same as the evaluation of $P$ over $G_2$ which implies that 
$\fctEvalReachPcSW{P}{\cMatch}{\{u\}}{\symWoD_1}=\fctEvalReachPcSW{P}{\cMatch}{\{u\}}{\symWoD_2}$.
In the second case, $\fctEvalReachPcSW{P}{\cMatch}{\{u\}}{\symWoD_1}$ is obtained by evaluating $P$ over graph $G_1=\{(u,p,u'),(a,a,a)\}$, and $\fctEvalReachPcSW{P}{\cMatch}{\{u\}}{\symWoD_2}$ is obtained by evaluating $P$ over graph $G_2=\{(u'',p,u'),(a,a,a)\}$. For the same reason as above we have that the evaluation of $P$ is the same over $G_1$ and over $G_2$ which implies that $\fctEvalReachPcSW{P}{\cMatch}{\{u\}}{\symWoD_1}=\fctEvalReachPcSW{P}{\cMatch}{\{u\}}{\symWoD_2}$.
We have proved that $\fctEvalReachPcSW{P}{\cMatch}{\{u\}}{\symWoD_1}=\fctEvalReachPcSW{P}{\cMatch}{\{u\}}{\symWoD_2}$, while $\fctEvalQW{Q(?x)}{\symWoD_1}^{\{u\}}\neq \fctEvalQW{Q(?x)}{\symWoD_2}^{\{u\}}$ which is our desired contradiction.

%
\subsection{Proof of Proposition~\ref{prop:Safeness:FirstTrivialProperties}} \label{proof:prop:Safeness:FirstTrivialProperties}

\subsubsection{Property~\ref{prop:Safeness:FirstTrivialProperties:BasicQuery}}
Let $\symQuery$ be an arbitrary basic \query\ of the form $\tuple{\symLPE,\symPattern}$ such that $\symLPE$ is Web-safe. To show that $\symQuery$ is Web-safe we provide Algorithm~\ref{algo:ExecBasicQuery}. In line~\ref{line:ExecBasicQuery:CallExecLPE} the algorithm calls a subroutine, \textsc{ExecLPE}, that evaluates a given LPE in the context of a given URI~(cf.~Algorithm~\ref{algo:ExecLPE}).
The correctness of the algorithm and its subroutine is easily checked.
Moreover, a trivial proof by induction on the possible structure of LPEs can show that for any
Web-safe LPE, the given subroutine looks up a finite number of URIs only. The crux of such a proof is twofold: First, the evaluation of LPEs of the form $\peKleene{\symLPE}$~(lines \ref{line:ExecLPE:Kleene:Begin} to \ref{line:ExecLPE:Kleene:Return} in Algorithm~\ref{algo:ExecLPE}) is guaranteed to reach a fixed point for any \emph{finite} \wold. Second, the evaluation of LPEs of the form $\peSubquery{?v}{\symQuery}$~(lines \ref{line:ExecLPE:SubQuery:Begin} to \ref{line:ExecLPE:SubQuery:Return}) uses an algorithm for subquery $\symQuery$ that has the properties as required in Definition~\ref{def:WebSafeness}. Due to the Web-safeness of the given LPE and, thus, of $\symQuery$, such an algorithm exists.

\begin{algorithm}[t]
{
\begin{algorithmic}[1]

	\STATE $\Phi$ := a new empty set of URIs
	\FORALL {$\symURI \in \symSeedURIs$}
		\STATE $\Phi$ := $\Phi \,\cup$ \textsc{ExecLPE}($\symLPE,\symURI$) \label{line:ExecBasicQuery:CallExecLPE}
	\ENDFOR
	\STATE $\symRDFgraph$ := a new empty set of {\triple}s~(i.e., an empty RDF graph)
	\STATE $\symRDFdatasetNGs$ := a new empty set of pairs consisting of a URI and an RDF graph
	\FORALL {$\symURI \in \Phi$}
		\IF {looking up URI $\symURI$ results in retrieving a \doc, say $d$}
			\STATE $\symRDFgraph$ := $\symRDFgraph \cup \fctData{d}$
			\STATE $\symRDFdatasetNGs$ := $\symRDFdatasetNGs \cup \lbrace \tuple{\symURI,\fctData{d}} \rbrace$
		\ENDIF
	\ENDFOR

	\RETURN $\fctEvalPGD{\symPattern}{\symRDFgraph}{\tuple{\symRDFgraph,\symRDFdatasetNGs}}$
		\quad \COMMENT{$\fctEvalPGD{\symPattern}{\symRDFgraph}{\tuple{\symRDFgraph,\symRDFdatasetNGs}}$ can be computed by using any algorithm that implements \par \hspace{26mm} // the standard~(set-based) SPARQL evaluation function~\cite{Arenas09:SemanticsAndComplexityOfSPARQLBookChapter}}
\end{algorithmic}
}
	\caption{ ~ Execution of a basic \query\ $\tuple{\symLPE,\symPattern}$ using a set $\symSeedURIs$ of URIs as seed.}
	\label{algo:ExecBasicQuery}
\end{algorithm}

\begin{algorithm}[t]
{
\begin{algorithmic}[1]
	\IF {looking up URI $\symCtxURI$ results in retrieving a \doc, say $\symCtxDoc$}
		\IF {$\symLPE$ is $\peEmpty$}
			\RETURN a new singleton set $\lbrace \symCtxURI \rbrace$
\medskip

		\ELSIF {$\symLPE$ is a link pattern $\symLP = \tuple{y_1,y_2,y_3}$}
			\STATE $\symLP'$ := $\tuple{y_1',y_2',y_3'}$, where $\tuple{y_1',y_2',y_3'}$ is a link pattern generated from $\symLP$ such that any occurrence of symbol~$\peContextURI$ in~$\symLP$ is replaced by URI~$\symCtxURI$
			\STATE $\Phi$ := a new empty set of URIs
			\FORALL {$\tuple{x_1,x_2,x_3} \in \fctData{\symCtxDoc}$}
				\IF {($y_1'=x_1$ \OR $y_1'=\!\peWildcard$) \AND ($y_2'=x_2$ \OR $y_2'=\!\peWildcard$) \AND ($y_3'=x_3$ \OR $y_3'=\!\peWildcard$)}
					\FORALL {$i \in \lbrace 1,2,3 \rbrace$}
						\IF {$y_i'=\!\peWildcard$ \AND $x_i$ is a URI whose lookup retrieves a \doc}
							\STATE $\Phi$ := $\Phi \cup \lbrace x_i \rbrace$
						\ENDIF
					\ENDFOR
				\ENDIF
			\ENDFOR
			\RETURN $\Phi$
\medskip

		\ELSIF {$\symLPE$ is of the form $\peConcat{\symLPE_1}{\symLPE_2}$}
			\STATE $\Phi'$\! := \textsc{ExecLPE}($\symLPE_1,\symCtxURI$)
			\STATE $\Phi$ := a new empty set of URIs
			\STATE \textbf{for all} $\symURI'\! \in \Phi'$ \textbf{do} $\Phi$ := $\Phi \,\cup$ \textsc{ExecLPE}($\symLPE_2,\symURI'$) \textbf{end for}
			\RETURN $\Phi$
\medskip

		\ELSIF {$\symLPE$ is of the form $\peAlt{\symLPE_1}{\symLPE_2}$}
			\STATE $\Phi_1$ := \textsc{ExecLPE}($\symLPE_1,\symCtxURI$)
			\STATE $\Phi_2$ := \textsc{ExecLPE}($\symLPE_2,\symCtxURI$)
			\RETURN $\Phi_1 \cup \Phi_2$
\medskip

		\ELSIF {$\symLPE$ is of the form $\peKleene{l}$} \label{line:ExecLPE:Kleene:Begin}
			\STATE $\Phi_\mathsf{cur}$ := \textsc{ExecLPE}($\peEmpty,\symCtxURI$)
			\STATE $\symLPE'$\! := $l$
			\REPEAT \label{line:ExecLPE:Kleene:RepeatBegin}
				\STATE $\Phi_\mathsf{prev}$ := $\Phi_\mathsf{cur}$
				\STATE $\Phi_\mathsf{cur}$ := $\Phi_\mathsf{cur} \,\cup$ \textsc{ExecLPE}($\symLPE'\!,\symCtxURI$)
				\STATE $\symLPE'$\! := an LPE of the form $\peConcat{\symLPE'\!}{l}$
			\UNTIL {$\Phi_\mathsf{cur} = \Phi_\mathsf{prev}$} \label{line:ExecLPE:Kleene:RepeatEnd}
			\RETURN $\Phi_\mathsf{cur}$ \label{line:ExecLPE:Kleene:Return}
\medskip

		\ELSIF {$\symLPE$ is of the form $\peTest{\symLPE'}$}
			\STATE $\Phi$ := \textsc{ExecLPE}($\symLPE'\!,\symCtxURI$)
			\STATE \textbf{if} $\Phi \neq \emptyset$ \textbf{then return} a new singleton set $\lbrace \symCtxURI \rbrace$ \textbf{else return} a new empty set \textbf{end if}
\medskip

		\ELSIF {$\symLPE$ is of the form $\peSubquery{?v}{\symQuery}$} \label{line:ExecLPE:SubQuery:Begin}
			\STATE $\Omega$ := \textsc{Exec}($\symQuery,\lbrace\symCtxURI\rbrace$)
				\hspace{2mm}\COMMENT {where \textsc{Exec} denotes an arbitrary algorithm that can be used \par
				\hspace{30mm} // to compute the $\lbrace\symCtxURI\rbrace$-based evaluation of $\symQuery$ over the queried \par
				\hspace{30mm} // \wold}
			\STATE $\Phi$ := a new empty set of URIs
			\STATE \textbf{for all} $\mu \in \Omega$ for which $?v \in \fctDom{\mu}$ and $\mu(?v) \in \symAllURIs$ \textbf{do} $\Phi$ := $\Phi \cup \lbrace \mu(?v) \rbrace$ \textbf{end for}
			\RETURN $\Phi$ \label{line:ExecLPE:SubQuery:Return}
\medskip

		\ENDIF
\medskip
	\ELSE
		\RETURN a new empty set
	\ENDIF
\end{algorithmic}
}
	\caption{ ~ \textsc{ExecLPE}($\symLPE,\symCtxURI$)}
	\label{algo:ExecLPE}
\end{algorithm}

\subsubsection{Property~\ref{prop:Safeness:FirstTrivialProperties:ProjectionAndSEED}}
First, let $\symQuery$ be an \query\ of the form $\pi_V \symQuery'$ such that subquery $\symQuery'$ is Web-safe. Due to the Web-safeness of $\symQuery'$, there exists an algorithm for $\symQuery'$ that has the properties as required in Definition~\ref{def:WebSafeness}. We may use this algorithm to construct an algorithm for $\symQuery$; that is, our algorithm for $\symQuery$ calls the algorithm for $\symQuery'$, applies the projection operator to the result, and returns the set of solution mappings resulting from this projection. Since the application of the projection operator does not involve URI lookups, the constructed algorithm for $\symQuery$ has the properties as required in Definition~\ref{def:WebSafeness}.
Second, let~$\symQuery$ be an \query\ of the form $( \OpSEED\ U\ \symQuery' )$ such that $\symQuery'$ is Web-safe. Hence, there exists an algorithm for $\symQuery'$ that has the properties as required in Definition~\ref{def:WebSafeness}. Then, showing the Web-safeness of $\symQuery$ is trivial because the algorithm for $\symQuery'$ can also be used~for~$\symQuery$.

\subsubsection{Property~\ref{prop:Safeness:FirstTrivialProperties:UNION}}
Let $\symQuery$ be an \query\ of the form $( \symQuery_1\OpUNION \ldots \OpUNION \symQuery_n )$ such that each subquery $\symQuery_i$~($1 \!\leq\! i \!\leq\! n$) is Web-safe. Hence, for each subquery $\symQuery_i$, there exists an algorithm that has the properties as required in Definition~\ref{def:WebSafeness}. Then, the Web-safe\-ness of query $\symQuery$ is easily shown by specifying another algorithm that calls the algorithms of the subqueries sequentially and unions their results.

\subsection{Proof of Lemma~\ref{lem:StrongBoundedness}} \label{proof:lem:StrongBoundedness}

Lemma~\ref{lem:StrongBoundedness} follows from Definition~\ref{def:StrongBoundedness} and
	Buil-Aranda et al.'s result~\cite[Proposition~1]{BuilAranda11:SemanticsAndOptimizationOfSPARQLFedExt}.

\subsection{Proof of Theorem~\ref{thm:Safeness:AND}} \label{proof:thm:Safeness:AND}

We prove Theorem~\ref{thm:Safeness:AND} based on Algorithm~\ref{algo:ExecUnionFreeQuery}, which is 	an iterative algorithm that generalizes the execution strategy outlined for query $\symQuery_\mathsf{ex}''$ in Example~\ref{ex:WebSafeness}. That is, the algorithm executes the subqueries $\symQuery_1, \symQuery_2, \ldots , \symQuery_m$ sequentially in the order $\prec$ such that each iteration step~(lines \ref{line:ExecUnionFreeQuery:MainForBegin} to \ref{line:ExecUnionFreeQuery:MainForEnd}) executes one of the subqueries by using the solution mappings computed during the previous step~(which are passed on via the sets $\Omega_0,\Omega_1, \ldots, \Omega_m$).

\clearpage 

\begin{algorithm}[t!]
{
\begin{algorithmic}[1]
	\REQUIRE $m \geq 1$
	\REQUIRE \query\ $\symQuery$ is given as an array $Q$ consisting of all subqueries of $\symQuery$ such that \par the order of the subqueries in this array satisfies the conditions as given in Theorem~\ref{thm:Safeness:AND}.

\medskip
	\STATE $\Omega_0$ := $\lbrace \muEmpty \rbrace$, where $\muEmpty$ is the empty solution mapping; i.e., $\fctDom{\muEmpty} = \emptyset$ \label{line:ExecUnionFreeQuery:InitOmega0}
	\FOR {$j$ := $1, \ldots , m$} \label{line:ExecUnionFreeQuery:MainForBegin}
\medskip

		\STATE $\Omega_\mathsf{tmp}$ := a new empty set of solution mappings
		\STATE $\symQuery_j$ := the $j$-th subquery in array $Q$

\medskip
		\IF {$\symQuery_j$ is of the form $( \OpSEED\ ?v\ \symQuery' )$}
			\STATE $U_\mathsf{tmp}$ := a new empty set of URIs \label{line:ExecUnionFreeQuery:InitUtmp}
			\FORALL {$\mu \in \Omega_{j-1}$}
				\STATE \textbf{if} $\mu(?v)$ is a URI \textbf{then} $U_\mathsf{tmp}$ := $U_\mathsf{tmp} \cup \lbrace \mu(?v) \rbrace$ \textbf{end if}
			\ENDFOR \label{line:ExecUnionFreeQuery:PopulateUtmpEnd}
			\FORALL {$\symURI \in U_\mathsf{tmp}$} \label{line:ExecUnionFreeQuery:UseUtmpBegin}
				\STATE $\Omega_\mathsf{tmp}$ := $\Omega_\mathsf{tmp} \,\cup$ \textsc{Exec}($\symQuery'$,$\lbrace\symURI\rbrace$)
				\hspace{0.5mm} \COMMENT{where \textsc{Exec} denotes an arbitrary algorithm that \par
				\hspace{40mm} // can be used to compute the $\lbrace\symURI\rbrace$-based evaluation \par
				\hspace{40mm} // of $\symQuery'$ over the queried \wold}
			\ENDFOR \label{line:ExecUnionFreeQuery:UseUtmpEnd}

		\ELSE
			\STATE $\Omega_\mathsf{tmp}$ := \textsc{Exec}($\symQuery_j,\symSeedURIs$)
				\hspace{2mm} \COMMENT{where \textsc{Exec} denotes an arbitrary algorithm that can \par
				\hspace{30mm} // be used to compute the $\symSeedURIs$-based evaluation of $\symQuery_j$ \par
				\hspace{30mm} // over the queried \wold}

 \label{line:ExecUnionFreeQuery:ElseStatement}

		\ENDIF
\medskip

		\STATE $\Omega_j$ := a new empty set of solution mappings \label{line:ExecUnionFreeQuery:InitOmegaJ}
		\FORALL {$\mu \in \Omega_{j-1}$} \label{line:ExecUnionFreeQuery:JoinLoopBegin}
			\FORALL {$\mu' \in \Omega_\mathsf{tmp}$}
				\IF {$\mu$ and $\mu'$ are compatible}
					\STATE $\Omega_j$ := $\Omega_j \cup \lbrace \mu_\mathsf{join} \rbrace$, where $\mu_\mathsf{join} = \mu \cup \mu'$ \label{line:ExecUnionFreeQuery:Join}
				\ENDIF
			\ENDFOR
		\ENDFOR \label{line:ExecUnionFreeQuery:JoinLoopEnd}

\medskip
	\ENDFOR \label{line:ExecUnionFreeQuery:MainForEnd}
	\RETURN $\Omega_m$ \label{line:ExecUnionFreeQuery:Return}
\end{algorithmic}
}
	\caption{ ~ Execution of an \query\ $\symQuery$ of the form $( \symQuery_1 \OpAND \symQuery_2 \OpAND \ldots \OpAND \symQuery_m )$ using a finite set $\symSeedURIs$ of URIs as seed.}
	\label{algo:ExecUnionFreeQuery}
\end{algorithm}

To prove that Algorithm~\ref{algo:ExecUnionFreeQuery} has the properties as required in Definition~\ref{def:WebSafeness} we have to show that the algorithm is sound and complete~(i.e., for any finite set $\symSeedURIs$ of URIs and any \wold\ $\symWoD$\!, the algorithm returns $\fctEvalQWS{\symQuery}{\symWoD}{\symSeedURIs}$) and that it is guaranteed to look up a finite number of URIs only. We show these properties by induction on the $m$ iteration steps performed by the algorithm. To this end, we assume that the indices as used for the subqueries $\symQuery_1, \symQuery_2, \ldots , \symQuery_m$ reflect the order $\prec$, that is, subquery $\symQuery_1$ is the first according to $\prec$, subquery $\symQuery_2$ is the second, and so on.

\textbf{Base Case ($m=1$):}
By the conditions in Theorem~\ref{thm:Safeness:AND}, the first subquery~(according to~$\prec$) must be Web-safe and, thus, cannot be of the form $( \OpSEED\ ?v\ \symQuery' )$. Hence, the algorithm enters the corresponding else-branch~(line~\ref{line:ExecUnionFreeQuery:ElseStatement}). 
Due to the Web-safeness of $\symQuery_1$, there exists an algorithm for subquery $\symQuery_1$, say $A_1$, that has the properties as required in Definition~\ref{def:WebSafeness}. Algorithm~\ref{algo:ExecUnionFreeQuery} uses algorithm $A_1$ to obtain $\Omega_\mathsf{tmp} = \fctEvalQWS{\symQuery_1}{\symWoD}{\symSeedURIs}$~(where $\symWoD$ is the queried \wold), which requires only a finite number of URI lookups.
Thereafter, Algorithm~\ref{algo:ExecUnionFreeQuery} computes $\Omega_1 = \Omega_0 \Join \Omega_\mathsf{tmp}$~(lines \ref{line:ExecUnionFreeQuery:InitOmegaJ} to \ref{line:ExecUnionFreeQuery:JoinLoopEnd}) and returns $\Omega_1$~(line~\ref{line:ExecUnionFreeQuery:Return}), which does not require any more URI lookups.
Hence, for $m=1$, the algorithm looks up a finite number of URIs~(if the queried \wold\ is finite). Since $\Omega_0$ contains only the empty solution mapping $\muEmpty$~(line~\ref{line:ExecUnionFreeQuery:InitOmega0}), which is compatible with any other solution mapping, we have $\Omega_1 = \Omega_\mathsf{tmp}$ and, thus, $\Omega_1 = \fctEvalQWS{\symQuery_1}{\symWoD}{\symSeedURIs}$.

\textbf{Induction Step ($m>1$):} By induction we assume that after completing the ($m$--1)-th iteration, the algorithm has looked up a finite number of URIs only and the current intermediate result $\Omega_{m-1}$ covers the conjunction of subqueries $\symQuery_1, \symQuery_2, \ldots , \symQuery_{m-1}$; that is, $\Omega_{m-1} = \fctEvalQWS{( \symQuery_1 \OpAND \symQuery_2 \OpAND \ldots \OpAND \symQuery_{m-1} )}{\symWoD}{\symSeedURIs}$.
	We
show that the $m$-th iteration also looks up a finite number of URIs only and that $\Omega_m = \fctEvalQWS{( \symQuery_1 \OpAND \symQuery_2 \OpAND \ldots \OpAND \symQuery_m )}{\symWoD}{\symSeedURIs}$.

If subquery $\symQuery_m$ is Web-safe, it is not difficult to see these properties:
Since $\symQuery_m$ is Web-safe, there exists an algorithm for $\symQuery_m$, say $A_m$, that has the properties as required in Definition~\ref{def:WebSafeness}.
The corresponding call of algorithm $A_m$ in line~\ref{line:ExecUnionFreeQuery:ElseStatement} of Algorithm~\ref{algo:ExecUnionFreeQuery} looks up a finite number of URIs only, and the subsequent join computation in lines \ref{line:ExecUnionFreeQuery:InitOmegaJ} to \ref{line:ExecUnionFreeQuery:JoinLoopEnd} does not require any more lookups. Moreover, the result of calling algorithm $A_m$ in line~\ref{line:ExecUnionFreeQuery:ElseStatement} is $\Omega_\mathsf{tmp} = \fctEvalQWS{\symQuery_m}{\symWoD}{\symSeedURIs}$ and, since the subsequent join computation returns $\Omega_m = \Omega_{m-1} \Join \Omega_\mathsf{tmp}$, we have $\Omega_m = \fctEvalQWS{( \symQuery_1 \OpAND \symQuery_2 \OpAND \ldots \OpAND \symQuery_m )}{\symWoD}{\symSeedURIs}$, as~desired.

It remains to discuss the case of subquery $\symQuery_m$ being of the form $( \OpSEED\ ?v\ \symQuery' )$, where, by the conditions in Theorem~\ref{thm:Safeness:AND}, subquery $\symQuery'$ is Web-safe. Hence, there exists an algorithm for $\symQuery'$\!, say $A'$\!, that has the properties as required in Definition~\ref{def:WebSafeness}.
In this case, Algorithm~\ref{algo:ExecUnionFreeQuery} first iterates over all solution mappings in $\Omega_{m-1}$ to populate a set $U_\mathsf{tmp}$ with all URIs that any of these mappings binds to variable $?v$~(lines~\ref{line:ExecUnionFreeQuery:InitUtmp} to \ref{line:ExecUnionFreeQuery:PopulateUtmpEnd}).
Due to the finiteness assumed for all queried \wolds~(cf.~Definition~\ref{def:WebSafeness}), $\Omega_{m-1}$ is finite. Hence, the resulting set $U_\mathsf{tmp}$ contains a finite number of URIs. Therefore, the subsequent loop in lines \ref{line:ExecUnionFreeQuery:UseUtmpBegin} to \ref{line:ExecUnionFreeQuery:UseUtmpEnd} calls algorithm $A'$ a finite number of times and, thus, the $m$-th iteration looks up a finite number of URIs only.
To show the remaining claim, $\Omega_m = \fctEvalQWS{( \symQuery_1 \OpAND \symQuery_2 \OpAND \ldots \OpAND \symQuery_m )}{\symWoD}{\symSeedURIs}$, we first show $\Omega_m \subseteq \fctEvalQWS{( \symQuery_1 \OpAND \symQuery_2 \OpAND \ldots \OpAND \symQuery_m )}{\symWoD}{\symSeedURIs}$. Let $\mu_\mathsf{join}$ be an arbitrary solution mapping in $\Omega_m$; i.e., $\mu_\mathsf{join} \in \Omega_m$. By lines \ref{line:ExecUnionFreeQuery:JoinLoopBegin} to \ref{line:ExecUnionFreeQuery:JoinLoopEnd}, there exist solution mappings $\mu$ and $\mu'$ such that
\enumA~$\mu \in \Omega_{m-1} = \fctEvalQWS{( \symQuery_1 \OpAND \symQuery_2 \OpAND \ldots \OpAND \symQuery_{m-1} )}{\symWoD}{\symSeedURIs}$,
\enumB~$\mu' \in \Omega_\mathsf{tmp} = \bigcup_{\symURI \in U_\mathsf{tmp}} \fctEvalQWS{ \symQuery' }{\symWoD}{\lbrace\symURI\rbrace}$,
\enumC~$\mu$ and $\mu'$ are compatible,
and \enumD~$\mu_\mathsf{join} = \mu \cup \mu'$\!.
Then, by Definition~\ref{def:LDQLSemantics}, we have $\mu_\mathsf{join} \!\in\! \fctEvalQWS{( \symQuery_1 \OpAND \symQuery_2 \OpAND \ldots \OpAND \symQuery_m )}{\symWoD}{\symSeedURIs}$ and, thus, $\Omega_m \!\subseteq\! \fctEvalQWS{( \symQuery_1 \OpAND \symQuery_2 \OpAND \ldots \OpAND \symQuery_m )}{\symWoD}{\symSeedURIs}$.
Finally, we show $\Omega_m \supseteq \fctEvalQWS{( \symQuery_1 \OpAND \symQuery_2 \OpAND \ldots \OpAND \symQuery_m )}{\symWoD}{\symSeedURIs}$. Assume an arbitrary solution mapping $\mu^* \in \fctEvalQWS{( \symQuery_1 \OpAND \symQuery_2 \OpAND \ldots \OpAND \symQuery_m )}{\symWoD}{\symSeedURIs}$. Then, by Definition~\ref{def:LDQLSemantics}, there exist two solution mappings $\mu^*_1$ and $\mu^*_2$ such that
\enumA~$\mu^*_1 \in \fctEvalQWS{( \symQuery_1 \OpAND \symQuery_2 \OpAND \ldots \OpAND \symQuery_{m-1} )}{\symWoD}{\symSeedURIs}$,
\enumB~$\mu^*_2 \in \fctEvalQWS{ \symQuery_m }{\symWoD}{\symSeedURIs}$,
\enumC~$\mu^*_1$ and $\mu^*_2$ are compatible,
and \enumD~$\mu^* = \mu^*_1 \cup \mu^*_2$. By our induction hypothesis, we have $\mu^*_1 \in \Omega_{m-1}$. Then, given lines \ref{line:ExecUnionFreeQuery:JoinLoopBegin} to \ref{line:ExecUnionFreeQuery:JoinLoopEnd}, we have to show that $\mu^*_2 \in \Omega_\mathsf{tmp}$ where $\Omega_\mathsf{tmp}$ is the set of solution mappings computed during the $m$-th iteration.
	Since $\symQuery_m$ is of the form $( \OpSEED\ ?v\ \symQuery' )$, it holds that  
$\Omega_\mathsf{tmp} = \bigcup_{\symURI \in U_\mathsf{tmp}} \fctEvalQWS{ \symQuery' }{\symWoD}{\lbrace\symURI\rbrace}$ where $U_\mathsf{tmp} = \lbrace \symURI \in \symAllURIs \mid \mu(?v) = \symURI \text{ for some } \mu \in \fctEvalQWS{( \symQuery_1 \OpAND \symQuery_2 \OpAND \ldots \OpAND \symQuery_{m-1} )}{\symWoD}{\symSeedURIs} \rbrace$.
Hence, to show that $\mu^*_2 \in \Omega_\mathsf{tmp}$ we show that there exists a URI $\symURI \in U_\mathsf{tmp}$ such that $\mu^*_2$ is in $\fctEvalQWS{ \symQuery' }{\symWoD}{\lbrace\symURI\rbrace}$\!.
Since $\mu^*_2 \in \fctEvalQWS{ \symQuery_m }{\symWoD}{\symSeedURIs}$, by Definition~\ref{def:LDQLSemantics}, solution mapping $\mu^*_2$ binds variable $?v$ to a URI, say $\symURI^*$; i.e., $?v \in \fctDom{\mu^*_2}$ and $\mu^*_2(?v) = \symURI^*$ with $\symURI^*\! \in \symAllURIs$.
Furthermore, by Lemma~\ref{lem:StrongBoundedness} and the
condition in Theorem~\ref{thm:Safeness:AND}~(i.e., $?v \in \bigcup_{\symQuery_k \prec \symQuery_m} \fctSBVars{\symQuery_k}$), solution mapping $\mu^*_1$ also has a binding for variable $?v$, and, since $\mu^*_1$ and $\mu^*_2$ are compatible, these bindings are the same, that is, $\mu^*_1(?v) = \mu^*_2(?v)$. Hence, for URI $\symURI^* = \mu^*_2(?v)$ it holds that $\symURI^* \in U_\mathsf{tmp}$. Then, by Definition~\ref{def:LDQLSemantics}, we obtain that $\mu^*_2 \in \fctEvalQWS{ \symQuery' }{\symWoD}{\lbrace\symURI\rbrace}$\!, which shows that $\mu^*_2 \in \Omega_\mathsf{tmp}$ and, thus, we can conclude that $\Omega_m \supseteq \fctEvalQWS{( \symQuery_1 \OpAND \symQuery_2 \OpAND \ldots \OpAND \symQuery_m )}{\symWoD}{\symSeedURIs}$.

\subsection{Proof of Corollary~\ref{cor:NormalForm}} \label{proof:cor:NormalForm}

Corollary~\ref{cor:NormalForm} is an immediate consequence of Lemma~\ref{lem:Equivalences:Distributiveness}.
}

\end{document}